%% file: MS_Thesis.tex
\documentclass[british,english,MS]{iitmdiss}
\usepackage[latin9]{inputenc}
\usepackage{babel}
\usepackage{array}
\usepackage{verbatim}
\usepackage{float}
\usepackage{multirow}
\usepackage{amsmath}
\usepackage{amsthm}
\usepackage{amssymb}
\usepackage[pdftex]{graphicx}
\usepackage{setspace}
\usepackage[unicode=true,
 bookmarks=true,bookmarksnumbered=false,bookmarksopen=false,
 breaklinks=false,pdfborder={0 0 0},pdfborderstyle={},backref=false,colorlinks=false]
 {hyperref}
\hypersetup{pdftitle={A QR Decomposition Approach to Factor Modeling of Multivariate Time Series},
 pdfauthor={Immanuel David Rajan M},
 pdfsubject={Multivariate Analysis},
 pdfkeywords={Factor Analysis, QR Decomposition, Stationarity, RRQR Decomposition, ICA Preprocessing}}

\makeatletter

\newcommand{\noun}[1]{\textsc{#1}}
\providecommand{\tabularnewline}{\\}
\floatstyle{ruled}
\newfloat{algorithm}{tbp}{loa}[chapter]
\providecommand{\algorithmname}{Algorithm}
\floatname{algorithm}{\protect\algorithmname}

\usepackage{atbeginend}
\usepackage{setspace}
\usepackage{t1enc}
\usepackage{times}
\usepackage{graphicx}
\usepackage{epstopdf}
\usepackage{amsmath} 
\theoremstyle{definition}
\ifx\thechapter\undefined
  \newtheorem{example}{\protect\examplename}
\else
  \newtheorem{example}{\protect\examplename}[chapter]
\fi
\newenvironment{lyxlist}[1]
	{\begin{list}{}
		{\settowidth{\labelwidth}{#1}
		 \setlength{\leftmargin}{\labelwidth}
		 \addtolength{\leftmargin}{\labelsep}
		 }}
	{\end{list}}
\theoremstyle{definition}
\ifx\thechapter\undefined
  \newtheorem{defn}{\protect\definitionname}
\else
  \newtheorem{defn}{\protect\definitionname}[chapter]
\fi
\theoremstyle{plain}
\ifx\thechapter\undefined
  \newtheorem{thm}{\protect\theoremname}
\else
  \newtheorem{thm}{\protect\theoremname}[chapter]
\fi
\theoremstyle{plain}
\ifx\thechapter\undefined
  \newtheorem{lem}{\protect\lemmaname}
\else
  \newtheorem{lem}{\protect\lemmaname}[chapter]
\fi
\theoremstyle{plain}
\newtheorem*{thm*}{\protect\theoremname}

\usepackage[left=3cm,right=2.1cm,top=2.1cm,bottom=2.1cm]{geometry}
\usepackage{fancyhdr,titlesec,atbeginend}
\department{Electrical Engineering}
\makeatother

\addto\captionsbritish{\renewcommand{\algorithmname}{Algorithm}}
\addto\captionsbritish{\renewcommand{\definitionname}{Definition}}
\addto\captionsbritish{\renewcommand{\examplename}{Example}}
\addto\captionsbritish{\renewcommand{\lemmaname}{Lemma}}
\addto\captionsbritish{\renewcommand{\theoremname}{Theorem}}
\addto\captionsenglish{\renewcommand{\definitionname}{Definition}}
\addto\captionsenglish{\renewcommand{\examplename}{Example}}
\addto\captionsenglish{\renewcommand{\lemmaname}{Lemma}}
\addto\captionsenglish{\renewcommand{\theoremname}{Theorem}}
\providecommand{\definitionname}{Definition}
\providecommand{\examplename}{Example}
\providecommand{\lemmaname}{Lemma}
\providecommand{\theoremname}{Theorem}

\begin{document}

\title{A QR Decomposition Approach to Factor Modeling of Multivariate Time
Series}

\author{Immanuel David Rajan M}

\date{April 2014}

\maketitle
\include{includes/dedication}

\include{includes/Front-Abstract}

\begin{singlespace}
\noindent \tableofcontents{}\listoftables

\noindent \listoffigures

\end{singlespace}

\include{includes/Notations}

\include{includes/Introduction}

\newpage{}

\vspace*{7.5cm}

\emph{~~~~~~~~~~~~~~~~~~~~~~~~~~~~~~~~~~~~~~~~~~~~~This
page is intentionally left blank}

\newpage{}

\include{includes/Chapter-2}

\include{includes/Previous_algorithms}\newpage{}

\vspace*{7.5cm}

\emph{~~~~~~~~~~~~~~~~~~~~~~~~~~~~~~~~~~~~~~~~~~~~~This
page is intentionally left blank}

\newpage{}

\include{includes/RRQR}\include{includes/Chapter-4}\include{includes/Chapter-5}\newpage{}

\vspace*{7.5cm}
\emph{~~~~~~~~~~~~~~~~~~~~~~~~~~~~~~~~~~~~~~~~~~~~~This
page is intentionally left blank}

\newpage{}

\include{includes/Chapter-6}\newpage{}

\vspace*{7.5cm}
\emph{~~~~~~~~~~~~~~~~~~~~~~~~~~~~~~~~~~~~~~~~~~~~~This
page is intentionally left blank}

\newpage{}

\include{includes/Chapter-7}\include{includes/Chapter-8}\newpage{}

\vspace*{7.5cm}
\emph{~~~~~~~~~~~~~~~~~~~~~~~~~~~~~~~~~~~~~~~~~~~~~This
page is intentionally left blank}

\newpage{}

\include{includes/Appendix}

\newpage{}

\vspace*{7.5cm}
\emph{~~~~~~~~~~~~~~~~~~~~~~~~~~~~~~~~~~~~~~~~~~~~~This
page is intentionally left blank}

\newpage{}

\begin{singlespace}
\bibliographystyle{plain}
\phantomsection\addcontentsline{toc}{chapter}{\bibname}\bibliography{ref2}
\end{singlespace}

\end{document}

%% file: includes/dedication.tex
\vspace*{7.5cm}

\selectlanguage{british}%
\emph{~~~~~~~~~~~~Dedicated to mom and dad for their care,}

\emph{~~~~~~~~~~~~\quad{}~ ~~~~~~~~~~~~~and
grandparents for their love.}

\vfill{}
\selectlanguage{english}%

%% file: includes/Front-Abstract.tex
\selectlanguage{british}%
\certificate\vspace{2cm}

\noindent This is to certify that the thesis entitled \textbf{A QR
Decomposition Approach to Factor Modeling of Multivariate Time Series},
submitted by \textbf{Immanuel David Rajan M}, to the Indian Institute
of Technology Madras, for the award of the degree of \textbf{Master
of Science (Research)}, is a bonafide record of the research work
carried out by him under my supervision. The contents of this thesis,
in full or in parts, have not been submitted to any other Institute
or University for the award of any degree or diploma.\vspace{3cm}

\begin{singlespace}
\noindent Dr. Bharath Bhikkaji

\noindent Professor and Research Guide,

\noindent Dept. of Electrical Engineering,

\noindent IIT Madras, Chennai - 600036.\vspace{2cm}

\noindent Place: Chennai

\noindent Date: 11-4-2014
\end{singlespace}

\acknowledgements

The fun and elegance of research was motivated primarily by my guide
and mentor Dr. Bharath Bhikkaji. His free approach towards research
enabled me to develop the curiosity needed to explore various fields
which made the prospect of doing research both fun and entertaining.
He also helped me along the path of research and pointed me in the
right directions when I needed it. It's of note to mention that he
never once did say that he ain't got the time. All in all, the best
Guide to have.

I thank Dr. Girish Ganesan, Santa Fe partners LLC, for his valuable
comments and criticisms. He also gave the finance data to be analysed
and helped in interpreting the results.

The two courses that I felt were great are ``Speech Signal Processing''
by Dr. Umesh Srinivasan and ``Applied Time Series Analysis'' by
Dr. Arun Thangirala. These two courses were entertaining and the assignments
were fun to solve.

Life here in IITM would have been bland if not for my friends (Instrumentation
and Control lab mates and my Hostel friends). Of particular mention,
Mithun for always being there and helping out in any way (annoyance
included); Pankaj (Drama Queen) for he'll kill me if I don't mention
him; Pavan, for all those long pointless but interesting talks; Jayesh,
for being that awesome crazy friend. 

Last but not the least, I'd like to thank my family for supporting
me through this endeavour. My mom, Vanathi, for her prayers and my
Dad, Dr. Manohar Jesudas, for providing inspiration and motivation.
My sister, Grace Nithia for her pestering and all other family members
for being a critic.
\selectlanguage{english}%
\begin{abstract}
An observed $K$-dimensional series $\left\{ y_{n}\right\} _{n=1}^{N}$
is expressed in terms of a lower $p$-dimensional latent series called
factors $f_{n}$ and random noise $\varepsilon_{n}$. The equation,
$y_{n}=Qf_{n}+\varepsilon_{n}$ is taken to relate the factors with
the observation. The goal is to determine the dimension of the factors,
$p$, the factor loading matrix, $Q$, and the factors $f_{n}$. Here,
it is assumed that the noise co-variance is positive definite and
allowed to be correlated with the factors. An augmented matrix, 
\[
\tilde{M}\triangleq\left[\begin{array}{cccc}
\tilde{\Sigma}_{yy}(1) & \tilde{\Sigma}_{yy}(2) & \ldots & \tilde{\Sigma}_{yy}(m)\end{array}\right]
\]
is formed using the observed sample autocovariances $\tilde{\Sigma}_{yy}(l)=\frac{1}{N-l}\sum_{n=1}^{N-l}\left(y_{n+l}-\bar{y}\right)\left(y_{n}-\bar{y}\right)^{\top}$,
$\bar{y}=\frac{1}{N}\sum_{n=1}^{N}y_{n}$. Estimating $p$ is equated
to determining the numerical rank of $\tilde{M}$. Using Rank Revealing
QR (RRQR) decomposition, a model order detection scheme is proposed
for determining the numerical rank and for estimating the loading
matrix $Q$. The rate of convergence of the estimates, as $K$ and
$N$ tends to infinity, is derived and compared with that of the existing
Eigen Value Decomposition based approach. Two applications of this
algorithm, i) The problem of extracting signals from their noisy mixtures
and ii) modelling of the S\&P index are presented.
\end{abstract}

%% file: includes/Notations.tex
\section*{Terminology}

\subsection*{Matrix Notations}
\begin{enumerate}
\item Let $a_{n}\in\mathbb{R}^{K}$ and $b_{n}\in\mathbb{R}^{K},\;n=1,2,...,N$
be two vector valued stochastic processes. The cross-covariance between
them is denoted by, 
\begin{equation}
\Sigma_{ab}(k)\triangleq\mathcal{E}\left\{ \left(A[n+k]-\mathcal{E}\{A\}\right)\left(B[n]-\mathcal{E}\{B\}\right)\right\} ,\label{eq:2.1}
\end{equation}
where $\mathcal{E}\{.\}$ is the expectation operator.
\item The sample covariance (finite time approximation) is denoted by, 
\begin{eqnarray}
\tilde{\Sigma}_{ab}(k) & \ \triangleq & \frac{1}{N-k}\sum_{n=1}^{N-k}(a_{n+k}-\bar{a})(b_{n}-\bar{b})^{\top},\label{eq:2.2}
\end{eqnarray}
where $\bar{a}$ and $\bar{b}$ are the sample means of the stochastic
series $\{a_{n}\}$ and $\{b_{n}\}$ given by $\bar{a}=\frac{1}{N}\sum_{n=1}^{N}a_{n}$,
$\bar{b}=\frac{1}{N}\sum_{n=1}^{N}b_{n}$ respectively.
\item The difference between the sample covariance and the ideal estimate
is denoted by the perturbation 
\begin{equation}
\Delta\Sigma_{ab}(k)\ \triangleq\tilde{\Sigma}_{ab}(k)-\Sigma_{ab}(k).\label{eq:2.3}
\end{equation}
\end{enumerate}

\subsection*{Order Notations}

The following notations are used for denoting the asymptotic convergence
rates:
\begin{enumerate}
\item Let $\{a_{n}\}$ and $\{b_{n}\}$ be deterministic sequences. $a_{n}=O(b_{n})$
implies, there exists a constant $M>0$ such that, $0<\lim_{n\rightarrow\infty}\,\frac{a_{n}}{b_{n}}\leq M.$
\item $a_{n}\asymp b_{n}$ implies for a given integer $n_{0}>0$ there
exists $m,\,M>0$ such that $m<\mid\frac{a_{n}}{b_{n}}\mid<M$, for
all $n>n_{o}$. 
\item $a_{n}=o(b_{n})$ implies, $\lim_{n\rightarrow\infty}\,\frac{a_{n}}{b_{n}}=0$.
\item A stochastic sequence $\{a_{n}\}$ is said to be of smaller order
($o$) in probability to the deterministic sequence $\{b_{n}\}$,
denoted by $a_{n}=o_{P}(b_{n})$, if $\frac{a_{n}}{b_{n}}\;\rightarrow_{P}\;0$
where $\rightarrow_{P}$ denotes convergence in probability. That,
given $\epsilon>0$ there exists an $\delta>0$ and an integer $n_{0}>0$
such that $P(\mid\frac{a_{n}}{b_{n}}\mid<\epsilon)>1-\delta,\;\forall n>n_{0}$.
In other words $P(\mid\frac{a_{n}}{b_{n}}\mid<\epsilon)\to1$ as $n\to\infty$.
\item Similarly with $\{a_{n}\}$, a stochastic sequence, and $\{b_{n}\}$,
a deterministic sequence, $a_{n}=O_{P}(b_{n})$ implies, given $\epsilon>0$
there exists a constant $M>0$ such that $P(\mid\frac{a_{n}}{b_{n}}\mid<M)\geq1-\epsilon,\;\forall n\geq n_{0}$
\item $a_{n}\asymp_{P}b_{n}$ implies, given $\epsilon>0$ there exists
constants $m,\,M>0$ such that $P(m<\mid\frac{a_{n}}{b_{n}}\mid<M)\geq1-\epsilon,\;\forall n\geq n_{0}$
\end{enumerate}
For more details on the order notations refer (\cite[pp 53-55]{lehmann1999elements})
and also Appendix 2.

\subsection*{Miscelanious Notations}
\begin{enumerate}
\item The MATLAB notations are used to address various blocks of the matrix:

\begin{enumerate}
\item $A(:,\,j)$ refers to the $j^{th}$ column of $A$
\item $A(j,\,:)$ refers to the $j^{th}$ row of $A$
\item $A(a:b,\,c:d)$ refers to the block of $A$ consisting of rows $a$
to $b$ and columns $c$ to $d$
\end{enumerate}
\item The $\sigma_{i}\left(A\right)$ notation is used to indicate the $i^{th}$
highest singular value of the matrix $A\in\mathbb{R}^{m\times n}$.
\item $\lambda_{i}\left(A\right)$ reperesents the $i^{th}$ largest eigen
value of the matrix $A$. 
\end{enumerate}

%% file: includes/Introduction.tex
\selectlanguage{british}%
\pagenumbering{arabic}

\chapter{\label{chap:Introduction}Introduction}

Factor modeling refers to modeling observations in terms of its constituent
factors. In Multivariate Statistics, factor modeling refers to modeling
a given high dimensional series called ``observations'' in terms
of lower dimensional time series called ``factors''. A linear relationship\foreignlanguage{english}{
\begin{equation}
y_{n}=Hx_{n}+\varepsilon_{n},\;n=1,2,\ldots,\ N,\label{eq:1.2}
\end{equation}
}is assumed between the observations and the factors. Here $y_{n}\in\mathbb{R}^{K}$\foreignlanguage{english}{
is the observed data, $x_{n}\in\mathbb{R}^{p}$ the vector of factors,
$H\in\mathbb{R}^{K\times p}$ the factor loading matrix and $\varepsilon_{n}\in\mathbb{R}^{K\times1}$
being the random noise. Without loss of generality, it could be assumed
that $y_{n}$ and $\varepsilon_{n}$ are of zero mean.}

Factor models were first introduced by Spearman in 1904, \cite{spearman1904proof},
\foreignlanguage{english}{wherein the hidden factor ``ability''
of a set of students were estimated from the ``points'' obtained
in the examinations (quizzes) conducted}. Consider the following example: 
\begin{example}
The scores of students in three different exams are as follows:
\selectlanguage{english}%
\noindent \begin{flushleft}
\begin{tabular}{|c|c|c|c|}
\hline 
 & \textbf{Classics, $y_{1}$} & \textbf{French, $y_{2}$} & \textbf{English, $y_{3}$ }\tabularnewline
\hline 
\hline 
\textbf{Student 1} & 10 & 20 & 15\tabularnewline
\hline 
\textbf{Student 2} & 20 & 40 & 30\tabularnewline
\hline 
\textbf{Student 3} & 15 & 30 & 22.5\tabularnewline
\hline 
\end{tabular}
\par\end{flushleft}
In matrix form it could be written as, 
\[
Y=\left[\begin{array}{ccc}
10 & 20 & 15\\
20 & 40 & 30\\
15 & 30 & 22.5
\end{array}\right].
\]
This matrix is rank one and could be written as $Y=HF$ where $H=\left[\begin{array}{ccc}
1 & 2 & 1.5\end{array}\right]^{\top}$ and $F=\left[\begin{array}{ccc}
10 & 20 & 15\end{array}\right]$. Normally, the students scores would not be exact as above, and all
the smaller deviations from this single factor model will be bundled
up into the matrix $E$. Leading to the more general model, 
\[
Y=HF+E.
\]
\end{example}
Many scenarios could be modelled using \eqref{eq:1.2}. For example,
the cock-tail party problem where $p$ number of speakers, present
in a room, are recorded by $K$ number of microphones, with $K>p$.
Here, each microphone records a mix of every person's speech. The
vector of observations from the microphones, $y_{n}\in\mathbb{R}^{K}$,
could be expressed as a linear combination of the speech signals,
$x_{n}\in\mathbb{R}^{p}$. In the stock market, the daily returns,
$y_{n}$, of a collection of stocks, referred to as a portfolio, can
be interpreted as a linear combination of a few hidden factors that
affect all the stocks. 

\selectlanguage{english}%
Also, in the context of signal processing, factor models have been
used in sensor array processing \cite{Wax,sidiropoulos2000parallel},
spectrum analysis \cite{Stoica} and time series analysis \cite{anderson1963use}.

The problem chosen in this thesis is to estimate $H$ and $\left\{ x_{1}\ldots x_{N}\right\} $
given $\left\{ y_{1}\ldots y_{N}\right\} $ in \eqref{eq:1.2}. Note
the pair $(H,\;x_{n})$ cannot be uniquely determined as it can be
replaced by $(HA,\;A^{-1}x_{n})$ for any non-singular $A$. However,
as $\mathcal{R}(H)=\mathcal{R}(HA)$, where $\mathcal{R}(.)$ denotes
the range, the space spanned by the columns of $H$ is unique. Hence,
without loss of generality, assuming a QR decomposition for $H$,
(\ref{eq:1.2}) can be replaced by 
\begin{equation}
y_{n}=QRx_{n}+\varepsilon_{n}=Qf_{n}+\varepsilon_{n},\label{eq:1.3}
\end{equation}
where 
\begin{equation}
f_{n}=Rx_{n}\label{eq:f_n_def}
\end{equation}
and $Q$ is a ${K\times p}$ matrix with orthonormal columns. The
goal here is to find a $Q$, such that its columns are orthonormal
and span $\mathcal{R}(H)$. Using $Q$, an estimate of the factors
$f_{n}$can be obtained by setting, $f_{n}=Q^{\top}y_{n}$.

\section{Outline of the Proposed approach }

In this thesis, a QR decomposition based approach is used to estimate
$p,\,Q$ and $f_{n}$. First a matrix 
\begin{eqnarray}
M & \triangleq & \left[\begin{array}{cccc}
\Sigma_{yy}(1) & \Sigma_{yy}(2) & \ldots & \Sigma_{yy}(m)\end{array}\right]\label{eq:M_def}
\end{eqnarray}
is constructed. Here, $\Sigma_{yy}(l)$ are as defined in \eqref{eq:2.1}.
Note, due to \eqref{eq:1.3}, 
\begin{eqnarray}
\Sigma_{yy}(l) & = & \mathcal{E}\left\{ y_{n+l}y_{n}^{\top}\right\} \nonumber \\
 & = & Q\Sigma_{ff}\left(l\right)Q^{\top}+Q\Sigma_{f\varepsilon}(l)+\Sigma_{\varepsilon f}(l)Q^{\top}+\Sigma_{\varepsilon\varepsilon}(l).\label{eq:Introduction_1}
\end{eqnarray}
Assuming that the factors are not correlated with the future noise,
\emph{i.e.,}$\Sigma_{\varepsilon f}(l)=\mathcal{E}\left\{ \varepsilon_{n+l}f_{n}^{\top}\right\} =0$,
and the noise not having correlations across time, \emph{i.e.,} $\Sigma_{\varepsilon\varepsilon}(l)=\mathcal{E}\left\{ \varepsilon_{n+l}\varepsilon_{n}^{\top}\right\} =0$
leads to,
\[
\Sigma_{yy}(l)=Q\left[\Sigma_{ff}\left(l\right)Q^{\top}+\Sigma_{f\varepsilon}(l)\right].
\]
Thus, 
\begin{eqnarray}
M & = & Q\left[\begin{array}{ccc}
\Sigma_{ff}(1)Q^{\top}+\Sigma_{f\varepsilon}(1) & \ldots & \Sigma_{ff}(m)Q^{\top}+\Sigma_{f\varepsilon}(m)\end{array}\right]\nonumber \\
 & \triangleq & QP,\label{Eqn11}
\end{eqnarray}
where 
\begin{eqnarray}
P & = & \left[\begin{array}{cccc}
P_{1} & P_{2} & \ldots & P_{m}\end{array}\right]\label{Dum1}
\end{eqnarray}
with 
\begin{eqnarray}
P_{l} & = & \Sigma_{ff}(l)Q^{\top}+\Sigma_{f\varepsilon}(l).\label{Dum2}
\end{eqnarray}
Assuming all factors have a non-zero correlation at least in one of
the lags $l=1,\ldots,m$, implies $P$ is full row rank \emph{(i.e.,
}$P$ is of rank $p$\emph{).} This further implies that $M$ is also
of rank $p$\emph{.}

Hence the QR decomposition \cite{golub2012matrix,gander1980algorithms},
\begin{equation}
M\triangleq\mathcal{Q}R,\label{eq:QR1-1}
\end{equation}
of $M$, (\ref{eq:M_def}), with $\mathcal{Q}$ being an orthonormal
matrix would have an upper triangular $R$ with last $K-p$ rows being
zeros, \textit{i. e.,} 
\[
R=\begin{array}{c}
p\\
K-p
\end{array}\left[\begin{array}{c}
R_{11}\\
0
\end{array}\right]_{K\times mK}.
\]
Thus $Q$ is obtained by setting 
\[
Q\triangleq\mathcal{Q}(:,1:p),
\]
the first $p$ columns of $\mathcal{Q}$ (in \noun{Matlab} Notations)
and an estimate of the factors is obtained by setting $\hat{f}_{n}=Q^{\top}y_{n}$.
Note, for any permutation matrix $\Pi$$\in\mathbb{R}^{mK\times mK}$
such that the first $p$ columns of $M\Pi$ are linearly independent,
would have a QR decomposition of the form
\begin{equation}
M\Pi\triangleq\mathcal{\hat{Q}}\hat{R},\label{eq:QR1-1-1}
\end{equation}
with $\hat{R}$ having the last $K-p$ rows as zeros. A different
$Q$ can be obtained by setting $Q\triangleq\hat{\mathcal{Q}}\left(:,\,1:p\right)$.
As both $\mathcal{Q}(:,1:p)$ and $\hat{\mathcal{Q}}(:,1:p)$ span
the same column space, there need not be any specific preference of
one over the other. However, it would be prudent to choose a $\Pi$
such that the $p$ columns selected by $\Pi$ have the best condition
number.

In the finite data case $M$, \eqref{eq:M_def}, is replaced by 
\begin{equation}
\tilde{M}\triangleq\left[\begin{array}{cccc}
\tilde{\Sigma}_{yy}(1) & \tilde{\Sigma}_{yy}(2) & \ldots & \tilde{\Sigma}_{yy}(m)\end{array}\right],\label{eq:M_perturbed-1}
\end{equation}
where $\tilde{\Sigma}_{yy}(l)$ is as in \eqref{eq:2.2}. Note $\mbox{rank}\left(\tilde{M}\right)\neq p$,
in fact it will be of full rank ($\mbox{rank}\left(\tilde{M}\right)=$$K$)
with probability one. The determination of $p$ and $Q$ from $\tilde{M}$
are not as trivial as before. The matrix $\tilde{M}$ could be interpreted
as a perturbed version of the matrix $M$, 
\[
\tilde{M}=M+\Delta M
\]
or in more general form, 
\[
\tilde{M}\Pi=M\Pi+\Delta M\Pi.
\]
And the problem could be restated as finding an estimate of $M\Pi$
from $\tilde{M}\Pi$. A general QR decomposition,
\begin{eqnarray}
\tilde{M}\Pi & \triangleq & \begin{array}{c}
p\qquad mK-p\\
\left[\begin{array}{cc}
\tilde{M}_{p} & \tilde{M}_{mK-p}\end{array}\right]
\end{array}\nonumber \\
 & \triangleq & \left[\begin{array}{cc}
\mathcal{\tilde{Q}}_{p} & \tilde{\mathcal{Q}}_{K-p}\end{array}\right]\left[\begin{array}{cc}
\tilde{R}_{11} & \tilde{R}_{12}\\
0 & \tilde{R}_{22}
\end{array}\right]\nonumber \\
 & = & \mbox{\ensuremath{\mathcal{\tilde{Q}}}}\tilde{R},\label{eq:rrqr1-1}
\end{eqnarray}
where $\tilde{R}_{22}\neq0$. Assuming $\mbox{rank}\left(M\right)=p$,
an estimate of $M\Pi$ can be obtained by setting, $M\Pi=\tilde{Q}(:,\,1:p)\tilde{R}_{11}$.
For a good approximation of $M\Pi$, it is desirable to have a $\Pi$
such that $\sigma_{min}\left(\tilde{R}_{11}\right)\gg\sigma_{max}\left(\tilde{R}_{22}\right)$,
where $\sigma_{min}\left(.\right)$ denotes the minimum singular value
and $\sigma_{max}\left(.\right)$ denotes the maximum singular value.
Rank Revealing QR (RRQR) algorithms are a class of algorithms that
find a permutation matrix $\Pi$ such that either, 
\[
\max_{\Pi}\left(\sigma_{min}\left(\tilde{R}_{11}\right)\right)
\]
or,
\[
\min_{\Pi}\left(\sigma_{max}\left(\tilde{R}_{22}\right)\right)
\]
or both are satisfied. Using such a $\Pi$, a good estimate of $M\Pi$
can be obtained by setting, $M\Pi=\tilde{\mathcal{Q}}_{p}\tilde{R}_{11}$.
An the estimate of the factor loading matrix and the factors can be
got by setting, $\hat{Q}=\tilde{\mathcal{Q}}_{p}$ and $\hat{f}_{n}=\hat{Q}^{\top}y_{n}$
respectively.

So far, $\mbox{rank}\left(M\right)=p$ was assumed to be known. Determination
of $\mbox{rank}\left(M\right)$ from $\tilde{M}$, is obtained by
equating it to the determination of \emph{numerical rank}, \cite{golub1976rank,stewart1984rank},
of the matrix $\tilde{M}$. In this thesis, an RRQR based approach
is used to determine numerical rank. An estimate, $\hat{p}$, of $\mbox{rank}\left(M\right)$,
is set to, 
\begin{equation}
\hat{p}=\arg\max_{i}\frac{\gamma_{i}+\epsilon}{\gamma_{i+1}+\epsilon},\label{eq:model_order}
\end{equation}
where $\gamma_{i}$ and $\gamma_{i+1}$ are the $i^{th}$ and $i+1^{th}$
diagonal values of $\tilde{R}$, \eqref{eq:rrqr1-1}, got by assuming
numerical rank to be $i$ in the Hybrid-III RRQR algorithm \cite{chandrasekaran1994rank}
with $\epsilon=\frac{\gamma_{1}}{\sqrt{NK}}$.

Note, Hybrid-III RRQR algorithm satisfies the bounds in Hybrid-I algorithm
and hence could also be used in estimation of factor loading matrix
$Q$. Here, Hybrid-I algorithm is used because of its lesser numerical
complexity.

In this thesis, asymptotic analysis of the estimates, $\hat{p},\,\hat{Q}$
and $\hat{f}_{n}$ are presented. The rates of convergence of the
estimates, when the dimension of the observations $K$ and the duration
of the observations $N$ tends to infinity, under the constraint $\frac{K^{\delta}}{\sqrt{N}}=o(1)$,
are derived. It is shown that the proposed algorithm has better convergence
rates than the Eigen Value Decomposition (EVD) based approach, \cite{lam2011estimation},
under certain conditions.

\section{Chapter Outline}

In Chapter \ref{chap:Data-Model}, the regulatory conditions under
which the proposed algorithm is designed are summarized along with
a step-by-step implementation procedure for the proposed algorithm.

In Chapter \ref{chap:Algorithm}, the two existing algorithms for
factor modeling, namely the Principal Component Analysis Method \foreignlanguage{british}{\cite{Mardia,AndersonPCA,anderson1963use,Joreskog,Kailath}
and the EVD based method}, are discussed in detail.

In Chapter \ref{chap:RRQR-Decomposition}, a detailed exposition of
RRQR algorithms are presented along with derivation of bounds satisfied
by the Hybrid algorithms.

In Chapter \ref{chap:Model-order}, the problem of finding the number
of factors, $p$ from $\tilde{M}$ is discussed in the numerical rank
perspective. The asymptotic rate of convergence for the estimate of
the model order $\hat{p}$ as $K,N\rightarrow\infty$ is presented
therein.

In Chapter \ref{chap:Perturbation-Analysis}, the Convergence of $\hat{Q}$
obtained by using $\tilde{M}$ to $Q$ obtained from $M$ is studied.
The asymptotic rates are also got for the same.

In Chapter \ref{chap:Comparison-with-the}, the proposed algorithm
is compared with the existing EVD based algorithm of \cite{lam2011estimation}. 

In Chapter \ref{chap:Illustration-and-Simulations}, the numerical
results obtained by simulating the algorithm using Monte-Carlo trials
is presented and compared with other existing algorithms. There is
also the simulation of the famous cock-tail party problem wherein
the proposed algorithm is used as pre-processing for the popular Independent
Component Analysis (ICA) technique. A real life application of the
proposed algorithm to the finance market is also presented.

Finally, in Chapter \ref{chap:Conclusion-and-Discussions:}, concludes
with a discussion of the entire thesis along with the prospect of
further developments.

\section{Contributions of this thesis}

In this thesis, a new algorithm for factor modeling is proposed. While
other popular algorithms like the PCA based algorithms have \foreignlanguage{british}{quite
stringent constraints on noise $\varepsilon_{n}$ \emph{i.e., }}
\begin{eqnarray}
\mathcal{E}\left\{ f_{n}\varepsilon_{m}^{\top}\right\}  & = & 0,\;\forall\;n,\;m\label{Eqn3-1}\\
{\rm and\;\;\;\;\;}\nonumber \\
\mathcal{E}\{\varepsilon_{n}\varepsilon_{m}^{\top}\} & = & \sigma^{2}I_{n\times n}\delta\left(n-m\right),\label{Eqn4-1}
\end{eqnarray}
where $I_{n\times n}$ is an Identity matrix and $\delta(l)$ is the
Kronecker delta function, the proposed algorithm has less stringent
constraints allowing the factors to be correlated with the past noise
and also that noise can have any-covariance $\Sigma$. 

The proposed algorithm is subjected to rigorous asymptotic analysis
wherein
\begin{itemize}
\item Rate of convergence of the model order estimate $\hat{p}$, \eqref{eq:model_order},
to $p$ as $K,N\rightarrow\infty$ are derived.
\item Convergence rates of the factor loading matrix $\hat{Q}$ and factors
$\hat{f}_{n}$ as $K,N\rightarrow\infty$ are also derived. 
\end{itemize}
From these asymptotic properties, it was found that the proposed algorithm
performs better than the EVD based algorithm proposed in \cite{lam2011estimation}
in certain scenarios.

Two interesting applications of the proposed algorithm is also discussed.
\begin{itemize}
\item The Cock-tail party problem of separating multiple signals from their
linear noisy mixtures.
\item Modeling the S\&P 500 index and predicting the daily returns of individual
companies.
\end{itemize}

%% file: includes/Chapter-2.tex
\chapter{\label{chap:Data-Model}\foreignlanguage{british}{Data Model and
Summary of Algorithm}}

\selectlanguage{british}%
In this chapter, the Data Model and assumptions upon which the proposed
algorithm works are summarized in Section (\ref{sec:Assumptions}).
Section (\ref{sec:Summary-Algorithm}) gives a step-by-step implementation
procedure for the proposed algorithm.
\selectlanguage{english}%

\section{\label{sec:Assumptions}Assumptions on the Data}

\subsection*{Regulatory Conditions}

The data model assumed as (\ref{eq:1.2}). Now, the factors could
be written as 
\[
X=\left[x_{1},\,x_{2},\,\ldots,x_{N}\right]=\left[x^{(1)},\,x^{(2)},\,\ldots,\,x^{(p)}\right]^{\top}
\]
where $x_{n}\in\mathbb{R}^{p}$ represents the values taken by the
factors at the $n^{th}$ instant and $x^{(i)}\in\mathbb{R}^{N}$ represents
the $i^{th}$ factor which is a series in time $n=1,\ldots N$. Now,
assume that linear combination of all the factors could result white
noise, then, the last factor $x^{(p)}$ could be written as 
\[
x_{n}^{(p)}=e_{n}+\sum_{i\neq p}\alpha_{i}x_{n}^{(i)},\,n=1,\ldots,N
\]
where $e_{n}$ is a white noise sequence. Thus, replacing the pair
$(H,\,x_{n})$ in (\ref{eq:1.2}) with the pair $(HA^{-1},\,Ax_{n})$
with $A=\left[\begin{array}{ccccc}
1\\
 & 1\\
 &  & \ddots\\
 &  &  & 1\\
-\alpha_{1} & -\alpha_{2} & \ldots & -\alpha_{p-1} & 1
\end{array}\right]$ would result in the factors $f=Ax_{n}=\left[x^{(1)},\,x^{(2)},\,\ldots,\,e\right]^{\top}$where
$e$ is a white noise sequence. Thus, since $e$ is white noise, it
does not have correlation at any lags $l=1,\,\ldots,\,m$. Thus, the
rank of $M$ in (\ref{eq:M_def}) would be $\left(p-1\right)$ and
the last factor would not be detected by the proposed algorithm. Thus,
the following condition is assumed.
\begin{lyxlist}{00.00.0000}
\item [{$\left(A1\right)$}] No linear combination of the factors $x^{(i)}$,
results in white noise \emph{i.e.,}\\
 $\left[\begin{array}{ccc}
\Sigma_{xx}(1) & \ldots & \Sigma_{xx}(m)\end{array}\right]$ is of rank $p$ (full row rank).
\end{lyxlist}
The following three conditions could be deduced from (\ref{eq:Introduction_1}): 
\begin{lyxlist}{00.00.0000}
\item [{$(A2)$}] The noise $\varepsilon$ is not correlated with past
factors, \emph{i.e.,} $\Sigma_{\varepsilon x}\left(k\right)=cov(\varepsilon_{n+k},x_{n})=Z,\,\forall k>0$
where $Z$ is a zero matrix. 
\item [{$(A3)$}] The factors $x_{n}$can have correlation with past noise
$\varepsilon_{n}$, \emph{i.e.,} $cov(x_{n+k},\varepsilon_{n})$ need
not necessarily be $Z$ but bounded as $N,K\rightarrow\infty$. 
\item [{$(A4)$}] The auto-covariance of the noise needs to be bounded,
\emph{i.e.,} $cov(\varepsilon_{n},\varepsilon_{n})$ is bounded as
$N,K\rightarrow\infty$. 
\end{lyxlist}
The following mixing conditions, \cite[pp 67-77]{fan2003nonlinear},
ensures that as $N\rightarrow\infty$, $\tilde{M}$, (\ref{eq:M_perturbed-1}),
tends to $M$, (\ref{eq:M_def}) for a fixed $K$\emph{, i.e., }the
finite sample estimates converge to the ensemble average. 
\begin{lyxlist}{00.00.0000}
\item [{$(A5)$}] The process $\{y_{n}\}$ is a zero mean strictly stationary
process with $\psi$ mixing with mixing coefficients $\psi\{.\}$
satisfying $\sum_{n\geq1}n\psi(n)^{\frac{1}{2}}<\infty$ also $\mathcal{E}(|y_{n}|^{4})<\infty$
element wise. refer (\cite{fan2003nonlinear}) for a comprehensive
study of mixing properties. The $\psi$ mixing could also be replaced
with $\alpha$ mixing with mixing coefficient $\alpha_{n}=O\left(n^{-5}\right)$
and $\mathcal{E}(|y_{n}|^{12})<\infty$ element wise. For properties
regarding the same, refer Theorem 27.4 in (\cite{billingsley2008probability}).
Refer the Appendix 2 for more details regarding mixing properties.
\end{lyxlist}

\subsection*{Factor Strength}

A factor is considered a strong factor if it affects all the observed
outputs and a weak factor if not. To model this mathematically, the
columns of $H$ determine the strength with which the factors, $x_{n}$,
affect the observations, $y_{n}$. The following two assumptions on
the columns of $H$ quantify the impact of the factors on the observations: 
\begin{lyxlist}{00.00.0000}
\item [{$(A6)$}] The columns of the matrix $H=\left[\begin{array}{cccc}
h_{1} & h_{2} & \ldots & h_{p}\end{array}\right]$ are such that, $O\left(\left\Vert h_{i}\right\Vert _{2}^{2}\right)=O\left(K^{1-\delta}\right)$,
$i=1,2,...,p$ and $0\leq\delta<1$. 
\item [{$(A7)$}] The columns of $H$ are independent of each other, \emph{i.e.,}
\begin{equation}
O\left(\min_{\theta_{j}}\left\Vert h_{i}-\sum_{j\neq i}\theta_{j}h_{j}\right\Vert _{2}^{2}\right)=O\left(K^{1-\delta}\right)\label{eq:factor_strength_condn2}
\end{equation}
\end{lyxlist}
The above assumptions \emph{A1-A7} are quite similar to the ones made
in (\cite{lam2011estimation}).

Finally, consider the following notations used to denote the strength
of factor correlation with noise:
\begin{defn}
Define 
\begin{eqnarray}
\kappa_{max} & \triangleq & \max_{1\leq l\leq m}\left\Vert \Sigma_{f\varepsilon}(l)\right\Vert _{2}\\
\kappa_{min} & \triangleq & \min_{1\leq l\leq m}\sigma_{min}\left(\Sigma_{f\varepsilon}(l)\right)
\end{eqnarray}
where $\kappa_{max}$ and $\kappa_{min}$ correspond to a measure
of factor correlation with noise. 
\end{defn}

\section{Summary of the Algorithm\label{sec:Summary-Algorithm}}

In this section, the steps in implementation of the proposed algorithm
are presented. 
\begin{enumerate}
\item Let $\left\{ y_{n}\right\} _{n=1}^{N}$ be the observed data. And
the data model is assumed as in (\ref{eq:1.3}). 
\item Compute $\tilde{M}$ defined in (\ref{eq:M_perturbed-1}), where $\tilde{\Sigma}_{yy}(l)$
is got by (\ref{Eqn11a}). 
\item To evaluate the Numerical Rank, $p$, 

\begin{enumerate}
\item Set $i=1$
\item First the permutation matrix $\Pi$ is got by using \noun{Hybrid-III}\emph{
}algorithm\emph{ }assuming a numerical rank of $i$\emph{.}
\item Take Gram-Schmidt QR (GS-QR) decomposition of $\tilde{M}\Pi$ \emph{i.e.,}
$\tilde{M}\Pi=\tilde{Q}\tilde{R}$. Compute $r_{i}=\frac{\gamma_{i}+\epsilon}{\gamma_{i+1}+\epsilon}$
where $\gamma_{i}$ and $\gamma_{i+1}$ are the $i^{th}$ and $i+1^{th}$
diagonal values of $\tilde{R}$ and $\epsilon=\frac{\gamma_{1}}{\sqrt{KN}}$.
\item Set $i=i+1$
\item \textbf{Repeat Until $i<P$ }where, $P$ is taken such that $p<P$. 
\end{enumerate}
\item Now, estimate of numerical rank, $\hat{p}$, is got by finding the
$i$ that maximizes $r_{i}$ as in (\ref{eq:model_order}). 
\item Evaluate RRQR decomposition of $\tilde{M}$ using \noun{Hybrid-I}
algorithm (\cite{chandrasekaran1994rank}) presented in \emph{Algorithm
(\ref{alg:Hybrid-I}). }Let the RRQR Decomposition be as in (\ref{eq:rrqr1-1}). 
\item Now, the finite sample estimate of $\hat{Q}$ is got by, $\hat{Q}_{N,K}=\tilde{Q}(:,1:\hat{p})$.
\item Finally the estimate of factors is got by, $\hat{f}_{n}=\hat{Q}_{N,K}^{\top}y_{n}$.
\end{enumerate}
An exposition of RRQR-Algorithms are presented in Chapter (\ref{chap:RRQR-Decomposition}).

%% file: includes/Previous_algorithms.tex
\chapter{\label{chap:Algorithm}Existing Algorithms}

\selectlanguage{british}%
In this chapter, two existing algorithms for factor modeling are analysed.
Section \ref{sec:ML-Method} discusses the PCA based method of factor
analysis and Section \ref{sec:EVD-Method} covers the Eigen Value
Decomposition based method. 

\section{\label{sec:ML-Method}PCA based Method}

The PCA based method is dealt in many literature, notably \cite{Mardia,AndersonPCA,anderson1963use,Joreskog,Kailath}.
Assume that the observed time series $\left\{ y_{n}\in\mathbb{R}^{K}\right\} _{n=1}^{N}$
could be modelled as, 
\[
y_{n}=Qf_{n}+\varepsilon_{n}
\]
where, $f_{n}\in\mathbb{R}^{p}$ are factors and $\varepsilon_{n}\in\mathbb{R}^{K}$
is random noise. The factors and noise are assumed to satisfy:
\begin{description}
\item [{Condition\_1:}] The factors and noise are uncorrelated with each
other \emph{i.e., $\Sigma_{f\varepsilon}(l)=\mathcal{E}\left\{ f_{n+l}\varepsilon_{l}^{\top}\right\} =0,\:\forall l$.}
\item [{Condition\_2:}] The noise should have covariance given by $\Sigma_{\varepsilon\varepsilon}(l)=\mathcal{E}\left\{ \varepsilon_{n+l}\varepsilon_{l}^{\top}\right\} =\delta(l)\sigma^{2}I,$
where $\delta(l)$ is the Kronecker delta function and $I\in\mathbb{R}^{K\times K}$
is an Identity matrix. 
\end{description}
Under the above conditions, the autocovariance matrix of the observations
is 
\[
\Sigma_{yy}(0)=\mathcal{E}\left\{ y_{n}y_{n}^{\top}\right\} =Q\Sigma_{ff}(0)Q^{\top}+\sigma^{2}I.
\]
Assuming that $\Sigma_{ff}(0)$ to be full rank, implies that $Q\Sigma_{ff}(0)Q^{\top}$
is of rank $p$ and will have $p$ non-zero positive Eigen Values.
Let $\lambda_{1}\geq\lambda_{2}\geq\ldots\geq\lambda_{p}$ be the
eigen values of $Q\Sigma_{ff}(0)Q^{\top}$. The Eigen Value Decomposition
of $\Sigma_{yy}(0)$ is, 
\[
\Sigma_{yy}(0)=U\left[\begin{array}{cccccc}
\lambda_{1}+\sigma^{2}\\
 & \ddots\\
 &  & \lambda_{p}+\sigma^{2}\\
 &  &  & \sigma^{2}\\
 &  &  &  & \ddots\\
 &  &  &  &  & \sigma^{2}
\end{array}\right]U^{\top}
\]
Thus, the number of factors equal $p$ and the factor loading matrix
could be set to, $\hat{Q}=U(:,\,1:p)$, the eigen vectors that correspond
to the largest $p$ Eigen values. The factors are then $\hat{f}_{n}=\hat{Q}^{\top}y_{n}.$
Henceforth, this technique will be refered to as PCA. 

Now, the above procedure is ideal. In practice, $\Sigma_{yy}(0)$
will be replaced by its sample co-variance, $\tilde{\Sigma}_{yy}(0)=\frac{1}{N}\sum_{n=1}^{N}y_{n}y_{n}^{\top}$.
In which case, estimating the number of factors and the factor loading
matrix becomes a bit more difficult. The problem of finding the number
of factors is dealt as a \emph{model order selection} problem using
AIC (Akiake Information Criteria) and MDL (Minimum Code Length, \cite{rissanen1978modeling})
in \cite{Kailath}. In \cite{bai2002}, an alternative Information
Theoretic Criteria is proposed and proved to be better than the AIC,
\cite{Akaike_IC}, BIC (Bayesian Information Criteria) and all of
their derivatives. \foreignlanguage{english}{In \cite{bai2002}, the
model order is taken as the $p$ that minimizes the Information Content,
\[
IC_{p}=\ln\left(V(p,\hat{f}_{n})\right)+p\left(\frac{K+N}{KN}\right)\ln\left(\frac{KN}{K+N}\right)
\]
where 
\[
V(p,\hat{f}_{n})=\min_{\hat{Q}}\frac{1}{KN}\sum_{n=1}^{N}\left\Vert y_{n}-\hat{Q}^{(p)}\hat{f}_{n}^{(p)}\right\Vert _{2}^{2}
\]
Here, $\hat{Q}^{(p)}$ and $\hat{f}_{n}^{(p)}$ are obtained using
PCA by assuming model order as $p$. Thus, the estimate of model order
is 
\[
\hat{p}_{N,K}=\arg\min_{p}\{IC_{p}\}.
\]
For more details regarding the same, refer \cite{bai2002}.}

\selectlanguage{english}%
To estimate the factor loading matrix, a Maximum Likelihood framework
is used. \foreignlanguage{british}{The Maximum Likelihood framework
for factor analysis was first introduced by Lawley in the year 1940,
\cite{lawley1940estimation}. A detailed version of this ML approach
is given in, \cite{Joreskog,Stoica2,bai2002,Mardia}. }

\selectlanguage{british}%
Let, 
\[
\tilde{\Sigma}_{yy}(0)=\tilde{U}\hat{\Lambda}\tilde{U}^{\top}
\]
be the eigen value decomposition of $\tilde{\Sigma}_{yy}(0)$, where
$\hat{\Lambda}$ is a diagonal matrix of eigen values arranged in
descending order, $\hat{\lambda}_{1}\geq\hat{\lambda}_{2}\geq\ldots\geq\hat{\lambda}_{K}$.
It was proved in \cite{Stoica2}, that the optimal solution to the
ML estimate for the error co-variance, $\sigma^{2}I$, is given by
$\hat{\sigma}^{2}=\sum_{i=p+1}^{K}\hat{\lambda}_{i}$ and the ML estimate
for $Q$ is given by $\hat{Q}=\tilde{U}\left(:,\,1:p\right)$.

The disadvantages of this PCA based method are, 
\begin{enumerate}
\item The ML estimate inherently assumes $y_{n}$ is iid Gaussian (Though
in real life simulations, it could not be ensured).
\item The requirement of conditions 1 and 2 regarding the noise $\varepsilon_{n}$
and factors $f_{n}$ as stated before.
\item No properties of the factors $f_{n}$ are taken into consideration. 
\end{enumerate}

\section{\label{sec:EVD-Method}EVD Method}

Recently, Lam et. al., had proposed a new technique for factor modeling,
\cite{lam2011estimation}. This technique helps in overcoming the
disadvantages pointed out in the previous PCA based method. 

\selectlanguage{english}%
In \cite{lam2011estimation}, the authors setup a matrix 
\begin{eqnarray}
S & \triangleq & \sum_{l=1}^{m}\Sigma_{yy}(l)\Sigma_{yy}(l)^{\top}\label{Eqn10}\\
 & = & MM^{\top},\nonumber 
\end{eqnarray}
where $M$ is as defined in (\ref{eq:M_def}). From (\ref{Eqn10})
and (\ref{Eqn11}), 
\begin{eqnarray}
S & = & QPP^{\top}Q^{\top}.\label{Eqn11d}
\end{eqnarray}

As $S$ is an $K\times K$ positive semi-definite matrix with rank
$p$ (assuming the $p\times p$ matrix $PP^{\top}$ to be full rank),
its Eigen value decomposition 
\begin{eqnarray}
S & = & U\Lambda U^{\top}\label{Eqn11e}
\end{eqnarray}
would have a diagonal matrix $\Lambda$ with the first $p$ diagonal
elements strictly positive and the rest $K-p$ diagonal elements being
zero. The factor loading matrix was chosen to be $Q=U(:,1:p)$, the
first $p$ columns of $U$ corresponding to the non-zero Eigen values.
Here, $U(:,1:p)$ refers to the MATLAB notation of selecting first
$p$ columns. An estimate $\hat{f}_{n}$ of the factors was obtained
by setting $\hat{f}_{n}=Q^{\top}y_{n}$.

In practice, the auto-covariances $\Sigma_{yy}(k),\,\Sigma_{ff}(k)\,\mbox{and}\,\Sigma_{f\varepsilon}(k)\;k=0,1,2,\ldots$
are approximated by their corresponding finite sample equivalents,
for example $\Sigma_{yy}(k)$ is replaced by 
\begin{eqnarray}
\tilde{\Sigma}_{yy}(k) & = & \frac{1}{N-k}\sum_{n=1}^{N-k}\left(y_{n+k}-\bar{y}\right)\left(y_{n}-\bar{y}\right)^{\top},\label{Eqn11a}
\end{eqnarray}
where $\bar{y}\triangleq\frac{1}{N}\sum_{t=1}^{N}y_{n}$. The Eigen
values of the matrix $S$ formed using (\ref{Eqn11a}) will not have
$K-p$ zero Eigen values. In fact all the Eigen values will be different
with probability $1$. Therefore, $p$ cannot be directly determined
by computing the Eigen values of $S$. In \cite{lam2011estimation},
the model order, $p$, was estimated using 
\begin{equation}
\hat{p}=\arg\max_{1<i<R}\frac{\lambda_{i}}{\lambda_{i+1}}\label{eq:m1a}
\end{equation}
where $\lambda_{i}>\lambda_{i+1}$ denotes the Eigen values of the
matrix $S$ arranged in descending order. The properties of this ratio
based estimator (\ref{eq:m1a}) are summarized in \cite{C1}.

The Proposed QR decomposition based algorithm will be compared with
this algorithm in Chapter \ref{chap:Comparison-with-the} with regards
to asymptotic properties of their estimates and in Chapter \ref{chap:Illustration-and-Simulations}
through their simulations.

%% file: includes/RRQR.tex
\chapter{\label{chap:RRQR-Decomposition}RRQR Decomposition}

In this chapter, the RRQR algorithms used in this thesis are discussed
in detail. The derivations of the results pertaining to the RRQR algorithm
are presented in this chapter for the sake of completion.

Define a matrix $A$ of dimension $m\times n$ such that, 
\[
A\triangleq\left[\begin{array}{cccc}
a_{1} & a_{2} & \ldots & a_{n}\end{array}\right],
\]
 where $a_{i}$ is an $m$ dimensional column vector. The QR decomposition
of $A$ is given by: 
\begin{eqnarray}
A & = & QR,\label{eq:The QR Decomposition}\\
\left[\begin{array}{cccc}
a_{1} & a_{2} & \ldots & a_{n}\end{array}\right] & = & \left[\begin{array}{cccc}
q_{1} & q_{2} & \ldots & q_{n}\end{array}\right]\left[\begin{array}{cccc}
<a_{1},q_{1}> & <a_{2},q_{1}> & \ldots & <a_{n},q_{1}>\\
0 & <a_{2,}q_{2}> & \ldots & <a_{n},q_{2}>\\
0 & 0 & \ddots & \vdots\\
0 & 0 & 0 & <a_{n},q_{n}>
\end{array}\right]\nonumber 
\end{eqnarray}
where $q_{i}'s$ are the orthonormal columns of $Q$, $R$ is the
upper triangular matrix and $<\,,\,>$ denotes inner product. There
are many algorithms to achieve the above decomposition refer \cite{gander1980algorithms,golub2012matrix}
for more details.

In the above equation, the matrix $Q$ and $R$ could be segmented
as
\begin{eqnarray}
Q & = & \begin{array}{c}
p\quad m-p\\
\left[\begin{array}{cc}
Q_{1} & Q_{2}\end{array}\right]
\end{array}\nonumber \\
R & \triangleq & \begin{array}{cc}
 & \begin{array}{cc}
p\quad\; & n-p\end{array}\\
\begin{array}{c}
p\\
m-p
\end{array} & \left[\begin{array}{cc}
R_{11}\qquad & R_{12}\\
0\qquad & R_{22}
\end{array}\right]
\end{array},\label{eq:QR_decompostion_2}
\end{eqnarray}
for any $p<\min\left(n,m\right)$. From (\ref{eq:The QR Decomposition}),
it can be seen that the minimum singular value of $R_{11}$, $\sigma_{min}\left(R_{11}\right)$
is dependent on the columns $a_{1},\,\ldots,\,a_{p}$ of $A$ and
maximum singular value of $R_{22}$, $\sigma_{max}\left(R_{22}\right)$,
depends on the columns $a_{p+1},\,\ldots,\,a_{n}$ of $A$. Thus permuting
the columns of $A$ and taking QR would result in different $\sigma_{min}\left(R_{11}\right)$
and $\sigma_{max}\left(R_{22}\right)$. Let $\Pi$ be a permutation
matrix, the QR decomposition of $A\Pi$ be given by, 
\[
A\Pi=\tilde{Q}\tilde{R}.
\]
Given a $p<\min\left(m,n\right)$, RRQR algorithms are a group of
algorithms that find a $\Pi$ such that, 
\begin{equation}
\Pi\triangleq\arg\max_{\Pi\in\mathcal{P}}\sigma_{min}\left(\tilde{R}_{11}\right)\label{eq:C1}
\end{equation}
or
\begin{equation}
\Pi\triangleq\arg\min_{\Pi\in\mathcal{P}}\sigma_{max}\left(\tilde{R}_{22}\right)\label{eq:C2}
\end{equation}
or both, where $\mathcal{P}$ is a set of all permutation matrices
and $\tilde{R}_{11}$, $\tilde{R}_{22}$ are got by segmenting $\tilde{R}$
as in (\ref{eq:QR_decompostion_2}). The algorithms that solve (\ref{eq:C1}),
(\ref{eq:C2}) or both (\ref{eq:C1}) and (\ref{eq:C2}) are are termed,
Type I, Type II and Type-III algorithms respectively, \cite{chandrasekaran1994rank}.

In this thesis, the Hybrid RRQR algorithms proposed by Chandrasekaran
et. al., \cite{chandrasekaran1994rank}, are of particular interest.
The Hybrid-I algorithm is a type I algorithm that satisfies, 
\begin{eqnarray*}
\sigma_{min}\left(R_{11}\right) & \geq & \frac{\sigma_{p}\left(A\right)}{\sqrt{p(n-p+1)}},\\
\sigma_{max}\left(R_{22}\right) & \leq & \sigma_{min}\left(R_{11}\right)\sqrt{p(n-p+1)}.
\end{eqnarray*}
The Hybrid-II is a type -II algorithm that satisfies 
\begin{eqnarray*}
\sigma_{max}\left(R_{22}\right) & \leq & \sigma_{p+1}\left(A\right)\sqrt{(p+1)(n-p)}\\
\sigma_{min}\left(R_{11}\right) & \geq & \frac{\sigma_{max}\left(R_{22}\right)}{\sqrt{(p+1)(n-p)}}.
\end{eqnarray*}
Finally, the Hybrid-III algorithm combines the best of the previous
two bounds to ensure, 
\begin{eqnarray*}
\sigma_{max}\left(R_{22}\right) & \leq & \sigma_{p+1}\left(A\right)\sqrt{(p+1)(n-p)}\\
\sigma_{min}\left(R_{11}\right) & \geq & \frac{\sigma_{p}\left(A\right)}{\sqrt{p(n-p+1)}}.
\end{eqnarray*}
Derivation of these bounds can be found in \cite{chandrasekaran1994rank}.
These derivations are also presented later in this chapter for the
sake of completion.

To understand the working of the Hybrid algorithms, one needs to understand
the working of two other algorithms: the QR decomposition with Column
Pivoting (QR-CP) algorithm proposed by Golub, \cite{golub2012matrix},
which is a Type-I algorithm and its corresponding Type-II version
proposed by Stewart (Stewart-II) in \cite{stewart1984rank}.

The following theorem relates to Interlacing Property of singular
values. It is essential for understanding the rest of the thesis:
\begin{thm}
Let $\Sigma\in\mathbb{R}^{m\times n}$, $m<n$, and $\Sigma_{11}^{k\times l}$,
$k<l$, be any sub matrix of $\Sigma$ with \foreignlanguage{british}{$\sigma_{1}\geq\sigma_{2}\geq\ldots\geq\sigma_{m}$
and $\lambda_{1}\geq\lambda_{2}\geq\ldots\geq\lambda_{k}$ being their
corresponding singular values. Then,
\[
\sigma_{j}\geq\lambda_{j}\geq\sigma_{n-k+j},j=1,2...k.
\]
}
\end{thm}
\selectlanguage{british}%
\begin{proof}
Refer to Appendix 1.
\end{proof}
\selectlanguage{english}%

\subsection*{QR-CP Algorithm}

Let the matrix to be permuted be $A\in\mathbb{R}^{m\times n}$. The
algorithm is as follows:
\begin{description}
\item [{Initialize}] 

\begin{description}
\item [{$s=0$\textmd{;}}] $\Pi=I\in\mathbb{R}^{n\times n}$. 
\item [{Let}]~
\item [{$j=\arg\max_{i}\left\Vert A(:,\,i)\right\Vert _{2}.$}]~
\item [{if}] $j\neq1$, interchange $1^{st}$ and $j^{th}$ column of $\Pi$
and take QR decomposition\footnote{Householder or Gram Schmidt methods.},
\[
A\Pi=Q^{(0)}R^{(0)},
\]
where $R^{(s)}$ denotes the matrix $R$ being segmented as, 
\[
R^{(s)}\triangleq\begin{array}{cc}
 & \begin{array}{cc}
s\quad\; & n-s\end{array}\\
\begin{array}{c}
s\\
m-s
\end{array} & \left[\begin{array}{cc}
R_{11}^{(s)}\qquad & R_{12}^{(s)}\\
0\qquad & R_{22}^{(s)}
\end{array}\right]
\end{array}.
\]
\end{description}
\item [{Repeat}] 

\begin{description}
\item [{$s=s+1;$}]~
\item [{$j=\arg\max_{i}\left\Vert R_{22}^{(s)}(:,\,i)\right\Vert _{2}$}]~
\item [{if}] $j\neq1$,

\begin{description}
\item [{\textmd{Interchange}}] columns $s+1$ and $s+j$ of $\Pi$ and
retriangularize $R^{(s)}\Pi=QR^{(s+1)}$\footnote{Let $A\Pi=QR$. If $R\bar{\Pi}=\bar{Q}\bar{R}$, Then, $A\Pi\bar{\Pi}=Q\bar{Q}\bar{R}$.
Thus, finding permutations to columns of $R$ is equivalent to finding
permutations to columns of $A$.}.
\end{description}
\end{description}
\item [{Until}] $s<\min\left(n,\,m\right)$.
\end{description}
Note that $R_{11}^{(s)}$ and $R_{22}^{(s)}$ changes in dimension
with every increase in $s$

The above QR-CP algorithm is an approximation of the greedy algorithm
which selects columns such that $\sigma_{min}\left(R_{11}^{(s)}\right)$
is maximized. For further details regarding this approximation, refer
\cite[pp 596-598]{chandrasekaran1994rank}.

Now, to look at the properties of the QR-CP algorithm, let 
\[
A\Pi=\tilde{Q}\tilde{R}
\]
where $\Pi$ is obtained using QR-CP algorithm. The following two
lemmas summarize the properties of QR-CP algorithm:
\begin{lem}
The diagonal elements in $\tilde{R}$ satisfy 
\begin{equation}
\gamma_{i}\geq\frac{\sigma_{i}\left(A\right)}{\sqrt{n-i+1}},\:i=1,\ldots,\min\left(n,\,m\right).\label{eq:prop1_main}
\end{equation}
where $\gamma_{i}$ is the $i^{th}$ diagonal element of $\tilde{R}$. 
\end{lem}
\begin{proof}
From (\ref{eq:The QR Decomposition}) and the above algorithm, it
can be seen that $\gamma_{i}$ corresponds to the highest $2-norm$
among the columns of $R_{22}^{(i-1)}$. Since $R_{22}^{(i-1)}$ has
$n-i+1$ columns, 
\[
\gamma_{i}\geq\frac{\left\Vert R_{22}^{(i-1)}\right\Vert _{F}}{\sqrt{n-i+1}}.
\]
Since $\left\Vert R_{22}^{(i-1)}\right\Vert _{F}\geq\left\Vert R_{22}^{(i-1)}\right\Vert _{2}$,
\begin{equation}
\gamma_{i}\geq\frac{\left\Vert R_{22}^{(i-1)}\right\Vert _{2}}{\sqrt{n-i+1}}.\label{eq:prop1_proof1}
\end{equation}
Interlacing property\footnote{Refer Appendix 1 on page \pageref{subsec:Interlacing-Property} for
a more elegant proof on Interlacing property.} of Singular Values \cite{golub2012matrix,li2012generalized} states
that: 
\begin{equation}
\sigma_{s}\left(\tilde{R}\right)\geq\left\Vert \tilde{R}_{22}^{(s)}\right\Vert _{2}\geq\sigma_{s+1}(\tilde{R}).\label{eq:prop1_proof2}
\end{equation}
Thus, combining (\ref{eq:prop1_proof1}) and (\ref{eq:prop1_proof2})
leads to (\ref{eq:prop1_main}).
\end{proof}
\begin{lem}
When using QR-CP and segmenting $\tilde{R}$ as in (\ref{eq:QR_decompostion_2}),
\begin{equation}
\sigma_{min}\left(R_{11}\right)\geq\frac{\sigma_{p}\left(A\right)}{\sqrt{pn}2^{p}}.\label{eq:prop2}
\end{equation}
\end{lem}
\begin{proof}
Note $R_{11}$ could be written as 
\[
R_{11}=DW.
\]
where $D=\mbox{diag}\left(R_{11}\right)$ is a diagonal matrix with
diagonal elements of $R_{11}$. By the property of QR-CP algorithm,
the diagonal elements of $R_{11}$ are the highest elements. Thus,
$W$ is an upper triangular matrix such that the diagonal elements
are equal to $1$ and the rest of the elements are less than $1$.
Thus, 
\[
\sigma_{min}\left(R_{11}\right)\geq\sigma_{min}\left(D\right)\sigma_{min}\left(W\right)
\]
Since, for a diagonal matrix, the singular values equal the diagonal
elements, from (\ref{eq:prop1_main}), 
\[
\sigma_{min}\left(D\right)\geq\frac{\sigma_{p}\left(A\right)}{\sqrt{n-p+1}}.
\]
Thus, since $\sigma_{min}\left(W\right)=\frac{1}{\left\Vert W^{-1}\right\Vert _{2}},$
\[
\sigma_{min}\left(R_{11}\right)\geq\frac{\sigma_{p}\left(A\right)}{\sqrt{n-p+1}\left\Vert W^{-1}\right\Vert _{2}}.
\]
Now, $\left\Vert W^{-1}\right\Vert _{2}\leq\sqrt{p}2^{p}$ refer \cite[pp 606]{chandrasekaran1994rank},
thus, we get (\ref{eq:prop2}).
\end{proof}
The above proof is paterned after the proof provided in \cite{chandrasekaran1994rank}
for the same.

An implementation scheme for the QR-CP algorithm when applied to $\tilde{M}$,
(\ref{eq:M_perturbed-1}), is summarized in \emph{Algorithm }\ref{alg:al1}. 

\begin{algorithm}
\caption{\label{alg:al1}QR Decomposition with Column Pivoting}
~\\
~
\begin{enumerate}
\item Let the matrix to be decomposed be $\tilde{M}\in\mathbb{R}^{K,mK}$.
\item Let the permutation matrix $\Pi$ be initialized as an Identity matrix
of order $mK$.
\item Take a value of $P$ such that, $p<P<\min(K,N)$
\item Take Gram-Schmidt QR (GS-QR) decomposition of $\tilde{M}$ i.e., $\tilde{M}=Q^{(0)}R^{(0)}$,
where $R^{(s)}\triangleq\begin{array}{cc}
 & \begin{array}{cc}
s\quad\; & mK-s\end{array}\\
\begin{array}{c}
s\\
K-s
\end{array} & \left[\begin{array}{cc}
R_{11}^{(s)}\qquad & R_{12}^{(s)}\\
0\qquad & R_{22}^{(s)}
\end{array}\right]
\end{array}$ and set $s=0.$ 
\item \textbf{Repeat} \\
~

\begin{enumerate}
\item Let $e_{i}$ be the $i^{th}$ column of an identity matrix of order
$mK-s$.
\item $j\triangleq\arg\max_{1\leq i\leq mK-s}\left\Vert R_{22}^{(s)}e_{i}\right\Vert _{2}$
\item Interchange the columns $j$ and $s+1$ in $\Pi$.
\item Compute, GS-QR of $\tilde{M}\Pi$ as $\tilde{M}\Pi\triangleq Q^{(s)}R^{(s)}$
\item $s=s+1$
\end{enumerate}
\item \textbf{Until} $s>P$
\item Return $\Pi$
\end{enumerate}
\end{algorithm}

\subsection*{Stewart's Type-II Algorithm}

The Stewart's Type-II algorithm proposed first in \cite{stewart1984rank}
is discussed here. In, \cite{chandrasekaran1994rank}, this Stewart's
algorithm is framed as a type-II equivalent of the QR-CP algorithm.
In, \cite{chandrasekaran1994rank}, both Type-I and Type-II algorithms
were unified in a single framework. This unification enables derivation
of the properties of the Type-II algorithm from its Type-I equivalents.
Note Stewart-II algorithm is applicable only for Invertible matrices.

\subsubsection*{Unification}

Assume $A\in\mathbb{R}^{n\times n}$ is invertible and of numerical
rank $p$. For details regarding Numerical rank refer Chapter \ref{chap:Model-order}.
For now, just take $p$ as the block size used for segmenting $R$
in (\ref{eq:The QR Decomposition}).

Note, the Type-II formulation is equivalent to, 
\begin{eqnarray}
\min_{\Pi}\sigma_{max}\left(R_{22}\right) & = & \min_{\Pi}\frac{1}{\sigma_{min}\left(R_{22}^{-1}\right)}\nonumber \\
 & = & \frac{1}{\max_{\Pi}\left(\sigma_{min}\left(R_{22}^{-1}\right)\right)}.\label{eq:Unification1}
\end{eqnarray}
Let, 
\[
R\Pi\triangleq\bar{Q}\bar{R},\;\bar{R}\triangleq\left[\begin{array}{cc}
R_{11}\qquad & R_{12}\\
0\qquad & R_{22}
\end{array}\right]
\]
Assuming $R$ to be invertible, inverting both sides of the above
equation leads to, 
\[
\Pi^{\top}R^{-1}=\left[\begin{array}{cc}
R_{11}^{-1}\qquad & -R_{11}^{-1}R_{12}R_{22}^{-1}\\
0\qquad & R_{22}^{-1}
\end{array}\right]\bar{Q}^{\top}.
\]
Now, taking transpose on both sides, 
\begin{equation}
R^{-\top}\Pi=\bar{Q}\left[\begin{array}{cc}
R_{11}^{-\top}\qquad & 0\\
-R_{22}^{-\top}R_{12}^{\top}R_{11}^{-\top}\qquad & R_{22}^{-\top}
\end{array}\right].\label{eq:Unification2}
\end{equation}
Thus from (\ref{eq:Unification2}), (\ref{eq:Unification1}) and the
fact that $\sigma_{min}\left(R_{22}^{-1}\right)=\sigma_{min}\left(R_{22}^{-\top}\right)$;
It can be seen that applying Type-II algorithm to $R$ is equivalent
to Type-I algorithm applied to $R^{-\top}$. To make it more clearer,
applying Type-I algorithm to $R^{-\top}$ gives, 
\[
R^{-\top}\Pi=\tilde{Q}P,
\]
where 
\begin{equation}
P=\begin{array}{cc}
 & \begin{array}{cc}
n-p\quad\; & p\end{array}\\
\begin{array}{c}
n-p\\
p
\end{array} & \left[\begin{array}{cc}
P_{11}\qquad & P_{12}\\
0\qquad & P_{22}
\end{array}\right]
\end{array}.\label{eq:P-equivalent}
\end{equation}
Note in Type-I algorithm $\sigma_{min}\left(P_{11}\right)$ is maximized.
With some block permutations the above $P$ could be written as, 
\[
\begin{array}{cc}
 & \begin{array}{cc}
p\quad\; & n-p\end{array}\\
\begin{array}{c}
p\\
n-p
\end{array} & \left[\begin{array}{cc}
P_{22}\qquad & 0\\
P_{12}\qquad & P_{11}
\end{array}\right]
\end{array}
\]
Now, permuting the rows and columns of the individual blocks, 
\[
\tilde{P}=\begin{array}{cc}
 & \begin{array}{cc}
p\qquad\qquad & n-p\end{array}\\
\begin{array}{c}
p\\
n-p
\end{array} & \left[\begin{array}{cc}
J_{p}P_{22}J_{p}\qquad & 0\\
J_{n-p}P_{12}J_{p}\qquad & J_{n-p}P_{11}J_{n-p}
\end{array}\right]
\end{array}
\]
where $J_{l}$ are permutation matrices of order $l$ with ones along
the anti-diagonal. This ensures $\tilde{P}$ is lower triangular.
Now, the above $\tilde{P}$ is equivalent to the structure in (\ref{eq:Unification2}),
thus $J_{n-p}P_{11}J_{n-p}$ can be equated to $R_{22}^{-\top}$ and
from (\ref{eq:Unification1}), maximizing $\sigma_{min}\left(R_{22}^{-\top}\right)$
is equivalent to minimizing $\sigma_{max}\left(R_{22}\right)$ which
is the original Type -II problem. Thus, from the above it is evident
that applying a Type-I algorithm to the rows of the inverse of a matrix
$A$ is equivalent to applying Type-II algorithm directly to $A$.

\subsubsection*{Algorithm}

Thus, the Stewart's algorithm could be written as applying QR-CP algorithm
to rows of the Inverse, \emph{i.e., }
\begin{description}
\item [{Initialize}]~
\item [{$s=0$\textmd{;}}] $\Pi=I\in\mathbb{R}^{n\times n}$; $A=QR$;
$R^{(0)}=R$, where $R^{(s)}=\begin{array}{cc}
 & \begin{array}{cc}
n-s\quad\; & s\end{array}\\
\begin{array}{c}
n-s\\
s
\end{array} & \left[\begin{array}{cc}
R_{11}^{(s)}\qquad & R_{12}^{(s)}\\
0\qquad & R_{22}^{(s)}
\end{array}\right]
\end{array}.$ 
\item [{Repeat}] 

\begin{description}
\item [{$j=\arg\max_{i}\left\Vert R_{11}^{(s)-1}(i,\,:)\right\Vert _{2}$}]~
\item [{if}] $j\neq n-s$,

\begin{description}
\item [{\textmd{Interchange}}] columns $n-s$ and $j$ of $\Pi$ and retriangularize
$R^{(s)}\Pi=QR^{(s+1)}$.
\end{description}
\item [{endif}]~
\item [{$s=s+1;$}]~
\end{description}
\item [{Until}] $s<n-p$.
\end{description}
Note that $R_{11}^{(s)}$ and $R_{22}^{(s)}$ changes in dimension
with every increase in $s$

The following lemma states the bounds satisfied by $R_{22}$ got using
Stewart-II decomposition.
\begin{lem}
Using Stewart's Type-II algorithm ensures the following:
\begin{equation}
\sigma_{max}\left(R_{22}\right)\leq\sigma_{p+1}\left(A\right)\sqrt{n\left(n-p\right)}2^{n-p}\label{eq:prop3}
\end{equation}
\end{lem}
\begin{proof}
Using Type-I, QR-CP algorithm (\ref{eq:P-equivalent}) ensures (ref
(\ref{eq:prop2})), 
\begin{equation}
\sigma_{min}\left(P_{11}\right)\geq\frac{\sigma_{n-p}\left(A\right)}{\sqrt{n-pn}2^{n-p}}.\label{eq:prop2-1}
\end{equation}
Using Unification principle, $P_{11}=R_{22}^{-\top}$ also, using
(\ref{eq:Unification1}) leads to, 
\begin{equation}
\sigma_{max}\left(R_{22}\right)=\frac{1}{\left(\sigma_{min}\left(P_{11}\right)\right)}.\label{eq:prop2-2}
\end{equation}
Thus from (\ref{eq:prop2-1}) and (\ref{eq:prop2-2}), (\ref{eq:prop3})
is got.
\end{proof}
Given a matrix $A\in\mathbb{R}^{n\times n}$. Applying Stewart-II
algorithm assuming a numerical rank of $n-1$ would ensure that, the
column with the lowest $2-norm$ in $A$ gets permuted with the last
column. Note, the weakest column in $A$ is the column that has the
lowest 2-norm or the row with the highest 2-norm in $A^{-1}$. 

\subsection*{Hybrid Algorithms}

The Hybrid Algorithms are got by combining QR-CP and the Stewart-II
algorithm. The first Hybrid algorithm we'll be dealing with is the
Hybrid-I algorithm. 

Consider the matrix $A\in R^{m\times n}$ with numerical rank $p$.
The QR decomposition of $A$ be given by $A=QR$. Now, segment $R$
in two different ways, 
\begin{eqnarray}
R & \triangleq & \begin{array}{cc}
 & \begin{array}{cc}
p\quad\; & n-p\end{array}\\
\begin{array}{c}
p\\
m-p
\end{array} & \left[\begin{array}{cc}
R_{11}\qquad & R_{12}\\
0\qquad & R_{22}
\end{array}\right]
\end{array},\label{eq:method-1}\\
 & \triangleq & \begin{array}{cc}
 & \begin{array}{cc}
p-1\quad\; & n-p+1\end{array}\\
\begin{array}{c}
p-1\\
m-p+1
\end{array} & \left[\begin{array}{cc}
\bar{R}_{11}\qquad & \bar{R}_{12}\\
0\qquad & \bar{R}_{22}
\end{array}\right]
\end{array}.\label{eq:method-2}
\end{eqnarray}
The Hybrid-I algorithm works by iterating between the QR-CP algorithm
applied to $\bar{R}_{22}$ and Stewart-II algorithm applied to $R_{11}$,
until no permutations occur. The Stewart-II applied to $R_{11}$ ensures
that the $p^{th}$ column in $R$ is the weakest column (the column
with the lowest 2-norm) in $R_{11}$. Also QR-CP applied to $\bar{R}_{22}$,
ensures that the $p^{th}$ column of $R$ is the best column (column
with the highest 2-norm) in $\bar{R}_{22}$ (\ref{eq:method-2}).

Hybrid-I satisfies both the QR-CP and Stewart-II's conditions. Note,
QR-CP to $\bar{R}_{22}$ to get the $p^{th}$ column of $R$ ensures
that, 
\begin{equation}
r_{pp}\geq\frac{\left\Vert \bar{R}_{22}\right\Vert _{2}}{\sqrt{n-p+1}},\label{eq:rpp_ineq3}
\end{equation}
where $r_{pp}$ is the $p^{th}$ diagonal element of $R$ (refer to
(\ref{eq:prop1_proof1})). Since $R_{22}$ (\ref{eq:method-1}) is
a sub matrix to $\bar{R}_{22}$,
\begin{equation}
r_{pp}\geq\frac{\left\Vert R_{22}\right\Vert _{2}}{\sqrt{n-p+1}}.\label{eq:rpp_ineq}
\end{equation}
Also, Stewart-II applied to $R_{11}$ (\ref{eq:method-1}) ensures
that 
\begin{equation}
r_{pp}\leq\sigma_{min}\left(R_{11}\right)\sqrt{p}.\label{eq:rpp_ineq2}
\end{equation}
combining (\ref{eq:rpp_ineq}) and (\ref{eq:rpp_ineq2}) leads to,
\[
\left\Vert R_{22}\right\Vert _{2}\leq\sigma_{min}\left(R_{11}\right)\sqrt{p\left(n-p+1\right)}.
\]
Using Interlacing property, $\sigma_{max}\left(\bar{R}_{22}\right)\geq\sigma_{p}\left(A\right)$
applying this fact in (\ref{eq:rpp_ineq3}) and using (\ref{eq:rpp_ineq2})
leads to, 
\[
\sigma_{min}\left(R_{11}\right)\geq\frac{\sigma_{p}\left(A\right)}{\sqrt{p\left(n-p+1\right)}}.
\]
Thus the required bounds for Hybrid-I are derived. The detailed step-by-step
procedure for implementation of Hybrid-I algorithm for $\tilde{M}$,
\ref{eq:M_perturbed-1}, is summarized in Algorithm (\ref{alg:Hybrid-I}).

\begin{algorithm}
\caption{\noun{\label{alg:Hybrid-I}Hybrid-I} algorithm of RRQR}
~\\
~
\begin{enumerate}
\item The Matrix to be decomposed is got as $\tilde{M}\in\mathbb{R}^{K\times Km}$.
Let the numerical rank be assumed to be $\hat{p}$.
\item Using the $\Pi$ from\emph{ Algorithm (\ref{alg:al1})}, apply QR
decomposition to $\tilde{M}\Pi$ as $\tilde{M}\Pi=\bar{Q}\bar{R}$
and set $s=0$; 
\item \textbf{Repeat~}

\begin{enumerate}
\item \textbf{QR-CP Block}

\begin{enumerate}
\item Set $permuted=0$. 
\item Let $\bar{R}^{(s)}\triangleq\begin{array}{cc}
 & \begin{array}{cc}
\hat{p}-1\quad\; & mK-\hat{p}+1\end{array}\\
\begin{array}{c}
\hat{p}-1\\
K-\hat{p}+1
\end{array} & \left[\begin{array}{cc}
\bar{R}_{11}^{(s)}\qquad & \bar{R}_{12}^{(s)}\\
0\qquad & \bar{R}_{22}^{(s)}
\end{array}\right]
\end{array}$
\item Let $e_{i}$ be the $i^{th}$ column of an identity matrix of order
$mK-\hat{p}+1$.
\item Let\textbf{ }$j\triangleq\arg\max_{1\leq i\leq mK-\hat{p}+1}\left\Vert \bar{R}_{22}^{(s)}e_{i}\right\Vert _{2}$
\item \textbf{If $j\neq1$ then}\\
\textbf{~}\\
\textbf{~}

\begin{itemize}
\item $permuted$= 1
\item Interchange columns $\hat{p}+j-1$ and $\hat{p}$ in $\Pi$.
\item apply QR-GS to $\tilde{M}\Pi$ to get new $\bar{Q}\bar{R}$.
\end{itemize}
\end{enumerate}
\item \textbf{Stewart-II Block}

\begin{enumerate}
\item Now divide $\bar{R}$ as $\bar{R}^{(s)}\triangleq\begin{array}{cc}
 & \begin{array}{cc}
\hat{p}\quad\; & mK-\hat{p}\end{array}\\
\begin{array}{c}
\hat{p}\\
K-\hat{p}
\end{array} & \left[\begin{array}{cc}
\bar{R}_{11}^{(s)}\qquad & \bar{R}_{12}^{(s)}\\
0\qquad & \bar{R}_{22}^{(s)}
\end{array}\right]
\end{array}$and Let $e_{i}$ be the $i^{th}$ column of an identity matrix of
order $\hat{p}$.
\item Let $j\triangleq\arg\max_{1\leq i\leq\hat{p}}\left\Vert e_{i}^{\top}\bar{R}_{11}^{(s)-1}\right\Vert $
\item \textbf{If $j\neq\hat{p}$ then}\\
\textbf{~}\\
\textbf{~}

\begin{itemize}
\item $permuted=1$
\item Exchange columns $j$ and $\hat{p}$ in $\Pi$
\item apply QR-GS to $\tilde{M}\Pi$ to get new $\bar{Q}^{(s)}\bar{R}^{(s)}$.
\item $s=s+1$
\end{itemize}
\end{enumerate}
\end{enumerate}
\item \textbf{Until $permuted\neq0$}
\item Now RRQR decomposition is given by taking GS-QR of $\tilde{M}\Pi$
as, $\tilde{M}\Pi=\tilde{Q}\tilde{R}$
\end{enumerate}
\end{algorithm}

While Hybrid-I algorithm uses QR-CP and Stewart-II algorithm to ensure
the $p^{th}$ column is the best in $\bar{R}_{22}$ and the worst
in $R_{11}$ respectively, segmenting $R$ as 
\begin{eqnarray}
R & \triangleq & \begin{array}{cc}
 & \begin{array}{cc}
p+1\quad\; & n-p-1\end{array}\\
\begin{array}{c}
p+1\\
m-p-1
\end{array} & \left[\begin{array}{cc}
R_{11}\qquad & R_{12}\\
0\qquad & R_{22}
\end{array}\right]
\end{array},\label{eq:method-1-1}\\
 & \triangleq & \begin{array}{cc}
 & \begin{array}{cc}
p\quad\; & n-p\end{array}\\
\begin{array}{c}
p\\
m-p
\end{array} & \left[\begin{array}{cc}
\bar{R}_{11}\qquad & \bar{R}_{12}\\
0\qquad & \bar{R}_{22}
\end{array}\right]
\end{array}.\label{eq:method-2-1}
\end{eqnarray}
and using QR-CP and Stewart-II to ensure that the $p+1^{th}$ column
of $R$ is the best possible in $\bar{R}_{22}$ and worst of $R_{11}$
respectively results in the Hybrid-II decomposition.

The Hybrid-II equivalent of (\ref{eq:rpp_ineq3}) is, 
\[
r_{p+1,p+1}\geq\frac{\sigma_{max}\left(\bar{R}_{22}\right)}{\sqrt{\left(n-p\right)}},
\]
where $r_{p+1,p+1}$ is the $p+1^{th}$ diagonal value of $R$. Also,
the equivalent of (\ref{eq:rpp_ineq2}) is given by 
\begin{equation}
r_{p+1,p+1}\leq\sigma_{min}\left(R_{11}\right)\sqrt{p+1}.\label{eq:rpp_ineq2-1}
\end{equation}
combining the above two inequalities results in 
\[
\sigma_{min}\left(R_{11}\right)\geq\frac{\sigma_{max}\left(\bar{R}_{22}\right)}{\sqrt{\left(p+1\right)\left(n-p\right)}}.
\]
Due to Interlacing property, $\left(\sigma_{min}\left(\bar{R}_{11}\right)\geq\sigma_{min}\left(R_{11}\right)\right)$,
\[
\sigma_{min}\left(\bar{R}_{11}\right)\geq\frac{\sigma_{max}\left(\bar{R}_{22}\right)}{\sqrt{\left(p+1\right)\left(n-p\right)}}.
\]
Also, using Interlacing Property and the fact that $\sigma_{i}\left(A\right)=\sigma_{i}\left(R\right)$,
\[
\sigma_{min}\left(R_{11}\right)\leq\sigma_{k+1}\left(A\right)
\]
Thus, 
\[
\sigma_{max}\left(\bar{R}_{22}\right)\leq\sigma_{k+1}\left(A\right)\sqrt{\left(p+1\right)\left(n-p\right)}.
\]
which completes the bounds for Hybrid-II algorithm. Thus, it can be
seen that applying Hybrid-II algorithm with the assumption of numerical
rank $p$ is equivalent to applying Hybrid-I algorithm assuming a
numerical rank of $p+1$.

The Hybrid-III algorithm is as follows: \\
Apply Hybrid-I algorithm on the matrix $A$ followed by the application
of Hybrid-II. Repeat this procedure until no permutations occur.

Thus, when Hybrid-III halts, it satisfies the bounds for both Hybrid-I
and Hybrid-II.

%% file: includes/Chapter-4.tex
\selectlanguage{british}%

\chapter{\label{chap:Model-order}Model order: A Numerical Rank Perspective}

In this chapter, the determination of model order (\emph{i.e.,} the
number of factors $p$) is framed as the problem of estimating the
numerical rank of the matrix $\tilde{M}$ (\ref{eq:M_perturbed-1}).
Recall that the matrix $\tilde{M}$ could be considered as the perturbed
version of the matrix $M$ given by: 
\[
\tilde{M}=M+\Delta M.
\]
While $M$ defined by (\ref{eq:M_def}) is of rank $p$, $\tilde{M}$
is of full row rank $K$ with probability one. When $M$ is known,
determining the model order is just to equate it to the rank of the
matrix $M$. Since in practice only $\tilde{M}$ is available, we
resort to the estimation of the numerical rank of $\tilde{M}$. Consider
the following definition taken from \foreignlanguage{english}{\cite{golub1976rank}:}
\selectlanguage{english}%
\begin{defn}
\cite{golub1976rank}: A matrix $A$ has numerical rank $(\mu,\,\varepsilon,\,p)$
with respect to norm $\left\Vert .\right\Vert $ if $\mu$, $\varepsilon$,
$p$ satisfy the following two conditions: 
\end{defn}
\begin{enumerate}
\item $p=\inf\;\left\{ rank(B):\left\Vert A-B\right\Vert \leq\varepsilon\right\} $ 
\item $\varepsilon<\mu\leq\sup\left\{ \eta:\left\Vert A-B\right\Vert \leq\eta\implies rank(B)\geq p\right\} $ 
\end{enumerate}
\selectlanguage{british}%
for ease of notation, $\left(\mu,\,\varepsilon,\,p\right)_{p}$ is
used to denote the numerical rank w.r.t norm $\left\Vert .\right\Vert _{p}$.
\foreignlanguage{english}{The following example is used to illustrate
the above definition:}
\selectlanguage{english}%
\begin{example}
Consider the matrix $A\triangleq\left[\begin{array}{ccc}
18 & 0 & 0\\
0 & 5 & 0\\
0 & 0 & 0.8
\end{array}\right]$. Let $\varepsilon,\,\mu$ be bounded by, 
\begin{equation}
0.8<\varepsilon<\mu<5\label{eq:e_mu_condn1}
\end{equation}
with the above bounds for, $\varepsilon,\,\mu$, the numerical rank
of $A$ is $p=2$ w.r.t $\left\Vert .\right\Vert _{2}$.

It can be explained as follows: Consider the matrix, 
\[
B=\left[\begin{array}{ccc}
18 & 0 & 0\\
0 & \epsilon & 0\\
0 & 0 & 0
\end{array}\right]
\]
Note, $\inf\{rank\left(B\right)\}=1\implies$$\epsilon=0$. Which
inturn implies $\left\Vert A-B\right\Vert _{2}=5$. This means, $\varepsilon\geq5$
if $B$ is to take rank of $1$ and since it is outside the bound,
(\ref{eq:e_mu_condn1}), $\inf\{rank\left(B\right)\}=2\,\forall\varepsilon<5.$
Now, consider the matrix 
\[
B=\left[\begin{array}{ccc}
18 & 0 & 0\\
0 & 5 & 0\\
0 & 0 & \epsilon
\end{array}\right].
\]
Note, $\inf\{rank\left(B\right)\}=3\implies\epsilon>0$ which inturn
implies$\left\Vert A-B\right\Vert _{2}$ should not be greater than
$0.8$. But the bound, (\ref{eq:e_mu_condn1}), ensures that $\varepsilon>0.8.$
Thus, numerical rank cannot be $3$ implying under the selected bound,
(\ref{eq:e_mu_condn1}), numerical rank of $A$ is $2$.
\end{example}
\selectlanguage{british}%
Note, the above example can be extended to any matrix, $A$, as stated
in the following theorem:
\selectlanguage{english}%
\begin{thm*}
\cite{golub1976rank} Let $\sigma_{1}>\sigma_{2}>\ldots>\sigma_{K}$
denote the singular values of the $K\times Km$ matrix $\tilde{M}$.
Then, $\tilde{M}$ is of numerical rank $\left(\mu,\,\varepsilon,\,p\right)$
w.r.t $\Vert.\Vert_{2}$ if 
\begin{equation}
\sigma_{p+1}\left(\tilde{M}\right)\leq\varepsilon<\mu\leq\sigma_{p}\left(\tilde{M}\right)\label{eq:SVD_num_rank}
\end{equation}
\end{thm*}
\begin{proof}
Refer \cite{golub1976rank} for a proof 
\end{proof}
Now, Consider a matrix $A$ perturbed by another matrix $\Delta A$.
The following theorem shows how this disturbance of $\Delta A$ will
affect the numerical rank of the matrix $\tilde{A}=A+\Delta A$.
\selectlanguage{british}%
\begin{thm*}
Let $\varepsilon,\,\mu$ be bounded by:
\[
\sigma_{p+1}(A)+\sigma_{max}(\Delta A)\leq\varepsilon<\mu\leq\sigma_{p}\left(A\right)-\sigma_{max}\left(\Delta A\right).
\]
Under this bound, both $A$ and $A+\Delta A$ is of same numerical
rank $\left(\mu,\,\varepsilon,\,p\right)_{2}$.
\end{thm*}
\begin{proof}
Assume 
\begin{equation}
\sigma_{p+1}(A)+\sigma_{max}(\Delta A)\leq\varepsilon<\mu\leq\sigma_{p}\left(A\right)-\sigma_{max}\left(\Delta A\right)\label{eq:bound-1}
\end{equation}
holds true. Since, 
\[
\sigma\left(A+\Delta A\right)<\sigma_{p+1}(A)+\sigma_{max}(\Delta A)
\]
and 
\[
\sigma_{p}\left(A+\Delta A\right)>\sigma_{p}\left(A\right)-\sigma_{max}\left(\Delta A\right)
\]
holds true, 
\begin{equation}
\sigma_{p+1}\left(A+\Delta A\right)\leq\varepsilon<\mu\leq\sigma_{p}\left(A+\Delta A\right)\label{eq:bound-2}
\end{equation}
is satisfied which implies that $A+\Delta A$ is of numerical rank
$\left(\mu,\,\varepsilon,\,p\right)_{2}$. Also, since $\sigma_{p+1}(A)+\sigma_{max}(\Delta A)\geq\sigma_{p+1}\left(A\right)$
and $\sigma_{p}\left(A\right)-\sigma_{max}\left(\Delta A\right)\leq\sigma_{p}\left(A\right)$,
(\ref{eq:bound-1}) implies 
\[
\sigma_{p+1}\left(A\right)\leq\varepsilon<\mu\leq\sigma_{p}\left(A\right)
\]
 which in turn implies $A$ is also of numerical rank $\left(\mu,\,\varepsilon,\,p\right)_{2}$
.
\end{proof}
The above theorem gives an idea of bounds on $\mu,\,\varepsilon$
based on the amount of perturbation $\Delta A$ such that the numerical
rank remains constant.

In the current scenario, we have the matrix $M$ of rank $p$ perturbed
by a matrix $\Delta M$ of full rank. Thus for $\tilde{M}=M+\Delta M$
to have numerical rank $\left(\mu,\,\varepsilon,\,p\right)_{2}$,
\[
\sigma_{max}\left(\Delta M\right)\leq\varepsilon<\mu\leq\sigma_{p}\left(M\right)-\sigma_{max}\left(\Delta M\right)
\]
In our scenario, $\sigma_{max}\left(\Delta M\right)\rightarrow0$
when $N\rightarrow\infty$ because of assumption $(A5)$ hence $\sigma_{max}\left(\Delta M\right)\ll\sigma_{p}\left(M\right)$
as $N\rightarrow\infty$. Hence the separation of $\sigma_{p+1}\left(\tilde{M}\right)$
and $\sigma_{p}\left(\tilde{M}\right)$ grows with increase in $N$
which enabled the ratio based estimate \cite{C1}, given by 
\begin{equation}
\hat{p}=\arg\max_{1\leq i\leq R}\frac{\sigma_{i}\left(\tilde{M}\right)}{\sigma_{i+1}\left(\tilde{M}\right)},\label{eq:SVD_model_order}
\end{equation}
 to work.

\selectlanguage{english}%
The above determination of numerical rank was with the help of using
Singular Value Decomposition (SVD), for asymptotic properties of (\ref{eq:SVD_model_order}),
refer to \cite{C1}. Now, to analyze how to determine the numerical
rank using RRQR decomposition, consider the following theorem proposed
by Golub\cite{golub1976rank}:
\begin{thm*}
Let 
\begin{eqnarray}
\tilde{M}\Pi & = & \left[\begin{array}{cc}
\mathcal{\tilde{Q}}_{\hat{p}} & \tilde{\mathcal{Q}}_{K-\hat{p}}\end{array}\right]\left[\begin{array}{cc}
\tilde{R}_{11}^{(\hat{p})} & \tilde{R}_{12}\\
0 & \tilde{R}_{22}^{(K-\hat{p})}
\end{array}\right]\label{eq:QR_columnpivoting}
\end{eqnarray}
be the QR decomposition of $\tilde{M}\Pi$. If there exists,$\hat{\mu},\hat{\epsilon}>0$
such that 
\begin{equation}
\sigma_{min}\left(\tilde{R}_{11}^{(\hat{p})}\right)=\hat{\mu}>\hat{\epsilon}=\left\Vert \tilde{R}_{22}^{(K-\hat{p})}\right\Vert _{2},\label{eq:Thm2_1}
\end{equation}
then, $\tilde{M}$ has numerical rank $\left(\hat{\mu},\,\hat{\epsilon},\,\hat{p}\right)$.
\end{thm*}
\begin{proof}
Refer to \cite{golub1976rank} for a proof.
\end{proof}
As the sample correlations are assumed to converge to their respective
expected values with increasing sample size $N$ (for a fixed $K$),
(\ref{eq:Thm2_1}) must hold with a high probability for large values
of $N$ at $\hat{p}=p$. This is proved in the following theorem.
\begin{thm}
\label{thm:A}Assume that 
\begin{eqnarray}
M & \triangleq & \left[\begin{array}{cccc}
\Sigma_{yy}(1) & \Sigma_{yy}(2) & \ldots & \Sigma_{yy}(m)\end{array}\right]\label{eq:M_def-1}
\end{eqnarray}
 is of rank $'p'$ where $\Sigma_{yy}(l)$ denotes the autocovariance
matrix of the observations $y_{n}$. Let $\Pi$ be the permutation
matrix obtained using \noun{Hybrid-III} and let $\tilde{M}\Pi$ be
decomposed as in 
\begin{eqnarray}
\tilde{M}\Pi & = & \left[\begin{array}{cc}
\tilde{M}_{p} & \tilde{M}_{mK-p}\end{array}\right]\nonumber \\
 & = & \left[\begin{array}{cc}
\mathcal{\tilde{Q}}_{p} & \tilde{\mathcal{Q}}_{K-p}\end{array}\right]\left[\begin{array}{cc}
\tilde{R}_{11} & \tilde{R}_{12}\\
0 & \tilde{R}_{22}
\end{array}\right]\nonumber \\
 & = & \mbox{\ensuremath{\mathcal{\tilde{Q}}}}\tilde{R},\label{eq:rrqr1-1-1}
\end{eqnarray}
It can be shown that under the assumptions (A1-A7), refer Chapter
\ref{chap:Algorithm} regarding assumptions, with $K,\,N\rightarrow\infty$
under the constraint\footnote{Note, (\ref{eq:thm2_1}) implies that $\sqrt{N}$ grows faster than
$K^{1+\delta}$.} 
\begin{equation}
\frac{K^{1+\delta}}{\sqrt{N}}=o(1),\label{eq:thm2_1}
\end{equation}
it holds that, 
\begin{equation}
o_{p}\left(\sigma_{min}\left(\tilde{R}_{11}^{(p)}\right)\right)=\left\Vert \tilde{R}_{22}^{(p)}\right\Vert _{2}.\label{eq:thm2_2}
\end{equation}
 
\end{thm}
\begin{proof}
The following three lemmas are required to prove the above theorem: 
\end{proof}
\begin{lem}
Let $\gamma_{i}$ and $\gamma_{i+1}$ be the $i^{th}$ and $\left(i+1\right)^{th}$
diagonal value of $\tilde{R}$, (\ref{eq:rrqr1-1-1}), got by applying
\noun{Hybrid-III}\emph{ }algorithm with assumption of numerical rank
to be $i$ .The Hybrid - III algorithm ensures that $\gamma_{i}$
satisfies the following property: 
\begin{equation}
\sigma_{i}\left(\tilde{M}\right)\sqrt{\left(i+1\right)\left(mK-i\right)}\geq\gamma_{i}\left(\tilde{R}\right)\geq\frac{\sigma_{i}\left(\tilde{M}\right)}{\sqrt{mK-i}},\label{eq:L1-1}
\end{equation}
where $\gamma_{i}$ is the $i^{th}$ diagonal element of $\tilde{R}$
got by applying \noun{Hybrid-III} algorithm assuming a numerical rank
of $i$ and $\sigma_{i}\left(\tilde{M}\right)$ is the $i^{th}$ highest
singular value of $\tilde{M}$.
\end{lem}
\begin{proof}
Let $\gamma_{s+1}$ be the $s+1^{th}$ diagonal value of $\tilde{R}$
in, (\ref{eq:rrqr1-1-1}), where $\Pi$ is got using \emph{\noun{Hybrid}}-\noun{III}
algorithm assuming numerical rank $s$. Segment $\tilde{R}$ as, 
\[
\tilde{R}\triangleq\begin{array}{cc}
 & \begin{array}{cc}
s\quad\; & mK-s\end{array}\\
\begin{array}{c}
s\\
K-s
\end{array} & \left[\begin{array}{cc}
\tilde{R}_{11}^{(s)}\qquad & \tilde{R}_{12}^{(s)}\\
0\qquad & \tilde{R}_{22}^{(s)}
\end{array}\right]
\end{array}.
\]
Let $\alpha_{i}$ be the 2-norm of the $i^{th}$ column of $\tilde{R}_{22}^{(s)}$
, \emph{i.e.,} 
\[
\alpha_{i}=\left\Vert \tilde{R}_{22}^{(s)}(:,\,i)\right\Vert _{2}.
\]
The column pivoting part of\emph{ }\noun{Hybrid-III} algorithm ensures
$\gamma_{s+1}$ to be, 
\[
\gamma_{s+1}=\max_{1\leq i\leq mK-s}\alpha_{i}.
\]
Note that, 
\begin{equation}
\left\Vert \tilde{R}_{22}^{(s)}\right\Vert _{2}\geq\gamma_{s+1}\geq\frac{\left\Vert \tilde{R}_{22}^{(s)}\right\Vert _{F}}{\sqrt{mK-s}}.\label{eq:pl1}
\end{equation}
Since $\left\Vert \tilde{R}_{22}^{(s)}\right\Vert _{F}\geq\left\Vert \tilde{R}_{22}^{(s)}\right\Vert _{2}$,
\begin{equation}
\left\Vert \tilde{R}_{22}^{(s)}\right\Vert _{2}\geq\gamma_{s+1}\geq\frac{\left\Vert \tilde{R}_{22}^{(s)}\right\Vert _{2}}{\sqrt{mK-s}}.\label{eq:pl1-1}
\end{equation}
Interlacing property of Singular Values \cite{golub2012matrix,li2012generalized}
leads to 
\[
\sigma_{s}\left(\tilde{R}\right)\geq\left\Vert \tilde{R}_{22}^{(s)}\right\Vert _{2}\geq\sigma_{s+1}(\tilde{R}),
\]
and the \emph{Hybrid - III }algorithm \cite{chandrasekaran1994rank}
guarantees,
\[
\left\Vert R_{22}\right\Vert _{2}\leq\sigma_{s+1}\left(\tilde{R}\right)\sqrt{\left(s+1\right)\left(mK-s\right)}.
\]
Thus, 
\begin{eqnarray}
\sigma_{s+1}\left(\tilde{R}\right)\sqrt{\left(s+1\right)\left(mK-s\right)}\geq\left\Vert \tilde{R}_{22}^{(s)}\right\Vert _{2} & \geq & \sigma_{s+1}(\tilde{R}).\label{eq:pl2}
\end{eqnarray}
 Since, $\sigma_{i}\left(\tilde{M}\right)=\sigma_{i}\left(\tilde{R}\right)$,
substituting (\ref{eq:pl2}) in (\ref{eq:pl1-1}) we get (\ref{eq:L1-1}).
\end{proof}
\begin{lem}
\label{lem:A1}The following rates of convergence hold as $K,N\rightarrow\infty$
under the assumptions (A1-A3): 
\begin{equation}
\left\Vert \Sigma_{ff}(l)\right\Vert _{F}\leq\left\Vert H\right\Vert _{F}^{2}\left\Vert \Sigma_{xx}(l)\right\Vert _{F}=O\left(K^{1-\delta}\right)=\left\Vert \Sigma_{ff}(l)\right\Vert _{2}\asymp\sigma_{min}\left(\Sigma_{ff}(l)\right)\label{eq:l3_1}
\end{equation}
\end{lem}
\begin{proof}
Note 
\begin{eqnarray*}
y_{n} & = & Hx_{n}+\varepsilon_{n}\\
 & = & Qf_{n}+\varepsilon_{n},
\end{eqnarray*}
where, $H=QR$ is the QR decomposition of $H$ and $f_{n}=Rx_{n}$.
Therefore, 
\begin{eqnarray*}
\left[\begin{array}{cccc}
h_{1} & h_{2} & \ldots & h_{p}\end{array}\right] & = & \left[\begin{array}{cc}
Q_{1} & Q_{2}\end{array}\right]\left[\begin{array}{c}
R_{1}\\
0
\end{array}\right],
\end{eqnarray*}
where 
\[
Q_{1}=\left[\begin{array}{ccc}
q_{1} & \ldots & q_{p}\end{array}\right]
\]
 and 
\[
R_{1}=\left[\begin{array}{cccc}
<h_{1},q_{1}> & <h_{2},q_{1}> & \ldots & <h_{p},q_{1}>\\
0 & <h_{2,}q_{2}> & \ldots & <h_{p},q_{2}>\\
0 & 0 & \ddots & \vdots\\
0 & 0 & 0 & <h_{p},q_{p}>
\end{array}\right].
\]
Here $<a,b>$ represents the inner product of $a$ and $b$. Thus
the $i^{th}$ diagonal element of $R$ given by 
\begin{eqnarray}
\left|<h_{i},q_{i}>\right| & = & \left\Vert h_{i}-\sum_{k=1}^{k=i-1}<h_{i},q_{k}>q_{k}\right\Vert _{2}\label{eq:lem1,R}
\end{eqnarray}
has a maximum value $\left|<h_{i},q_{i}>\right|\,\mbox{can take is}\left\Vert h_{i}\right\Vert _{2}=O\left(K^{\frac{1-\delta}{2}}\right)$.
Since in equation (\ref{eq:lem1,R}) each $q_{k}$ could be expressed
as a linear combination of $h_{1},\ldots,h_{k}$, the minimum value
of (\ref{eq:lem1,R}), got under the assumption (\ref{eq:factor_strength_condn2}),
is $\left|<h_{i},q_{i}>\right|{}_{min}=O\left(K^{\frac{1-\delta}{2}}\right).$
Hence, the diagonal elements of $R$ are of the order , $\mathcal{O}\left(K^{\frac{1-\delta}{2}}\right)$.
And since the off-diagonal elements are of lesser or same order, it
can be inferred that 
\[
\left\Vert R\right\Vert _{F}=\sqrt{p}O\left(K^{\frac{1-\delta}{2}}\right)=O\left(K^{\frac{1-\delta}{2}}\right).
\]

Since $\Sigma_{xx}(l)$ is a constant matrix independent of $K$ ,
$(A1)$, 
\[
\left\Vert \Sigma_{xx}(l)\right\Vert _{F}=O(1).
\]
 As 
\[
\Sigma_{f}(k)=R\Sigma_{x}(k)R^{T},\:\forall k=1,2,\ldots,
\]
\begin{equation}
\left\Vert \Sigma_{f}(k)\right\Vert _{F}\leq\left\Vert R\right\Vert _{F}\left\Vert \Sigma_{x}(k)\right\Vert _{F}\left\Vert R^{T}\right\Vert _{F}=O\left(K^{1-\delta}\right).\label{eq:lemma3_1a}
\end{equation}
The proof for 
\begin{equation}
\left\Vert \Sigma_{ff}(l)\right\Vert _{2}\asymp\sigma_{min}\left(\Sigma_{ff}(l)\right)\asymp O\left(K^{1-\delta}\right)\label{eq:lemma3_1}
\end{equation}
 is given in \cite{lam2011estimation}. Thus, from (\ref{eq:lemma3_1a})
and (\ref{eq:lemma3_1}), (\ref{eq:l3_1}) is got.
\end{proof}
\begin{lem}
\label{lem:lemma 5}The following rate of convergences hold under
the assumptions (A1-A7) :

\begin{eqnarray}
\left\Vert \Delta\Sigma_{f\varepsilon}(l)\right\Vert _{2}\asymp_{P}\left\Vert \Delta\Sigma_{f\varepsilon}(l)\right\Vert _{F} & = & O_{P}\left(K^{1-\frac{\delta}{2}}N^{-\frac{1}{2}}\right)\label{eq:l4_1}\\
\left\Vert \Delta\Sigma_{ff}(l)\right\Vert _{2}\asymp_{P}\left\Vert \Delta\Sigma_{ff}(l)\right\Vert _{F} & = & O_{P}\left(K^{1-\delta}N^{-\frac{1}{2}}\right)\label{eq:l4_2}\\
\left\Vert \Delta\Sigma_{\varepsilon f}(l)\right\Vert _{2}\asymp_{P}\left\Vert \Delta\Sigma_{\varepsilon f}(l)\right\Vert _{F} & = & O_{P}\left(K^{1-\frac{\delta}{2}}N^{-\frac{1}{2}}\right)\label{eq:l4_3}\\
\left\Vert \Delta\Sigma_{\varepsilon\varepsilon}(l)\right\Vert _{2}\asymp_{P}\left\Vert \Delta\Sigma_{\varepsilon\varepsilon}(l)\right\Vert _{F} & = & O_{P}\left(KN^{-\frac{1}{2}}\right)\label{eq:l4_4}\\
\left\Vert F\right\Vert _{2}^{2} & = & O_{P}\left(K^{1-\delta}N\right)\label{eq:l4_5}
\end{eqnarray}
where $F$ is as in (\ref{eq:equation for F}) and $X\triangleq\left[\begin{array}{cccc}
x_{1} & x_{2} & \ldots & x_{N}\end{array}\right]$.
\end{lem}
\begin{proof}
$\tilde{\Sigma}_{f\varepsilon}(l)=R\tilde{\Sigma}_{x\varepsilon}(l)$
and $\Sigma_{f\varepsilon}(l)=R\Sigma_{x\varepsilon}(l)$. Hence,
\begin{eqnarray}
\left\Vert \Delta\Sigma_{f\varepsilon}(l)\right\Vert _{F} & = & \left\Vert R\left[\Delta\Sigma_{x\varepsilon}(l)\right]\right\Vert _{F}\nonumber \\
 & \leq & \left\Vert R\right\Vert _{F}\left\Vert \Delta\Sigma_{x\varepsilon}(l)\right\Vert _{F}.\label{eq:lemma2_1}
\end{eqnarray}
Each element in $\left[\tilde{\Sigma}_{x\varepsilon}(l)-\Sigma_{x\varepsilon}(l)\right]$
converges at the rate of $O_{P}\left(N^{-\frac{1}{2}}\right)$ According
to assumption $A5$ in Chapter (\ref{chap:Data-Model}).\footnote{Refer Appendix 2 (mixing properties) for details}
and there are $p_{l}K$ elements. Therefore, 
\begin{equation}
\left\Vert \Delta\Sigma_{x\varepsilon}(l)\right\Vert _{F}=O_{P}\left(\sqrt{\frac{p_{l}K}{N}}\right).\label{eq:lemma2_21}
\end{equation}
Using (\ref{eq:lemma2_21}) and the fact that $\left\Vert R\right\Vert _{F}=O\left(K^{\frac{1-\delta}{2}}\right)$,
along with (\ref{lem:A1}) and substituting in (\ref{eq:lemma2_1})
leads to (\ref{eq:l4_1}). A similar derivation yields, (\ref{eq:l4_3}). 

Note that 
\begin{eqnarray}
\left[\Delta\Sigma_{ff}(l)\right] & = & R\left[\Delta\Sigma_{xx}(l)\right]R^{\top}.\label{eq:lemma2_2}
\end{eqnarray}
 Since, every element in $\left[\tilde{\Sigma}_{xx}(l)-\Sigma_{xx}(l)\right]$
converges at $O_{P}\left(N^{-\frac{1}{2}}\right)$and there are in
total $p_{l}^{2}$elements 
\[
\left\Vert \Delta\Sigma_{xx}(l)\right\Vert _{F}=O_{P}\left(p_{l}N^{-\frac{1}{2}}\right).
\]
 Therefore, 
\begin{eqnarray*}
\left\Vert \Delta\Sigma_{ff}(l)\right\Vert _{F} & \leq & \left\Vert R\right\Vert _{F}\left\Vert \Delta\Sigma_{xx}(l)\right\Vert _{F}\left\Vert R\right\Vert _{F}\\
\, & = & O\left(K^{\frac{1-\delta}{2}}\right)O_{P}\left(p_{l}N^{-\frac{1}{2}}\right)\mathcal{O}\left(K^{\frac{1-\delta}{2}}\right)\\
\, & = & O_{P}\left(K^{1-\delta}N^{-\frac{1}{2}}\right).
\end{eqnarray*}

Now we prove, (\ref{eq:l4_4}). As, every element in the matrix $\left[\tilde{\Sigma}_{\varepsilon\varepsilon}(l)-\Sigma_{\varepsilon\varepsilon}(l)\right]$
converges as $O_{P}\left(N^{-\frac{1}{2}}\right)$ (Note condition
$(A5)$ )and there are a total of $K^{2}$elements, (\ref{eq:l4_4})
is obtained. Finally, to prove (\ref{eq:l4_5}), Note $\left\Vert F\right\Vert _{2}\leq\left\Vert R\right\Vert _{F}\left\Vert X\right\Vert _{2}=O_{P}\left(K^{\frac{1-\delta}{2}}N^{\frac{1}{2}}\right)$.
Squaring which, (\ref{eq:l4_5}), is obtained. 

The convergence rates of $\left\Vert \Delta\Sigma_{f\varepsilon}(l)\right\Vert _{2},\,\left\Vert \Delta\Sigma_{ff}(l)\right\Vert _{2},\,\left\Vert \Delta\Sigma_{\varepsilon f}(l)\right\Vert _{2}$
and $\left\Vert \Delta\Sigma_{\varepsilon\varepsilon}(l)\right\Vert _{2}$
are proved in \cite{lam2011estimation}. It is found that these rates
are same as that obtained for Frobenius norm derived above.
\end{proof}
\begin{lem}
\label{lem:lemma 6}Consider the matrices $M$ as in (\ref{eq:M_def-1}).
Let $\tilde{M}$ and $\Delta M$ be given by: 
\begin{equation}
\tilde{M}\triangleq\left[\begin{array}{cccc}
\tilde{\Sigma}_{yy}(1) & \tilde{\Sigma}_{yy}(2) & \ldots & \tilde{\Sigma}_{yy}(m)\end{array}\right].\label{eq:M_perturbed-1-1}
\end{equation}
and
\begin{eqnarray}
\Delta M & \triangleq & \tilde{M}-M\nonumber \\
 & \triangleq & \left[\begin{array}{cccc}
\Delta\Sigma_{yy}(1) & \Delta\Sigma_{yy}(2) & \ldots & \Delta\Sigma_{yy}(m)\end{array}\right]\label{eq:Delta_m-1}
\end{eqnarray}
 The following results hold under the assumptions (A1-A7) as $K,N\rightarrow\infty$: 

\begin{eqnarray}
\sigma_{p}\left(M\right) & \asymp & \begin{cases}
K^{1-\delta} & ,if\:\kappa_{max}=o\left(K^{1-\delta}\right)\\
\kappa_{min} & ,if\:o\left(\kappa_{min}\right)=K^{1-\delta}
\end{cases}\label{eq:l5_1}\\
\sigma_{1}\left(M\right) & \asymp & \max\left(\kappa_{max},\,K^{1-\delta}\right).\label{eq:l5_1b}\\
\left\Vert \Delta M\right\Vert _{2} & \asymp_{P} & \left\Vert \Delta M\right\Vert _{F}=O_{P}\left(KN^{-\frac{1}{2}}\right)\label{eq:l5_2}\\
\sigma_{p}(\tilde{M}) & \asymp_{P} & \begin{cases}
K^{1-\delta} & ,if\:\kappa_{max}=o\left(K^{1-\delta}\right)\\
\kappa_{min} & ,if\:o\left(\kappa_{min}\right)=K^{1-\delta}
\end{cases},\:\forall\frac{K^{\delta}}{\sqrt{N}}=o(1).\label{eq:l5_3}\\
\sigma_{1}\left(\tilde{M}\right) & \asymp_{P} & \max\left(\kappa_{max},\,K^{1-\delta}\right).\,\forall\frac{K^{\delta}}{\sqrt{N}}=o(1).\label{eq:l5_3b}
\end{eqnarray}
\end{lem}
\begin{proof}
$\sigma_{p}\left(M\right)$ is given by\footnote{Note: The following inequality holds true: $\sigma_{n}\left(A\right)-\sigma_{max}\left(B\right)\leq\sigma_{n}\left(A+B\right)\leq\sigma_{n}\left(A\right)+\sigma_{max}\left(B\right)$
(refer pp 363 of \cite{bernstein2009matrix})}

\begin{eqnarray}
\sigma_{p}(M) & = & \sigma_{p}\left[\begin{array}{ccc}
\Sigma_{yy}(1) & \ldots & \Sigma_{yy}(m)\end{array}\right]\nonumber \\
 & \geq & \sigma_{p}\left[\begin{array}{ccc}
\Sigma_{ff}(1) & \ldots & \Sigma_{ff}(m)\end{array}\right]-\sigma_{max}\left[\begin{array}{ccc}
\Sigma_{f\varepsilon}(1) & \ldots & \Sigma_{f\varepsilon}(m)\end{array}\right]\nonumber \\
 & (or)\label{eq:proof5}\\
\sigma_{p}(M) & \geq & \sigma_{p}\left[\begin{array}{ccc}
\Sigma_{f\varepsilon}(1) & \ldots & \Sigma_{f\varepsilon}(m)\end{array}\right]-\sigma_{max}\left[\begin{array}{ccc}
\Sigma_{ff}(1) & \ldots & \Sigma_{ff}(m)\end{array}\right]\nonumber 
\end{eqnarray}
Also,
\begin{eqnarray}
\sigma_{p}\left(M\right) & \leq & \sigma_{p}\left[\begin{array}{ccc}
\Sigma_{ff}(1) & \ldots & \Sigma_{ff}(m)\end{array}\right]+\sigma_{max}\left[\begin{array}{ccc}
\Sigma_{f\varepsilon}(1) & \ldots & \Sigma_{f\varepsilon}(m)\end{array}\right]\nonumber \\
 & (or)\label{eq:proof5-2}\\
\sigma_{p}(M) & \leq & \sigma_{p}\left[\begin{array}{ccc}
\Sigma_{f\varepsilon}(1) & \ldots & \Sigma_{f\varepsilon}(m)\end{array}\right]+\sigma_{max}\left[\begin{array}{ccc}
\Sigma_{ff}(1) & \ldots & \Sigma_{ff}(m)\end{array}\right]\nonumber 
\end{eqnarray}
If $\kappa_{max}=o\left(K^{1-\delta}\right)$, we get, 
\begin{eqnarray}
O\left(\sigma_{p}\left[\begin{array}{ccc}
\Sigma_{ff}(1) & \ldots & \Sigma_{ff}(m)\end{array}\right]\right)\geq\sigma_{p}(M) & \geq & O\left(\sigma_{p}\left[\begin{array}{ccc}
\Sigma_{ff}(1) & \ldots & \Sigma_{ff}(m)\end{array}\right]\right)\nonumber \\
\sigma_{p}\left(M\right) & \asymp & K^{1-\delta}.\label{eq:proof7}
\end{eqnarray}
If $K^{1-\delta}=o(\kappa_{min})$,
\begin{eqnarray}
O\left(\sigma_{p}\left[\begin{array}{ccc}
\Sigma_{f\varepsilon}(1) & \ldots & \Sigma_{f\varepsilon}(m)\end{array}\right]\right)\geq\sigma_{p}(M) & \geq & O\left(\sigma_{p}\left[\begin{array}{ccc}
\Sigma_{f\varepsilon}(1) & \ldots & \Sigma_{f\varepsilon}(m)\end{array}\right]\right)\nonumber \\
\sigma_{p}\left(M\right) & \asymp & \kappa_{min}\label{eq:proof8-1}
\end{eqnarray}
Thus, from (\ref{eq:proof7}) and (\ref{eq:proof8-1}) we get (\ref{eq:l5_1}).
Also, 
\begin{eqnarray}
\sigma_{1}(M) & = & \sigma_{1}\left[\begin{array}{ccc}
\Sigma_{yy}(1) & \ldots & \Sigma_{yy}(m)\end{array}\right]\nonumber \\
 & \geq & \sigma_{1}\left[\begin{array}{ccc}
\Sigma_{ff}(1) & \ldots & \Sigma_{ff}(m)\end{array}\right]-\sigma_{max}\left[\begin{array}{ccc}
\Sigma_{f\varepsilon}(1) & \ldots & \Sigma_{f\varepsilon}(m)\end{array}\right]\label{eq:proof5-1}\\
 & (or)\nonumber \\
\sigma_{1}\left(M\right) & \geq & \sigma_{1}\left[\begin{array}{ccc}
\Sigma_{f\varepsilon}(1) & \ldots & \Sigma_{f\varepsilon}(m)\end{array}\right]-\sigma_{1}\left[\begin{array}{ccc}
\Sigma_{ff}(1) & \ldots & \Sigma_{ff}(m)\end{array}\right]\nonumber 
\end{eqnarray}
and following the same arguments as above, we get 
\[
\sigma_{1}\left(M\right)\asymp\max\left(\kappa_{max},\,K^{1-\delta}\right).
\]
Now, Frobenius norm of $\Delta M$ is given by, 
\begin{eqnarray}
\left\Vert \Delta M\right\Vert _{F} & = & \left\Vert \begin{array}{ccc}
\Delta\Sigma_{yy}(1) & \ldots & \Delta\Sigma_{yy}(m)\end{array}\right\Vert _{F}\nonumber \\
 & \leq & \left\Vert \Delta\Sigma_{yy}(1)\right\Vert _{F}+\ldots+\left\Vert \Delta\Sigma_{yy}(m)\right\Vert _{F}\label{eq:proof8-3}
\end{eqnarray}
where 
\begin{eqnarray}
\left\Vert \Delta\Sigma_{yy}(l)\right\Vert _{F} & \leq & \sqrt{p}\left\Vert \Delta\Sigma_{f\varepsilon}(l)\right\Vert _{F}+p\left\Vert \Delta\Sigma_{ff}(l)\right\Vert _{F}+\sqrt{p}\left\Vert \Delta\Sigma_{\varepsilon f}(l)\right\Vert _{F}+\left\Vert \Delta\Sigma_{\varepsilon\varepsilon}(l)\right\Vert _{F}\nonumber \\
 & = & O_{P}\left(K^{1-\frac{\delta}{2}}N^{-\frac{1}{2}}\right)+O_{P}\left(K^{1-\delta}N^{-\frac{1}{2}}\right)\nonumber \\
 & \ldots & +O_{P}\left(K^{1-\frac{\delta}{2}}N^{-\frac{1}{2}}\right)+O_{P}\left(KN^{-\frac{1}{2}}\right)\nonumber \\
 & = & O_{P}\left(KN^{-\frac{1}{2}}\right).\label{eq:proof9}
\end{eqnarray}
Thus, substituting (\ref{eq:proof9}) in (\ref{eq:proof8-3}), 
\begin{equation}
\left\Vert \Delta M\right\Vert _{F}=O_{P}\left(KN^{-\frac{1}{2}}\right).\label{eq:proof10}
\end{equation}
Note that $\left\Vert \Delta M\right\Vert _{2}$ also has the same
convergence rate as shown in \cite{lam2011estimation}. Now,
\begin{eqnarray}
\sigma_{p}(M)+\sigma_{max}(\Delta M)\geq\sigma_{p}(\tilde{M}) & \geq & \sigma_{p}(M)-\sigma_{max}(\Delta M)\label{eq:lemm6_fin}
\end{eqnarray}
and 
\begin{equation}
\sigma_{1}\left(M\right)-\sigma_{max}\left(\Delta M\right)\geq\sigma_{1}\left(\tilde{M}\right)\geq\sigma_{1}\left(M\right)-\sigma_{max}\left(\Delta M\right)\label{eq:lemma6_finb}
\end{equation}
upon substituting (\ref{eq:l5_1}), (\ref{eq:l5_1b}) and (\ref{eq:l5_2})
in (\ref{eq:lemm6_fin}) and (\ref{eq:lemma6_finb}), (\ref{eq:l5_3})
and (\ref{eq:l5_3b}) is got.
\end{proof}
~ Now proceeding to prove the Theorem:
\begin{proof}
When using \noun{Hybrid-I}II on $\tilde{M}$, (\ref{eq:M_perturbed-1-1}),
to get a decomposition as given in (\ref{eq:rrqr1-1-1}), the following
inequalities are satisfied: 
\begin{eqnarray}
\sigma_{min}\left(\tilde{R}_{11}^{(p)}\right) & \ge & \frac{\sigma_{p}\left(\tilde{M}\right)}{\sqrt{p\left(K-p+1\right)}}\label{eq:theorem2_1}\\
\mbox{and}\nonumber \\
\sigma_{max}\left(\tilde{R}_{22}^{(p)}\right) & \leq & \sigma_{p+1}\left(\tilde{M}\right)\sqrt{\left(p+1\right)\left(K-p\right)}.\label{eq:theorem2_2}
\end{eqnarray}
Refer to \emph{Section 11 }in\emph{ }\cite{chandrasekaran1994rank}
for a justification of these inequalities. Note \noun{Hybrid-III}
is applied on $\tilde{M}$ assuming a numerical rank $p$. From (\ref{eq:l5_3})
in \emph{lemma }\ref{lem:lemma 6}, under the constraint $\frac{K^{\delta}}{\sqrt{N}}=o(1)$,

\begin{eqnarray}
\sigma_{min}\left(R_{11}^{(p)}\right) & \geq & \begin{cases}
O_{P}\left(\frac{K^{1-\delta}}{\sqrt{p\left(K-p+1\right)}}\right) & ,if\:\kappa_{max}=o\left(K^{1-\delta}\right)\\
O_{P}\left(\frac{\kappa_{min}}{\sqrt{p\left(K-p+1\right)}}\right) & ,if\:o\left(\kappa_{min}\right)=K^{1-\delta}
\end{cases}\\
 & \geq & \begin{cases}
O_{P}\left(K^{0.5-\delta}\right) & ,\kappa_{max}=o\left(K^{1-\delta}\right)\\
O_{P}\left(\frac{\kappa_{min}}{\sqrt{K}}\right) & ,o\left(\kappa_{min}\right)=K^{1-\delta}
\end{cases}.\label{eq:theorem2_3}
\end{eqnarray}
 Equation (\ref{eq:theorem2_2}) and 
\begin{eqnarray*}
\sigma_{p+1}\left(\tilde{M}\right) & \leq & \sigma_{p+1}\left(M\right)+\left\Vert \Delta M\right\Vert _{2},\\
 & \leq & \left\Vert \Delta M\right\Vert _{2}\quad\mbox{as,}\sigma_{p+1}(M)=0
\end{eqnarray*}
implies 
\begin{eqnarray*}
\left\Vert \tilde{R}_{22}^{(p)}\right\Vert _{2} & \leq & \left\Vert \Delta M\right\Vert _{2}O\left(\sqrt{K}\right).
\end{eqnarray*}
Using (\ref{eq:l5_2}), 
\begin{equation}
\left\Vert \tilde{R}_{22}^{(p)}\right\Vert _{2}=O_{P}\left(K^{\frac{3}{2}}N^{-\frac{1}{2}}\right)\label{eq:theorem2_4}
\end{equation}
Hence from (\ref{eq:theorem2_3}) and (\ref{eq:theorem2_4}), under
the constraint 
\[
\frac{K^{1+\delta}}{\sqrt{N}}=o(1),
\]
\[
o_{P}\left(\sigma_{min}\left(\tilde{R}_{11}^{(p)}\right)\right)=\left\Vert \tilde{R}_{22}\right\Vert _{2}
\]
is satisfied.
\end{proof}
Note for the matrix $M$, $\sigma_{min}\left(R_{11}^{(p)}\right)>\left\Vert R_{22}^{(K-p)}\right\Vert _{2}=0$,
implying that $p$ is a valid numerical rank. Nevertheless, it is
quite possible that $\sigma_{min}\left(R_{11}^{(p-1)}\right)>\left\Vert R_{22}^{(K-p-1)}\right\Vert _{2}$,
implying the non-uniqueness of a numerical rank. In the case of $M$
determining $p$ is simple as the last $K-p$ rows of the $R$ matrix
are zeros. In the case of $\tilde{M}$, an estimate $\hat{p}$ that
maximizes the ratio 
\begin{equation}
r_{i}=\frac{\gamma_{i}+\epsilon}{\gamma_{i+1}+\epsilon},\,i=1,2,\dots,K\label{eq:the_ratio}
\end{equation}
where $\epsilon=\frac{\gamma_{1}}{\sqrt{KN}}$ is as mentioned in
(\ref{eq:model_order}) is used. The following theorem shows that
for large values of $N$, $\hat{p}$ is equal to $p$ with a very
high probability.
\begin{thm}
\label{thm:B}Let $M$ as in (\ref{eq:M_def}) be of rank $'p'$ and
let QR decomposition of $\tilde{M}\Pi$ be as in (\ref{eq:rrqr1-1})
where $\Pi$ is the permutation matrix got by \noun{Hybrid-III} decomposition.
Define $\gamma_{i}$ as the $i^{th}$ diagonal value of $\tilde{R}$
in (\ref{eq:rrqr1-1}). Under the assumptions (A1-A7), refer Chapter
\ref{chap:Algorithm} regarding assumptions, the following properties
hold true as $K,\,N\rightarrow\infty$ under the constraint $\frac{K^{\delta}}{\sqrt{N}}=o(1)$
regarding the ratio (\ref{eq:the_ratio}). 

Let $\mathcal{UB}\left(r_{i}\right)$ and $\mathcal{LB}\left(r_{i}\right)$
denote the upper and lower bounds for rate of growth of $r_{i}$as
$K,N\rightarrow\infty$.

Case1.~ for $i=1,\,p\neq1$ ,
\begin{eqnarray}
\mathcal{UB}\left(r_{1}\right) & = & \begin{cases}
O_{P}\left(\sqrt{K}\right), & \kappa_{max}=o\left(K^{1-\delta}\right)\\
O_{P}\left(\frac{\kappa_{max}}{\kappa_{min}}\sqrt{K}\right), & o\left(\kappa_{min}\right)=K^{1-\delta}
\end{cases}\nonumber \\
\mathcal{LB}\left(r_{1}\right) & = & \begin{cases}
O\left(\frac{1}{\sqrt{K}}\right), & \kappa_{max}=o\left(K^{1-\delta}\right)\\
O\left(\frac{\kappa_{max}}{\kappa_{min}\sqrt{K}}\right), & o\left(\kappa_{min}\right)=K^{1-\delta}
\end{cases}\label{eq:Theorem3_case1}
\end{eqnarray}

Case2.~ for $2<i<p$,
\begin{eqnarray}
\mathcal{UB}\left(\{r_{i}\}_{i=2}^{p-1}\right) & = & O_{P}\left(K\right),\nonumber \\
\mathcal{LB}\left(\{r_{i}\}_{i=2}^{p-1}\right) & = & O_{P}\left(K^{-1}\right),\label{eq:Theorem3_case2}
\end{eqnarray}

Case3.~for $i=p$,
\begin{eqnarray}
\mathcal{UB}\left(r_{p}\right) & = & O_{P}\left(NK^{2}\right)\nonumber \\
\mathcal{LB}\left(r_{p}\right) & = & \begin{cases}
O_{P}\left(K^{-(1+\delta)}N^{\frac{1}{2}}\right), & \kappa_{max}=o\left(K^{1-\delta}\right)\\
O_{P}\left(\kappa_{min}N^{\frac{1}{2}}K^{-2}\right), & o\left(\kappa_{min}\right)=K^{1-\delta}
\end{cases}\label{eq:Theorem3_case3}
\end{eqnarray}

Case4.~ for $i>p$,
\begin{eqnarray}
\mathcal{UB}\left(\left\{ r_{i}\right\} _{i>p}\right) & = & \begin{cases}
O_{P}\left(K^{\frac{3}{2}+\delta}\right), & \kappa_{max}=o\left(K^{1-\delta}\right)\\
O_{P}\left(\kappa_{min}^{-1}K^{\frac{5}{2}}\right), & o\left(\kappa_{min}\right)=K^{1-\delta}
\end{cases}\nonumber \\
\mathcal{LB}\left(\left\{ r_{i}\right\} _{i>p}\right) & = & \begin{cases}
O_{P}\left(K^{-\left(\frac{3}{2}+\delta\right)}\right), & \kappa_{max}=o\left(K^{1-\delta}\right)\\
O_{P}\left(\kappa_{min}K^{-\frac{5}{2}}\right), & o\left(\kappa_{min}\right)=K^{1-\delta}
\end{cases}\label{eq:Theorem3_case4}
\end{eqnarray}
\end{thm}
\begin{proof}
The rate of convergences for $\gamma_{i}+\epsilon$ for different
$i$ is derived under the constraint $\frac{K^{\delta}}{\sqrt{N}}=o(1)$
as follows : 

For $i=1$, 
\[
\frac{\sigma_{1}\left(\tilde{M}\right)}{\sqrt{mK}}\leq\gamma_{1}\leq\sigma_{1}\left(\tilde{M}\right)
\]
Thus from (\ref{eq:l5_3b}),
\begin{eqnarray*}
\mathcal{UB}\left(\gamma_{1}\right) & = & O_{P}\left(\max\left(\kappa_{max},\,K^{1-\delta}\right)\right)\\
\mathcal{LB}\left(\gamma_{2}\right) & = & O_{P}\left(\frac{O_{P}\left(\max\left(\kappa_{max},\,K^{1-\delta}\right)\right)}{O\left(\sqrt{K}\right)}\right)
\end{eqnarray*}
 Since, $\epsilon=\frac{\gamma_{1}}{\sqrt{KN}}$ ,
\begin{eqnarray*}
\mathcal{UB}\left(\epsilon\right) & = & \frac{O_{P}\left(\max\left(\kappa_{max},\,K^{1-\delta}\right)\right)}{\sqrt{KN}}\\
\mathcal{LB}\left(\epsilon\right) & = & \frac{O_{P}\left(\max\left(\kappa_{max},\,K^{1-\delta}\right)\right)}{O\left(K\sqrt{N}\right)}
\end{eqnarray*}
 Thus,
\begin{eqnarray}
\mathcal{UB}\left(\gamma_{1}+\epsilon\right) & = & O_{P}\left(\max\left(\kappa_{max},\,K^{1-\delta}\right)\right)\nonumber \\
\mathcal{LB}\left(\gamma_{1}+\epsilon\right) & = & \frac{O_{P}\left(\max\left(\kappa_{max},\,K^{1-\delta}\right)\right)}{O\left(\sqrt{K}\right)}\label{eq:thm3_1}
\end{eqnarray}
For $p>i>2$, from (\ref{eq:L1-1}) and (\ref{eq:l5_3}),
\begin{eqnarray}
\mathcal{UB}\left(\left\{ \gamma_{i}+\epsilon\right\} _{i=2}^{p}\right) & = & \begin{cases}
O_{P}\left(K^{1.5-\delta}\right), & \forall\kappa_{max}=o\left(K^{1-\delta}\right)\\
O_{P}\left(\kappa_{min}\sqrt{K}\right), & \forall o\left(\kappa_{min}\right)=K^{1-\delta}
\end{cases}\nonumber \\
\mathcal{LB}\left(\left\{ \gamma_{i}+\epsilon\right\} _{i=2}^{p}\right) & = & \begin{cases}
\frac{O_{P}\left(K^{1-\delta}\right)}{O\left(\sqrt{K}\right)}, & \forall\kappa_{max}=o\left(K^{1-\delta}\right)\\
\frac{O_{P}\left(\kappa_{min}\right)}{O\left(\sqrt{K}\right)}, & \forall o\left(\kappa_{min}\right)=K^{1-\delta}
\end{cases}\label{eq:thm3_3a}
\end{eqnarray}
For $i>p$, from (\ref{eq:L1-1}) and (\ref{eq:l5_2}),
\begin{eqnarray*}
\frac{\sigma_{p+1}\left(\tilde{M}\right)}{\sqrt{mK-p}} & \leq & \left\{ \gamma_{i}\right\} _{i>p}\leq\sigma_{p+1}\left(\tilde{M}\right)\sqrt{mK-p},\\
0 & \leq & \left\{ \gamma_{i}\right\} _{i>p}\leq\left\Vert \Delta M\right\Vert _{2}\sqrt{mK-p}.
\end{eqnarray*}
thus,
\begin{eqnarray}
\mathcal{UB}\left(\left\{ \gamma_{i}+\epsilon\right\} _{i>p+1}\right) & = & O_{P}\left(K^{\frac{3}{2}}N^{-\frac{1}{2}}\right)\nonumber \\
\mathcal{LB}\left(\left\{ \gamma_{i}+\epsilon\right\} _{i>p+1}\right) & = & \mathcal{LB}\left(\epsilon\right)=\frac{O_{P}\left(\max\left(\kappa_{max},\,K^{1-\delta}\right)\right)}{O\left(K\sqrt{N}\right)}\label{eq:thm3_4}
\end{eqnarray}

Let $\mathcal{U}\left(.\right)$ and $\mathcal{L}\left(.\right)$
denote upper-bound and lower-bound respectively. The ratio (\ref{eq:the_ratio}),
is given by, 
\begin{eqnarray*}
\mathcal{UB}\left(r_{i}\right) & = & \frac{\mathcal{UB}\left(\gamma_{i}+\epsilon\right)}{\mathcal{LB}\left(\gamma_{i+1}+\epsilon\right)}.\\
\mathcal{LB}\left(r_{i}\right) & = & \frac{\mathcal{LB}\left(\gamma_{i}+\epsilon\right)}{\mathcal{UB}\left(\gamma_{i+1}+\epsilon\right)}
\end{eqnarray*}

For $i=1$, From (\ref{eq:thm3_1}), and (\ref{eq:thm3_3a}) we get
(\ref{eq:Theorem3_case1}).

For $p>i>1$, from \ref{eq:thm3_3a}, we get (\ref{eq:Theorem3_case2}).

For $i=p$, 
\[
\frac{\mathcal{LB}\left(\gamma_{p}+\epsilon\right)}{\mathcal{UB}\left(\gamma_{p+1}+\epsilon\right)}\leq r_{p}\leq\frac{\mathcal{UB}\left(\gamma_{p}+\epsilon\right)}{\mathcal{LB}\left(\gamma_{p+1}+\epsilon\right)}
\]
Thus, from (\ref{eq:thm3_3a}) and (\ref{eq:thm3_4}), we get (\ref{eq:Theorem3_case3})
For $i>p,$
\[
\frac{\mathcal{LB}\left(\left\{ \gamma_{i}+\epsilon\right\} _{i>p}\right)}{\mathcal{UB}\left(\left\{ \gamma_{i}+\epsilon\right\} _{i>p}\right)}\leq\left\{ r_{i}\right\} _{i>p}\leq\frac{\mathcal{U}B\left(\left\{ \gamma_{i}+\epsilon\right\} _{i>p}\right)}{\mathcal{L}B\left(\left\{ \gamma_{i}+\epsilon\right\} _{i>p}\right)}
\]
Thus using (\ref{eq:thm3_4}), leads to (\ref{eq:Theorem3_case4})
.
\end{proof}
As per Theorem (\ref{thm:B}), the ratio $r$ will be maximized at
$\hat{p}=p$ with high probability if the lower bounds of \emph{Case3}
grows faster than the upper bound of \emph{Case1, Case2} and \emph{Case4}.
This is satisfied when the rate of increase of $K,\,N\rightarrow\infty$
is constrained by, 
\begin{equation}
\frac{K^{\frac{5}{2}+2\delta}}{N^{\frac{1}{2}}}=o(1).\label{eq:model_condition}
\end{equation}

Note that the QR decomposition based factor modeling is better than
the EVD based approach as the term, $\epsilon=\frac{\gamma_{1}}{\sqrt{KN}}$,
ensures that in the maximization of the ratio $r_{i}$, $i=p$ is
a peak for $r_{i}$ whereas $r_{i}$ at $i=p$ is just a knee point
in the EVD based model order determination.

%% file: includes/Chapter-5.tex
\chapter{\label{chap:Perturbation-Analysis}Perturbation Analysis of $\hat{Q}$}

In this Chapter, $\tilde{M}$ as defined in (\ref{eq:M_perturbed-1})
is treated as the perturbed version of $M$ (\ref{eq:M_def}). Let
\begin{eqnarray}
\Delta M & \triangleq & \tilde{M}-M\nonumber \\
 & \triangleq & \left[\begin{array}{cccc}
\Delta\Sigma_{yy}(1) & \Delta\Sigma_{yy}(2) & \ldots & \Delta\Sigma_{yy}(m)\end{array}\right]\label{eq:Delta_m}
\end{eqnarray}
and $\Delta\Sigma_{yy}(l)$ denote the perturbation in $M$ and $\Sigma_{yy}(l)$
respectively. Also, $\hat{Q}$ got from $\tilde{M}\Pi$ is taken as
the perturbed version of $Q$ got from $M\Pi$, where $\Pi$ is any
permutation matrix. Assumption \emph{(A5) }implies that $\tilde{M\Pi}$
converges to $M\Pi$ as $N\rightarrow\infty$ (for a fixed $K$),
refer to \emph{Theorem 27.4} in \cite{billingsley2008probability}.
Thus $\hat{Q}$ should go to $Q$ as $N\rightarrow\infty.$ Therefore,
the search is for a permutation $\Pi$ that renders $\left\Vert \hat{Q}-Q\right\Vert $
reasonably small for a finite $N$. As mentioned earlier, here a unique
$\Pi$ is determined using\emph{ }\noun{Hybrid-I}\emph{ RRQR }algorithm.
This choice is justified in \emph{Section }\ref{sec:Perturbation-sec1}.
In \emph{Section \ref{sec:Convergence-sec2}, }the convergence rate
of $\left\Vert \hat{Q}-Q\right\Vert _{F}$ for the proposed algorithm
along with a measure of the convergence rate of factors, $\hat{f}_{n}$
is presented. 

\section{\label{sec:Perturbation-sec1}Perturbation analysis of QR decomposition:}

Let, $\tilde{M}_{p}$ denote the first $p$ columns selected from
$\tilde{M}\Pi$ as mentioned in (\ref{eq:rrqr1-1}) and $M_{p}$,
$\Delta M_{p}$ denote the first $p$ columns selected from $M\Pi$
and $\Delta M\Pi$ respectively. Thus, 
\begin{equation}
\tilde{M}_{p}=M_{p}+\Delta M_{p}.\label{eq:M_p_defn}
\end{equation}
Consider the following two lemmas taken from \cite{chang1997perturbation}: 
\begin{lem}
\label{lem:peturbation_condition}Let $M_{p}\in\mathbb{R}^{K\times p}$
be of full column rank, $p$, then: $M_{p}+\Delta M_{p}$ is also
of rank $p$ if, 
\begin{equation}
\sigma_{min}(M_{p})>\sigma_{max}(P_{1}\Delta M_{p})\label{eq:l1_1}
\end{equation}
where, $P_{1}$ is the projection matrix that projects $\Delta M_{p}$
to $\mathcal{R}(M_{p})$.
\end{lem}
\begin{proof}
Refer to \cite{chang1997perturbation} for a proof.
\end{proof}
\begin{lem}
\label{lem:peturbation_quantity-1}Let $M_{p}\in\mathbb{R}^{K\times p}$
be a matrix of rank $p$, with the $QR$ decomposition $M_{p}=Q_{1}R$
and $\Delta M_{p}$ be the perturbation on $M_{p}$ such that $M_{p}+\Delta M_{p}$
is also of rank $p$. $M_{p}+\Delta M_{p}$ has unique $QR$ factorization:
\begin{eqnarray*}
M_{p}+\Delta M_{p} & = & \left(Q_{1}+\Delta Q_{1}\right)\left(R+\Delta R\right)
\end{eqnarray*}
where, 
\begin{eqnarray}
\left\Vert \Delta Q_{1}\right\Vert _{F} & \leq & \sqrt{2}\frac{\left\Vert \Delta M_{p}\right\Vert _{F}}{\sigma_{min}(M_{p})}+\mathcal{O}\left(\epsilon^{2}\right)\label{eq:l2_1-2}
\end{eqnarray}
and 
\[
\epsilon=\frac{\left\Vert \Delta M_{p}\right\Vert _{F}}{\left\Vert M_{p}\right\Vert _{2}}.
\]
\end{lem}
\begin{proof}
Refer to \cite{chang1997perturbation} for a proof.
\end{proof}
Note that in the current context $\left\Vert \Delta Q_{1}\right\Vert _{F}=\left\Vert \hat{Q}-Q\right\Vert _{F}.$
Lemmas 1 and 2 suggest choosing a permutation matrix $\Pi$ such that
$\sigma_{min}\left(M_{p}\right)$ is maximized. Maximizing $\sigma_{min}(M_{p})$
is equivalent to maximizing $\sigma_{min}\left(\tilde{M}_{p}\right)$
as 
\[
\sigma_{min}\left(M_{p}\right)\geq\sigma_{min}\left(\tilde{M}_{p}\right)-\sigma_{max}\left(\Delta M_{p}\right).
\]
\noun{Hybrid-I} RRQR algorithm \cite{chandrasekaran1994rank} selects
a $\Pi$ by trying\footnote{To get a $\Pi$ that maximizes $\sigma_{min}\left(\tilde{M}_{p}\right)$
is a NP-hard problem and applying a brute force method would require
finding $\sigma_{min}\left(\tilde{M}_{p}\right)$ for $^{Km}C_{p}$
matrices where $Km$ is the total number of columns in $\tilde{M}$.
\noun{Hybrid-I} RRQR is a computationally simpler technique that generates
a permutation $\Pi$ which satisfies (\ref{eq:thm4_1})} to maximize $\sigma_{min}\left(\tilde{M}_{p}\right)$. The effect
of using \noun{Hybrid-I} algorithm of RRQR is presented in \emph{Theorem
4}
\begin{thm}
\label{lem:rrqr}Let $\tilde{M}\in\mathbb{R}^{K\times mK}$ . Then
the \noun{Hybrid-I} RRQR algorithm could be used to select a $\Pi$
such that the first $'p'$ columns from $\tilde{M}\Pi$, $\tilde{M}_{p}$,
satisfies 
\begin{eqnarray}
\sigma_{p}\left(\tilde{M}\right)\geq\sigma_{min}(\tilde{M}_{p}) & \geq & \frac{\sigma_{p}(\tilde{M})}{\sqrt{p(K-p+1)}}\label{eq:thm4_1}
\end{eqnarray}
\end{thm}
\begin{proof}
Applying Hybrid I algorithm \cite{chandrasekaran1994rank} to $\tilde{M}$,
leads to a permutation matrix $\Pi$ given by (\ref{eq:rrqr1-1}),
where,
\begin{equation}
\sigma_{min}(\tilde{R}_{11})\geq\frac{\sigma_{p}(\tilde{R})}{\sqrt{p(K-p+1)}}\label{eq:l3_2}
\end{equation}
Refer to \cite{chandrasekaran1994rank} for a complete proof. Select
$\tilde{M}_{p}$ as the first $p$ columns of $\tilde{M}\Pi$. QR
decomposition of $\tilde{M}_{p}$ gives, 
\[
\tilde{M}_{p}=\mathcal{Q}_{p}\begin{array}{c}
\tilde{R}_{11}\end{array}.
\]
Hence, 
\begin{equation}
\sigma_{min}\left(\tilde{M}_{p}\right)=\sigma_{min}\left(\tilde{R}_{11}\right).\label{eq:l3_3}
\end{equation}
Interlacing property of singular values, \cite{golub2012matrix,li2012generalized},
ensures,
\begin{equation}
\sigma_{p}\left(\tilde{M}\right)\geq\sigma_{min}\left(\tilde{M}_{p}\right).\label{eq:l4_4-1}
\end{equation}
 Thus, (\ref{eq:l3_2}), (\ref{eq:l3_3}) and (\ref{eq:l4_4-1}) leads
to (\ref{eq:thm4_1}).
\end{proof}
Henceforth, the term RRQR is used to refer the \noun{Hybrid-I} RRQR
algorithm for simplicity. In the next subsection, convergence of $\hat{Q}$
to $Q$, obtained by using the RRQR on $\tilde{M}$ and $M$ respectively,
as $N,\,K\rightarrow\infty$ is investigated.

\section{\label{sec:Convergence-sec2}Convergence Analysis of $\hat{Q}$}

The rate of convergence of $\hat{Q}$ to $Q$ as $N,\,K\rightarrow\infty$
is summarized in the following theorem: 
\begin{thm}
Under the conditions (A1-A7), refer Chapter \ref{chap:Algorithm}
regarding assumptions, with the optimistic assumption that the upper
bound is achieved in (\ref{eq:thm4_1}) and with rate of increase
of $K$ and $N$ constrained by $\frac{K^{\delta}}{N}=o(1)$, the
following convergence rate for $\hat{Q}$ is got:
\begin{equation}
\left\Vert \hat{Q}-Q\right\Vert _{F}\asymp\begin{cases}
O_{P}\left(K^{\frac{\delta}{2}}N^{-\frac{1}{2}}\right), & \kappa_{max}=o\left(K^{1-\delta}\right)\\
O_{P}\left(\kappa_{min}^{-1}K^{1-\frac{\delta}{2}}N^{-\frac{1}{2}}\right), & K^{1-\delta}=o\left(\kappa_{min}\right)
\end{cases}\label{eq:conver_optimistic}
\end{equation}
and in the worst case scenario of (\ref{eq:thm4_1}), where the lower-bound
of (\ref{eq:thm4_1}) is taken, under the constraint \textup{$\frac{K^{1+\delta}}{N}=o(1)$,}
\begin{equation}
\left\Vert \hat{Q}-Q\right\Vert _{F}\asymp\begin{cases}
O_{P}\left(K^{\frac{1+\delta}{2}}N^{-\frac{1}{2}}\right), & \kappa_{max}=o\left(K^{1-\delta}\right)\\
O_{P}\left(\kappa_{min}^{-1}K^{\frac{3-\delta}{2}}N^{-\frac{1}{2}}\right), & K^{1-\delta}=o\left(\kappa_{min}\right)
\end{cases}\label{eq:conver-worstcase}
\end{equation}
\end{thm}
\begin{proof}
The perturbation on $Q$ is given by (\ref{eq:l2_1-2}). Let, $\tilde{M}_{p}$
denote the first $p$ columns selected from $\tilde{M}\Pi$ as mentioned
in (\ref{eq:rrqr1-1}) and $M_{p}$, $\Delta M_{p}$ denote the first
$p$ columns selected from $M\Pi$ and $\Delta M\Pi$ respectively.
Then, 
\begin{equation}
\tilde{M}_{p}=M_{p}+\Delta M_{p}.\label{eq:M_p_defn-1}
\end{equation}
According to (\ref{eq:l1_1}), The following equation needs to be
satisfied:
\[
\sigma_{min}\left(\tilde{M}_{p}\right)>\sigma_{max}\left(\Delta M_{p}\right).
\]
According to \emph{(\ref{eq:thm4_1}),} 
\begin{eqnarray}
\sigma_{p}(\tilde{M}) & \geq & \sigma_{min}\left(\tilde{M}_{p}\right)\geq\frac{\sigma_{p}(\tilde{M})}{\sqrt{p\left(K-p+1\right)}}\label{eq:proof1_2}
\end{eqnarray}
Thus according to (\ref{eq:l5_3}) and (\ref{eq:proof1_2}), The best
estimate of $\sigma_{p}\left(\tilde{M}_{p}\right)$ is given by, 
\begin{equation}
\sigma_{min}\left(\tilde{M}_{p}\right)\asymp_{P}\begin{cases}
K^{1-\delta} & ,if\:\kappa_{max}=o\left(K^{1-\delta}\right)\\
\kappa_{min} & ,if\:o\left(\kappa_{min}\right)=o\left(K^{1-\delta}\right)
\end{cases},\:\forall\frac{K^{\delta}}{\sqrt{N}}=o(1).\label{eq:proof11a}
\end{equation}
and in the worst case scenario, 

\begin{equation}
\sigma_{min}\left(\tilde{M}_{p}\right)\asymp_{P}\begin{cases}
\frac{K^{1-\delta}}{\sqrt{K}} & ,if\:\kappa_{max}=o\left(K^{1-\delta}\right)\\
\frac{\kappa_{min}}{\sqrt{K}} & ,if\:o\left(\kappa_{min}\right)=o\left(K^{1-\delta}\right)
\end{cases},\:\forall\frac{K^{\delta}}{\sqrt{N}}=o(1).\label{eq:proof11b}
\end{equation}
Note that $\Delta M_{p}$ are $p$ columns selected from $\Delta M\Pi$.
Using (\ref{eq:proof8-3}) and (\ref{eq:proof9}) it can be inferred
that $O\left(\left\Vert \Delta M\right\Vert _{F}\right)=O\left(\Delta\Sigma_{yy}(l)\right)$.
Let $\Delta\Sigma_{yy}^{\left(p\right)}(l)$ be $p$ columns selected
from $\Delta\Sigma_{yy}(l)$. Then, 
\begin{eqnarray}
\left\Vert \Delta M_{p}\right\Vert _{F} & \asymp & O\left(\Delta\Sigma_{yy}^{\left(p\right)}(l)\right)\nonumber \\
 & \leq & O\left(\left\Vert \Delta\Sigma_{f\varepsilon}^{(p)}(l)\right\Vert _{F}\right)+O\left(\left\Vert \Delta\Sigma_{ff}^{(p)}(l)\right\Vert _{F}\right)\ldots\label{eq:delta_mp1}\\
 & \ldots & +O\left(\left\Vert \Delta\Sigma_{\varepsilon f}^{(p)}(l)\right\Vert _{F}\right)+O\left(\left\Vert \Delta\Sigma_{\varepsilon\varepsilon}^{(p)}\right\Vert _{F}\right)\nonumber 
\end{eqnarray}
Note that $\left\Vert \Delta\Sigma_{f\varepsilon}^{(p)}(l)\right\Vert _{F}\leq\left\Vert \Delta\Sigma_{f\varepsilon}(l)\right\Vert _{F}$,
$\left\Vert \Delta\Sigma_{\varepsilon f}^{(p)}(l)\right\Vert _{F}\leq\left\Vert \Delta\Sigma_{\varepsilon f}(l)\right\Vert _{F}$
and $\left\Vert \Delta\Sigma_{ff}^{(p)}(l)\right\Vert _{F}\leq\left\Vert \Delta\Sigma_{ff}(l)\right\Vert _{F}$.
From the proof of (\ref{eq:l4_4}), it can be seen that $\Delta\Sigma_{\varepsilon\varepsilon}(l)$
is an uniform perturbation in all columns. Thus, let $\Delta\Sigma_{\varepsilon\varepsilon}^{(p)}$
be the perturbation in $p$ columns due to $\Delta\Sigma_{\varepsilon\varepsilon}(l)$
then 
\begin{eqnarray}
\left\Vert \Delta\Sigma_{\varepsilon\varepsilon}^{(p)}\right\Vert _{F} & \asymp_{P} & \left(Kp\right)^{\frac{1}{2}}N^{\frac{1}{2}}\nonumber \\
 & \asymp_{P} & K^{\frac{1}{2}}N^{\frac{1}{2}}.\label{eq:delta_mp2}
\end{eqnarray}
Thus substituting (\ref{eq:l4_1}), (\ref{eq:l4_2}), (\ref{eq:l4_3})
and (\ref{eq:delta_mp2}) in (\ref{eq:delta_mp1}) leads to 
\begin{eqnarray}
\left\Vert \Delta M_{p}\right\Vert _{F} & \leq & O_{P}\left(K^{1-\frac{\delta}{2}}N^{-\frac{1}{2}}\right)+O_{P}\left(K^{1-\delta}N^{-\frac{1}{2}}\right)+O_{P}\left(K^{\frac{1}{2}}N^{-\frac{1}{2}}\right)\nonumber \\
 & = & O_{P}\left(K^{1-\frac{\delta}{2}}N^{-\frac{1}{2}}\right).\label{eq:proof1_6}
\end{eqnarray}
Thus in the best case scenario (\ref{eq:proof11a}), under the constraint
$\frac{K^{\delta}}{\sqrt{N}}=o(1)$, \emph{(\ref{eq:l1_1})} is satisfied
with high probability as 
\[
o\left(\sigma_{min}\left(\tilde{M}_{p}\right)\right)=\left\Vert \Delta M_{p}\right\Vert _{F}.
\]
Similarly in the worst case scenario, from (\ref{eq:proof11b}) and
(\ref{eq:proof1_6}), it can be shown that under the condition $\frac{K^{\frac{1+\delta}{2}}}{\sqrt{N}}=o(1)$,
\emph{lemma 1} is satisfied with high probability.

In order to estimate $\left\Vert \Delta Q\right\Vert _{F}$ as given
by (\ref{eq:l2_1-2}), $\sigma_{min}\left(M_{p}\right)$ has to be
estimated. Note 
\begin{equation}
\sigma_{min}\left(\tilde{M}_{p}\right)+\sigma_{max}\left(\Delta M_{p}\right)\geq\sigma_{min}\left(M_{p}\right)\geq\sigma_{min}\left(\tilde{M}_{p}\right)-\sigma_{max}\left(\Delta M_{p}\right).\label{eq:minM_p}
\end{equation}
Now, $\sigma_{max}\left(\Delta M_{p}\right)\leq\left\Vert \Delta M_{p}\right\Vert _{F}$
thus substituting, (\ref{eq:proof11a}) and (\ref{eq:proof1_6}) in
(\ref{eq:minM_p}), for the optimistic scenario,
\begin{equation}
\sigma_{min}\left(M_{p}\right)\asymp_{P}\begin{cases}
K^{1-\delta} & ,if\:\kappa_{max}=o\left(K^{1-\delta}\right)\\
\kappa_{min} & ,if\:o\left(\kappa_{min}\right)=o\left(K^{1-\delta}\right)
\end{cases},\:\forall\frac{K^{\delta}}{N}=o(1).\label{eq:proof11a-1}
\end{equation}
and in the worst case scenario substituting (\ref{eq:proof11b}) and
(\ref{eq:proof1_6}) in (\ref{eq:minM_p}), 

\begin{equation}
\sigma_{min}\left(M_{p}\right)\asymp_{P}\begin{cases}
\frac{K^{1-\delta}}{\sqrt{K}} & ,if\:\kappa_{max}=o\left(K^{1-\delta}\right)\\
\frac{\kappa_{min}}{\sqrt{K}} & ,if\:o\left(\kappa_{min}\right)=o\left(K^{1-\delta}\right)
\end{cases},\:\forall\frac{K^{1+\delta}}{N}=o(1).\label{eq:proof11b-1}
\end{equation}
In the optimistic viewpoint, using (\ref{eq:proof11a-1}) and (\ref{eq:proof1_6})
in (\ref{eq:l2_1-2}) leads to (\ref{eq:conver_optimistic}). Similarly
for the worst case scenario, substituting (\ref{eq:proof1_6}) and
(\ref{eq:proof11b-1}) in (\ref{eq:l2_1-2}), leads to (\ref{eq:conver-worstcase}). 
\end{proof}
The following theorem determines the convergence rate of the estimated
factors. This is similar to \emph{Theorem 3 }of \cite{lam2011estimation}.
Let 
\begin{eqnarray}
Y & \triangleq & \left[\begin{array}{cccc}
y_{1} & y_{2} & \ldots & y_{N}\end{array}\right]\nonumber \\
F & \triangleq & \left[\begin{array}{cccc}
f_{1} & f_{2} & \ldots & f_{N}\end{array}\right]\label{eq:equation for F}\\
E & \triangleq & \left[\begin{array}{cccc}
\varepsilon_{1} & \varepsilon_{2} & \ldots & \varepsilon_{N}\end{array}\right].\nonumber 
\end{eqnarray}
 The convergence of the factors are got as a measure of the convergence
of the Root Mean Square Error (RMSE) given by, $\left(KN\right)^{-0.5}\left\Vert \hat{Q}\hat{F}-QF\right\Vert _{F}$,
where $\hat{F}$ is the estimated factors.
\begin{thm}
Under the conditions (A1-A7), refer Chapter \ref{chap:Algorithm}
regarding assumptions,

\begin{equation}
\left(KN\right)^{-0.5}\left\Vert \hat{Q}\hat{F}-QF\right\Vert _{F}\leq O_{P}\left(K^{-\frac{\delta}{2}}\left\Vert \hat{Q}-Q\right\Vert _{F}+K^{-\frac{1}{2}}\right)
\end{equation}
 
\end{thm}
\begin{proof}
The RMSE error of the estimated $\hat{Y}=\hat{Q}\hat{F}$ is given
by, $KN^{-\frac{1}{2}}\left\Vert \hat{Q}\hat{F}-QF\right\Vert _{F}$
. Now, 
\begin{eqnarray*}
\left[\hat{Q}\hat{F}-QF\right] & = & \left[\hat{Q}\hat{Q}^{T}QF-\hat{Q}\hat{Q}^{T}E-QF\right]\\
\, & = & \left[\hat{Q}\hat{Q}^{T}-I\right]QF-\hat{Q}\left[\hat{Q}-Q\right]^{T}E+\hat{Q}Q^{T}E\\
 & = & K_{1}+K_{2}+K_{3}
\end{eqnarray*}
where, $K_{1}=\left[\hat{Q}\hat{Q}^{T}-I\right]QF$, $K_{2}=\hat{Q}\left[\hat{Q}-Q\right]^{T}E$,
$K_{3}=\hat{Q}Q^{T}E$. Now,
\begin{eqnarray*}
K_{1} & = & \left[\hat{Q}\hat{Q}^{T}-I+QQ^{T}-QQ^{T}\right]QF\\
\, & = & \left[\hat{Q}\hat{Q}^{T}-QQ^{T}\right]QF-\left[I-QQ^{T}\right]QF
\end{eqnarray*}
Note, $\left[I-QQ^{T}\right]$ is a projection matrix onto the null
space of $Q$ hence, $\left[I-QQ^{T}\right]QF=0$. Hence,
\begin{equation}
K_{1}=\left[\hat{Q}\hat{Q}^{T}-QQ^{T}\right]QF\label{eq:proof2_1}
\end{equation}
Now using $\hat{Q}=Q+\Delta Q$, 
\begin{eqnarray*}
\left[\hat{Q}\hat{Q}^{T}-QQ^{T}\right] & = & \left[\Delta Q\Delta Q^{\top}+Q\Delta Q^{\top}+\Delta QQ^{\top}\right]
\end{eqnarray*}
hence, 
\begin{eqnarray*}
\left\Vert \hat{Q}\hat{Q}^{T}-QQ^{T}\right\Vert _{F} & \leq & \left\Vert \Delta Q\right\Vert _{F}^{2}+\left\Vert Q\Delta Q^{\top}\right\Vert _{F}
\end{eqnarray*}
Thus, 
\[
\left\Vert K_{1}\right\Vert _{F}=O\left(\left\Vert \Delta Q\right\Vert _{F}\left\Vert F\right\Vert _{F}\right)
\]
Substituting (\ref{eq:l4_5}), 
\begin{eqnarray}
\left\Vert \hat{Q}\hat{Q}^{T}-QQ^{T}\right\Vert _{F} & \leq & O\left(\left\Vert \Delta Q\right\Vert _{F}\right)O\left(K^{\frac{1+\delta}{2}}N^{\frac{1}{2}}\right)\label{eq:proof2_2-1}
\end{eqnarray}
Now, 
\begin{eqnarray*}
\left\Vert K_{2}\right\Vert _{F} & = & \left\Vert \hat{Q}\left[\hat{Q}-Q\right]^{T}E\right\Vert _{F}\\
 & \leq & \left\Vert \Delta Q^{T}E\right\Vert _{F}\leq\left\Vert \Delta Q^{T}\right\Vert _{F}\left\Vert E\right\Vert _{2}
\end{eqnarray*}
also, 
\begin{eqnarray*}
\left\Vert K_{3}\right\Vert _{F} & = & \left\Vert \hat{Q}Q^{T}E\right\Vert _{F}\\
 & \leq & \left\Vert Q^{T}E\right\Vert _{F}\\
 & \leq & \left(\sum_{n=1}^{N}\sum_{j=1}^{j=p}\left(q_{j}^{T}\varepsilon_{n}\right)^{2}\right)^{0.5}
\end{eqnarray*}
Now consider the Random variable $q_{j}^{\top}\varepsilon_{n}$ where
$\mathcal{E}\left\{ q_{j}^{\top}\varepsilon_{n}\right\} =0$ and $var\left(q_{j}^{\top}\varepsilon_{n}\right)=q_{j}\Sigma_{\varepsilon\varepsilon}q_{j}^{\top}\leq\sigma_{max}\left(\Sigma_{\varepsilon\varepsilon}\right)=c<\infty$
where $c$ is a constant independent of $K,\,N$ according to assumption
\emph{(A4).} Thus, 
\begin{eqnarray*}
\left\Vert K_{3}\right\Vert _{F} & \leq & p^{\frac{1}{2}}N^{\frac{1}{2}}O_{P}(1).
\end{eqnarray*}
Since $\left\Vert \Delta Q^{T}\right\Vert _{F}=o_{P}(1)$ and $\left\Vert Q\right\Vert _{F}=\sqrt{p}$,
$\left\Vert K_{2}\right\Vert _{F}$ is dominated by $\left\Vert K_{3}\right\Vert _{F}$
in probability. Hence, $\left(KN\right)^{-\frac{1}{2}}\left\Vert \hat{Q}\hat{f}-Qf\right\Vert _{F}\leq O_{P}\left(K^{-\frac{\delta}{2}}\left\Vert \hat{Q}-Q\right\Vert _{F}+K^{-\frac{1}{2}}\right)$.
\end{proof}

%% file: includes/Chapter-6.tex
\chapter{\label{chap:Comparison-with-the}Comparison with the EVD Method}

In this Chapter, the proposed algorithm is compared against the EVD
based algorithm proposed in \cite{lam2011estimation} with regards
to the asymptotic properties derived in the previous Chapter.

The convergence of the estimate, $\hat{Q}^{EVD}$ of $Q$ got using
the EVD algorithm and the corresponding expressions for rate of convergence
are presented in \cite{lam2011estimation}. The following theorem
is taken from \cite{lam2011estimation}.
\begin{thm}
Under the conditions A2-A7 and (A8),

$(A8)$ ~~~~~~~~~There should be at least one $\Sigma_{xx}(l)$
of full rank $p$. 

and all the Eigen values of $S$ (\ref{Eqn10}) being distinct the
following convergence rates hold, 
\begin{eqnarray}
\left\Vert \hat{Q}^{EVD}-Q\right\Vert _{2} & \asymp & \begin{cases}
O_{P}\left(K^{\delta}N^{-\frac{1}{2}}\right), & \begin{array}{c}
\kappa_{max}=o\left(K^{1-\delta}\right)\\
K^{\delta}N^{-\frac{1}{2}}=o(1)
\end{array}\\
\; & \;\\
O_{P}\left(\kappa_{min}^{-2}\kappa_{max}KN^{-\frac{1}{2}}\right), & \begin{array}{c}
K^{1-\delta}=o\left(\kappa_{min}\right)\\
\kappa_{min}^{-2}\kappa_{max}KN^{-\frac{1}{2}}=o(1)
\end{array}
\end{cases}\label{eq:EVD_convergence}
\end{eqnarray}
\end{thm}
\begin{proof}
Refer to \cite{lam2011estimation} for a proof. 
\end{proof}
In the best case scenario, (\ref{eq:conver_optimistic}), the ratio
of the convergence rates of the EVD algorithm vs. the proposed RRQR
algorithm is given by,
\begin{equation}
\frac{\left\Vert \hat{Q}^{\left(EVD\right)}-Q\right\Vert _{2}}{\left\Vert \hat{Q}^{\left(RRQR\right)}-Q\right\Vert _{F}}\asymp\begin{cases}
K^{\frac{\delta}{2}} & ,\kappa_{max}=o\left(K^{1-\delta}\right)\\
\frac{\kappa_{max}}{\kappa_{min}}K^{\frac{\delta}{2}} & ,o\left(\kappa_{min}\right)=K^{1-\delta}
\end{cases},\label{eq:bestcase_ratio}
\end{equation}
and in the worst case scenario, (\ref{eq:conver-worstcase}), the
ratio is given by, 
\begin{equation}
\frac{\left\Vert \hat{Q}^{\left(EVD\right)}-Q\right\Vert _{2}}{\left\Vert \hat{Q}^{\left(RRQR\right)}-Q\right\Vert _{F}}\asymp\begin{cases}
K^{-\left(\frac{1-\delta}{2}\right)} & ,\kappa_{max}=o\left(K^{1-\delta}\right)\\
\frac{\kappa_{max}}{\kappa_{min}\sqrt{K}} & ,o\left(\kappa_{min}\right)=K^{1-\delta}
\end{cases}.\label{eq:worstcase_ratio}
\end{equation}
It is evident from (\ref{eq:bestcase_ratio}), the proposed algorithm
performs better than the EVD algorithm in the best case scenario.
However in the worst case scenario (\ref{eq:worstcase_ratio}), it
is evident that EVD Algorithm would perform better. To probe this
comparison further, consider the following example:
\begin{example}
\label{example 7.1}Consider a data model as in (\ref{eq:1.2}), with
$p=2$ and the factors given by 
\begin{equation}
x_{n}\triangleq\left[\begin{array}{cc}
x_{n}^{(1)} & x_{n}^{(2)}\end{array}\right]^{\top}\label{eq:x_ndef-1}
\end{equation}
 where,
\begin{eqnarray}
x_{n}^{(1)} & = & \tilde{e}_{n}^{(1)}+\alpha_{1}\tilde{e}_{n-1}^{(1)},\nonumber \\
x_{n}^{(2)} & = & \tilde{e}_{n}^{(2)}+\alpha_{2}\tilde{e}_{n-2}^{(2)}\label{eq:x_n_factors-1}
\end{eqnarray}
with $\tilde{e}_{n}^{(1)}$ and $\tilde{e}_{n}^{(2)}$ being $N(0,1)$
iid process independent of the measurement noise, $\varepsilon_{n}\sim N\left(0,\,I\right)$.
The factor loading matrix $H=\left[\begin{array}{cc}
h_{1} & h_{2}\end{array}\right]$ is a $K\times2$ matrix, where $h_{1}$ and $h_{2}$ are orthogonal
with$\left\Vert h_{1}\right\Vert _{2}^{2}\asymp K^{1-\delta_{1}}$
and $\left\Vert h_{2}\right\Vert _{2}^{2}\asymp K^{1-\delta_{2}}$
where $0\leq\delta_{1}<\delta_{2}\leq1$. This is similar to the factor
strength assumption taken in \cite{C1}. 
\end{example}
Let the $QR$ decomposition of $H$ be given by, 
\[
H=\left[\begin{array}{cc}
Q_{1} & Q_{2}\end{array}\right]\left[\begin{array}{c}
R_{11}\\
0
\end{array}\right]
\]
where $Q_{1}$ is the $K\times2$ matrix that spans the column space
of $H$, $Q_{2}$ is the $K\times(K-2)$ matrix which spans the perpendicular
space of $H$ and 
\[
R_{11}=\left[\begin{array}{cc}
\left\Vert h_{1}\right\Vert _{2} & 0\\
0 & \left\Vert h_{2}\right\Vert _{2}
\end{array}\right].
\]
Note that, since $x_{n}$ is independent of $\varepsilon_{n}$, 
\[
\Sigma_{yy}(l)=Q_{1}R_{11}\Sigma_{xx}(l)R_{11}^{\top}Q_{1}^{\top}
\]
Thus,

\begin{eqnarray}
\Sigma_{yy}(1) & = & Q_{1}\left[\begin{array}{cc}
\left\Vert h_{1}\right\Vert _{2}^{2} & 0\\
0 & 0
\end{array}\right]Q_{1}^{\top}\label{eq:eg1_1-1}
\end{eqnarray}
\begin{eqnarray}
\Sigma_{yy}(2) & = & Q_{1}\left[\begin{array}{cc}
0 & 0\\
0 & \left\Vert h_{2}\right\Vert _{2}^{2}
\end{array}\right]Q_{1}^{\top}\label{eq:eg1_2-1}
\end{eqnarray}

Further more, it is assumed that $o\left(K^{1-\delta_{2}}\right)=KN^{-\frac{1}{2}}$.
It can be shown that 
\begin{eqnarray}
\left\Vert \hat{Q}^{EVD}-Q\right\Vert _{2} & \asymp & \frac{\left\Vert h_{1}\right\Vert _{2}^{2}}{\left\Vert h_{2}\right\Vert _{2}^{4}}O_{P}\left(KN^{-\frac{1}{2}}\right)\label{eq:The_convergence_EVD}\\
\frac{O_{P}\left(KN^{-\frac{1}{2}}\right)}{\left\Vert h_{2}\right\Vert _{2}^{2}}\leq\left\Vert \hat{Q}^{RRQR}-Q\right\Vert _{F} & \leq & \frac{O_{P}\left(K^{\frac{3}{2}}N^{-\frac{1}{2}}\right)}{\left\Vert h_{2}\right\Vert _{2}^{2}}\label{eq:conv_diff_fac}
\end{eqnarray}

\begin{proof}
In \cite{lam2011estimation}, it is seen that when using the EVD method,
\begin{equation}
\left\Vert \hat{Q}-Q\right\Vert _{2}\propto\frac{\left\Vert \Delta\left(MM^{\top}\right)\right\Vert _{2}}{\sigma_{p}\left(MM^{\top}\right)}.\label{eq:EVD_conv}
\end{equation}
For the current example, from (\ref{eq:eg1_1-1}) and (\ref{eq:eg1_2-1}),
\begin{equation}
\sigma_{p}\left(MM^{\top}\right)=\left\Vert h_{2}\right\Vert _{2}^{4}.\label{eq:eg_pr1}
\end{equation}
Also it holds that, 
\begin{eqnarray}
\left\Vert \Delta\left(MM^{\top}\right)\right\Vert _{2} & \leq & \left\Vert M\Delta M^{\top}\right\Vert _{2}+\left\Vert \Delta MM^{\top}\right\Vert _{2}+\left\Vert \Delta M\Delta M^{\top}\right\Vert _{2}\nonumber \\
 & = & O\left(\left\Vert M\Delta M^{\top}\right\Vert _{2}\right)\nonumber \\
 & = & O\left(\left\Vert M\right\Vert _{2}\left\Vert \Delta M\right\Vert _{2}\right)\label{eq:EVD_deltaM}
\end{eqnarray}
and substituting from (\ref{eq:l5_2}) and (\ref{eq:eg1_1-1}) in
(\ref{eq:EVD_deltaM}) leads to, 
\begin{eqnarray}
\left\Vert \Delta\left(MM^{\top}\right)\right\Vert _{2} & \asymp & \left\Vert h_{1}\right\Vert _{2}^{2}O_{P}\left(KN^{-\frac{1}{2}}\right).\label{eq:eg_pr2}
\end{eqnarray}
Thus, substituting (\ref{eq:eg_pr1}) and (\ref{eq:eg_pr2}) in (\ref{eq:EVD_conv})
the convergence rate for the EVD method is got as (\ref{eq:The_convergence_EVD}).
The convergence rate for QR decomposition is given by (\ref{eq:l2_1-2}).
Note from (\ref{eq:eg1_2-1}), 
\[
\sigma_{2}\left(M\right)=\left\Vert h_{2}\right\Vert _{2}^{2}.
\]
Now, from (\ref{eq:l5_2}) and (\ref{eq:lemm6_fin}), under the condition
$o\left(\left\Vert h_{2}\right\Vert _{2}^{2}\right)=KN^{-\frac{1}{2}}$,
\[
\sigma_{2}\left(\tilde{M}\right)\asymp_{P}\left\Vert h_{2}\right\Vert _{2}^{2}.
\]
and from (\ref{eq:thm4_1}), 
\[
\left\Vert h_{2}\right\Vert _{2}^{2}\geq\sigma_{2}\left(\tilde{M}_{2}\right)\geq\frac{\left\Vert h_{2}\right\Vert _{2}^{2}}{\sqrt{K}}.
\]
Thus, from (\ref{eq:proof1_6}) and (\ref{eq:minM_p}), 
\[
\left\Vert h_{2}\right\Vert _{2}^{2}\geq\sigma_{2}\left(M_{2}\right)\geq\frac{\left\Vert h_{2}\right\Vert _{2}^{2}}{\sqrt{K}}
\]
finally, applying (\ref{eq:proof1_6}) and the upper and lower bounds
for $\sigma_{2}\left(M_{2}\right)$ in (\ref{eq:l2_1-2}), (\ref{eq:conv_diff_fac})
is obtained.
\end{proof}
Thus the ratio of the convergence rates is given by, 
\begin{equation}
O\left(K^{\delta_{2}-\delta_{1}}\right)\leq\frac{\left\Vert \hat{Q}^{EVD}-Q\right\Vert _{2}}{\left\Vert \hat{Q}^{RRQR}-Q\right\Vert _{F}}\leq O_{P}\left(K^{\delta_{2}-\delta_{1}-\frac{1}{2}}\right)\label{eq:example_ratio_comp}
\end{equation}
Thus, In the best case scenario, the proposed algorithm is better
by a factor of $O\left(K^{\delta_{2}-\delta_{1}}\right)$ and in the
worst case scenario the convergence rate is $O_{P}\left(K^{\delta_{2}-\delta_{1}-\frac{1}{2}}\right)$,
where $\delta_{2}>\delta_{1}$. Thus if, $\delta_{2}-\delta_{1}>\frac{1}{2}$,
then even in the worst case, convergence rate of the proposed algorithm
is better than that of the EVD algorithm.

%% file: includes/Chapter-7.tex
\chapter{\label{chap:Illustration-and-Simulations}Illustration and Simulations:}

Numerical illustrations of the proposed algorithm are presented in
this Chapter. Here, four numerical examples are considered. As the
original motivation for this work comes from (\cite{lam2011estimation}),
a simulation example presented therein is considered first. The results
obtained by using the proposed algorithm is compared with the EVD
algorithm proposed in (\cite{lam2011estimation}). The second problem
considered here is the simulation of \emph{Example \ref{example 7.1}
 }presented in \emph{Chapter \ref{chap:Comparison-with-the}. }The
results of the performance of the proposed algorithm, the EVD and
the PCA (\cite{bai2002}) are compared. In the third example, a new
algorithm is proposed for the separation of multiple audio signals
from noisy mixtures. The new algorithm is a conjunction of the proposed
RRQR algorithm and the Independent Component Analysis (ICA) technique.
The presented algorithm has better noise characteristics when compared
to the conventional Noisy ICA technique (\cite{hyvarinen2004independent}).
Finally, the proposed algorithm is applied to model the US stock market
data in comparison with the PCA based model.

\section{Simulation 1\label{subsec:Simulation-1}}

This simulation is used to illustrate the rate of convergence of $\hat{Q}$
to $Q$ for the RRQR algorithm in comparison with the EVD method.
Consider (\ref{eq:1.2}) with a factor loading matrix 
\begin{equation}
H=\left[\begin{array}{cccc}
2\cos\left(\frac{2\pi}{K}\right) & 2\cos\left(\frac{4\pi}{K}\right) & \ldots & 2\cos\left(2\pi\right)\end{array}\right]^{\top},\label{Sim1}
\end{equation}
 the factors $x_{n}$ being generated by 
\begin{equation}
x_{n}=0.9x_{n-1}+\eta_{n},\label{Sim2}
\end{equation}
 where process noise $\eta_{n}$ and measurement noise $\epsilon_{n}$
are Gaussian noises with mean zero and variance $4$. Data sets were
generated using (\ref{eq:1.2}), (\ref{Sim1}) and (\ref{Sim2}) for
$K=20,180,400$ and $1000$ and data length $N=200$ and $500$.

Matrix $\tilde{M}$, (\ref{eq:M_perturbed-1}), is generated for each
pair $(K,N)$ with $m=5$. The $\mathbb{L}_{2}$ norm of the error
$\parallel\hat{Q}_{i}-Q\parallel,i=1,2$, is tabulated in Table~\ref{tab:Simulation1}.
Here $\hat{Q}_{1}$ and $\hat{Q}_{2}$ are the estimates obtained
using EVD and RRQR methods respectively. \emph{Figure} (\ref{fig:Convergence})
shows the convergence of $\left\Vert \hat{Q}-Q\right\Vert _{F}$ for
RRQR and EVD methods with $K=60$ while varying $N$ for two different
deltas, $\delta=0$ given by (\ref{Sim1}), and $\delta=0.5$ given
by 
\begin{eqnarray*}
H_{1} & = & \left[\begin{array}{cc}
H(1:\frac{K}{2}) & Z\end{array}\right]^{T}.
\end{eqnarray*}
where $Z$ is a vector of $\frac{K}{2}$ zeros. It was noted that
there was not much change in the observation when varying $K$ hence
was neglected.

Assuming a fixed $N=100$ and $K=180$ in the above model, using 100
monte carlo trials, the mean and standard deviation of the ratio,
(\ref{eq:m1a}), 
\[
r_{i}=\frac{\lambda_{i}}{\lambda_{i+1}}
\]
used in detecting the model order in the EVD method is plotted in
figure \ref{fig:EVD_modelorder}. While the mean and standard deviation
of the ratio, (\ref{eq:model_order}), 
\[
r_{i}=\frac{\gamma_{i}+\epsilon}{\gamma_{i+1}+\epsilon},
\]
proposed for determining the model order using RRQR decomposition
is plotted in figure \ref{fig:RRQR_model_order}. From figure \ref{fig:EVD_modelorder}
and \ref{fig:RRQR_model_order}, it can be seen that the RRQR based
model order detection scheme works as good as EVD even when $K\asymp N$.
Also, it can be seen from the error bar that the proposed RRQR method
has lesser standard deviation and hence more consistant when compared
to the EVD estimate.

\begin{center}
\begin{table*}[t]
\caption{\label{tab:Simulation1}Summary of the results of Simulation $1$,all
values multiplied by 1000.}

\centering{}%
\begin{tabular}{|c|c|c||c|c|}
\hline 
$K$  & $\parallel\hat{Q}_{1}-Q\parallel$  & $\parallel\hat{Q}_{2}-Q\parallel$  & $\parallel\hat{Q}_{1}-Q\parallel$  & $\parallel\hat{Q}_{2}-Q\parallel$ \tabularnewline
\hline 
 & \multicolumn{2}{c||}{$N=200$} & \multicolumn{2}{c|}{$N=500$}\tabularnewline
\hline 
\hline 
20  & 15.9 & 11.8  & 14.9 & 10.9 \tabularnewline
\hline 
180  & 17.4 & 12.3  & 15.9 & 11.3 \tabularnewline
\hline 
400  & 18 & 12.1  & 16.5 & 11.2 \tabularnewline
\hline 
1000  & 18 & 12.3 & 11.6  & 11.4\tabularnewline
\hline 
\end{tabular}
\end{table*}
\par\end{center}

\begin{center}
\begin{figure*}
\selectlanguage{british}%
\begin{centering}
\includegraphics[width=14cm,height=8cm]{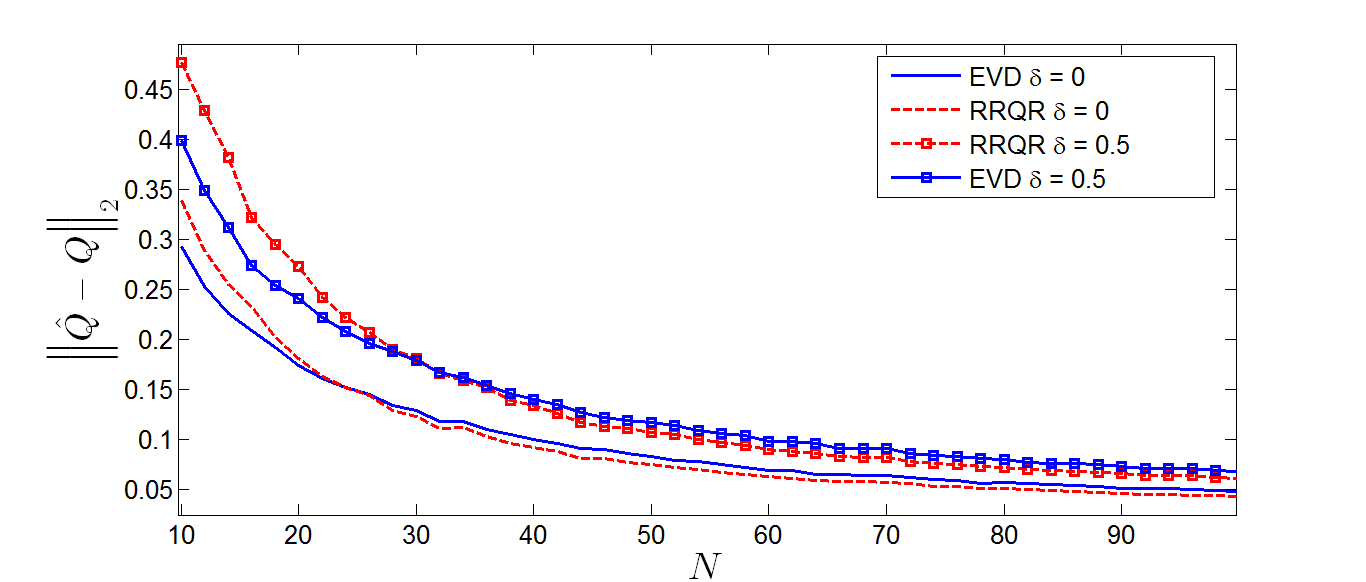}
\par\end{centering}
\selectlanguage{english}%
\centering{}\caption{\label{fig:Convergence}Convergence of $\left\Vert \hat{Q}-Q\right\Vert _{F}$
for EVD and RRQR for fixed $K=60$}
\end{figure*}
\par\end{center}

\begin{center}
\begin{figure*}
\selectlanguage{british}%
\begin{centering}
\includegraphics[width=14cm,height=8cm]{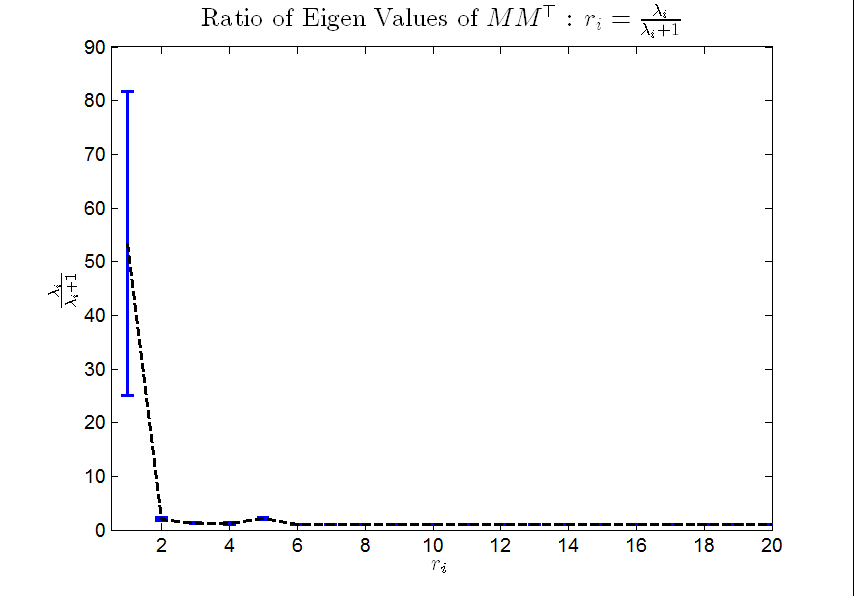}
\par\end{centering}
\selectlanguage{english}%
\centering{}\caption{\label{fig:EVD_modelorder}Plot of ratio of eigen values of $MM^{\top}$,
$r_{i}$ (\ref{eq:m1a}) , for fixed $K=180;\,N=100$.}
\end{figure*}
\par\end{center}

\begin{center}
\begin{figure*}
\selectlanguage{british}%
\begin{centering}
\includegraphics[width=14cm,height=8cm]{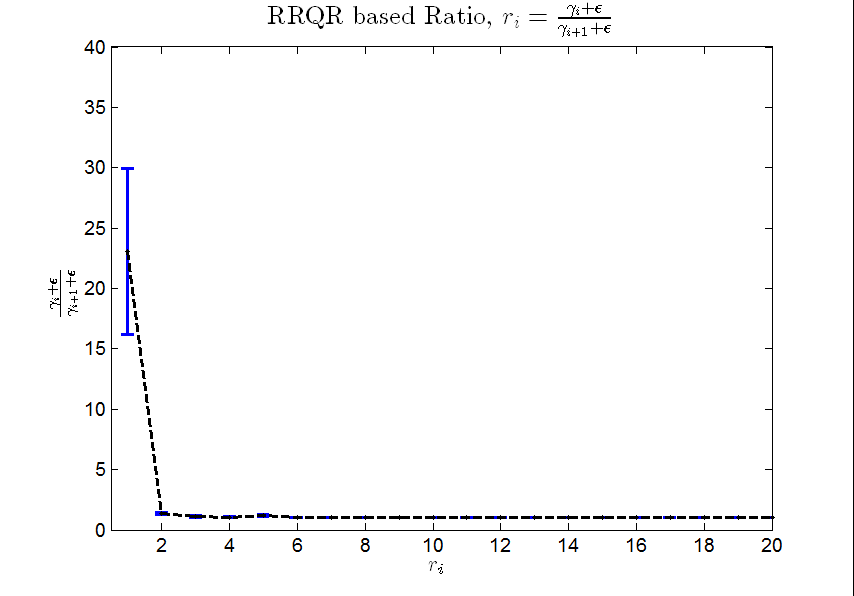}
\par\end{centering}
\selectlanguage{english}%
\centering{}\caption{\label{fig:RRQR_model_order}Plot of ratio of diagonal values of $R$
in $M\Pi=QR$, $r_{i}$, (\ref{eq:model_order}), where $\Pi$ is
got using RRQR Method for fixed $K=180;\,N=100$.}
\end{figure*}
\par\end{center}

\section{Simulation 2\label{subsec:Simulation-2}}

Consider the data model given by (\ref{eq:1.2}) with the number of
factors, $p=2$ as in \emph{Example 1 }and the factors $x_{n}$ taken
as in (\ref{eq:x_n_factors-1}) with $\alpha_{1}=0.5$ and $\alpha_{2}=0.5$.
The factor strengths are taken as $\delta_{1}=0$ and $\delta_{2}=0.5$
with the elements of the columns of $H$ given by $h_{1}\sim U(-4,\,4)$
and $h_{2}\left(1:\frac{K}{2}\right)\sim U(-4,\,4)$ and $h_{2}\left(\frac{K}{2}+1:K\right)=0$.
Two scenarios are considered for the random noise $\varepsilon_{n}$
in (\ref{eq:1.2}). Firstly, $\varepsilon_{n}$ is taken to be generated
by $N(0,\,I)$, $I\in\mathbb{R}^{K\times K}$ is an identity matrix.
In the second scenario, $\varepsilon_{n}$ is taken as $N(0,0.1\Sigma)$,
where the elements of $\Sigma\in\mathbb{R}^{K\times K}$ are given
by 

\begin{equation}
\sigma_{i,j}=\frac{1}{2}\left(\left|i\right|^{2w}-\left|i-j\right|^{2w}+\left|j\right|^{2w}\right),\label{errcorr}
\end{equation}
where $w$, the Hurst parameter, is taken as $0.6$.

The data sets were simulated for all pairs of $(K,\,N)$ where $K$
takes values of $100,\,200,\,400$ and $N$ takes, $100,\,200$. The
simulated data set was run against the QR-PC algorithm, the EVD algorithm
and also the PCA algorithm proposed in (\cite{bai2002}). In (\cite{bai2002}),
the model order is the $p$ that minimizes the Information Content,
\[
IC_{p}=\ln\left(V(p,\hat{f}_{n})\right)+p\left(\frac{K+N}{KN}\right)\ln\left(\frac{KN}{K+N}\right)
\]
where 
\[
V(p,\hat{f}_{n})=\min_{\hat{Q}}\frac{1}{KN}\sum_{n=1}^{N}\left\Vert y_{n}-\hat{Q}^{(p)}\hat{f}_{n}^{(p)}\right\Vert _{2}^{2}
\]
Here, $\hat{Q}^{(p)}$ and $\hat{f}_{n}^{(p)}$ are obtained using
PCA by assuming model order as $p$. Thus, the estimate of model order
is 
\[
\hat{p}_{N,K}=\arg\min_{p}\{IC_{p}\}.
\]
The Root Mean Square Error (RMSE) and Forecast Error (FE) were observed
for all pairs $\left(K,\,N\right)$ and the mean over $100$ montecarlo
runs are tabulated. The RMSE error is 
\begin{equation}
\epsilon_{rmse}=\left(\frac{\sum_{n=1}^{N}\left\Vert \hat{Q}_{N,K}\hat{f}_{n}-Hx_{n}\right\Vert _{2}}{KN}\right)^{0.5}\label{eq:error_rmse}
\end{equation}
 and the mean FE given by 
\begin{equation}
\epsilon_{fe}=\frac{1}{N}\sum_{n=1}^{N}\left(K^{-0.5}\left\Vert \hat{y}_{n}^{(1)}-y_{n}\right\Vert _{2}\right),\label{eq:error_fe}
\end{equation}
where $\hat{y}_{n}^{(1)}$ is the one step ahead forecast of $y_{n}$
got by using $\hat{y}_{n}^{(1)}=\hat{Q}_{N,K}\hat{f}_{n}^{(1)}$ and
$\hat{f}_{n}^{(1)}$ is the one step ahead forecast of the factors
got by fitting AR model to factors using Yule-Walker algorithm. The
results are summarized in \noun{Table \ref{tab:Simulation2}. }

It could be noted that while the PCA method is better at estimating
the RMSE error, the EVD and proposed methods perform better when FE
is concerned. Also, the proposed method performs better than the EVD
in certain cases. It could be noted that when $\epsilon\sim N(0,\,I)$,
all the three algorithms were able to detect the correct model order.
However when $\epsilon\sim N\left(0,\,\Sigma\right)$, only the EVD
and RRQR methods were able to detect the model order correctly, while
PCA overestimated the model order by a huge difference. %
{} Finally it could also be noted that the number of lags, $m$, considered
is increased, the performance of the EVD and RRQR methods increased.

\begin{center}
\begin{table}
\caption{\label{tab:Simulation2}Summary of Observations from Simulation 2.
The RMSE and FE values are multiplied by 1000.}

\begin{singlespace}
\noindent \centering{}%
\begin{tabular}{|c|c|c|c|c|c|c|c|c|c|c|c|c|}
\hline 
\multirow{3}{*}{} & Noise & \multicolumn{6}{c|}{I. $\varepsilon_{n}\sim N(0,I)$} & \multicolumn{5}{c|}{II. $\varepsilon_{n}\sim N(0,0.1\Sigma)$}\tabularnewline
\cline{2-13} 
 & $K$ & \multicolumn{1}{c|}{$100$} & $100$ & $100$ & $200$ & $200$ & $200$ & $100$ & $100$ & $100$ & $200$ & $200$\tabularnewline
\cline{2-13} 
 & $N$ & $100$ & $200$ & $300$ & $100$ & $200$ & $300$ & $100$ & $200$ & $300$ & $200$ & $300$\tabularnewline
\hline 
\hline 
\multirow{3}{*}{$\begin{array}{c}
\mbox{RRQR,}\\
m=2
\end{array}$ } & $\hat{p}$ & $2$ & $2$ & $2$ & $2$ & $2$ & $2$ & $2$ & $2$ & $2$ & $2$ & $2$\tabularnewline
\cline{2-13} 
 & $\epsilon_{fe}$ & $1474$ & $1486$ & $1460$ & $1465$ & $1481$ & $1480$ & $1244$ & $1230$ & $1226$ & $1250$ & $1238$\tabularnewline
\cline{2-13} 
 & $\epsilon_{rmse}$ & $52.8$ & $23.0$ & $16.2$ & $37.9$ & $22.7$ & $15.4$ & $138$ & $80$ & $16$ & $103$ & $26.0$\tabularnewline
\hline 
\hline 
\multirow{3}{*}{$\begin{array}{c}
\mbox{RRQR,}\\
m=9
\end{array}$} & $\hat{p}$ & $2$ & $2$ & $2$ & $2$ & $2$ & $2$ & $2$ & $2$ & $2$ & $2$ & $2$\tabularnewline
\cline{2-13} 
 & $\epsilon_{fe}$ & $1479$ & $1479$ & $1482$ & $1481$ & $1489$ & $1482$ & $1229$ & $1229$ & $1238$ & $1227$ & $1231$\tabularnewline
\cline{2-13} 
 & $\epsilon_{rmse}$ & $45.2$ & $23.5$ & $15.9$ & $31.1$ & $22.3$ & $15.4$ & $214$ & $126$ & $53.2$ & $15.6$ & $45.3$\tabularnewline
\hline 
\hline 
\multirow{3}{*}{$\begin{array}{c}
\mbox{EVD,}\\
m=2
\end{array}$} & $\hat{p}$ & $2$ & $2$ & $2$ & $2$ & $2$ & $2$ & $2$ & $2$ & $2$ & $2$ & $2$\tabularnewline
\cline{2-13} 
 & $\epsilon_{fe}$ & $1475$ & $1488$ & $1464$ & $1467$ & $1484$ & $1483$ & $1241$ & $1229$ & $1226$ & $1250$ & $1239$\tabularnewline
\cline{2-13} 
 & $\epsilon_{rmse}$ & $40.1$ & $20.7$ & $14.5$ & $26.5$ & $19.7$ & $13.3$ & $113$ & $79.5$ & $14.1$ & $109$ & $54.0$\tabularnewline
\hline 
\hline 
\multirow{3}{*}{$\begin{array}{c}
\mbox{EVD,}\\
m=2
\end{array}$} & $\hat{p}$ & $2$ & $2$ & $2$ & $2$ & $2$ & $2$ & $2$ & $2$ & $2$ & $2$ & $2$\tabularnewline
\cline{2-13} 
 & $\epsilon_{fe}$ & $1477$ & $1480$ & $1487$ & $1476$ & $1491$ & $1484$ & $1226$ & $1230$ & $1238$ & $1229$ & $1233$\tabularnewline
\cline{2-13} 
 & $\epsilon_{rmse}$ & $34.5$ & $19.6$ & $14.1$ & $22.9$ & $18.2$ & $13.0$ & $189$ & $105$ & $27.5$ & $186$ & $94.8$\tabularnewline
\hline 
\hline 
\multirow{3}{*}{PCA} & $\hat{p}$ & $2$ & $2$ & $2$ & $2$ & $2$ & $2$ & $18$ & $22$ & $24$ & $29$ & $33$\tabularnewline
\cline{2-13} 
 & $\epsilon_{fe}$ & $1470$ & $1484$ & $1461$ & $1455$ & $1477$ & $1479$ & $1243$ & $1229$ & $1238$ & $1249$ & $1238$\tabularnewline
\cline{2-13} 
 & $\epsilon_{rmse}$ & $19.9$ & $12.2$ & $9.4$ & $11.2$ & $10$ & $7.4$ & $107.2$ & $75.8$ & $61.9$ & $114.7$ & $93.6$\tabularnewline
\hline 
\end{tabular}
\end{singlespace}
\end{table}
\par\end{center}

\section{Simulation 3 (ICA Pre-processing)\label{subsec:Simulation-3}}

The Independent Component Analysis, (\cite{comon1994independent,hyvarinen2000independent}),
is a well known technique for estimating the factors. The basic ICA
technique is used to extract $p$ sources from $p$ mixtures. When
the mixtures are noise corrupted and the number of such mixtures being
larger than the total number of sources, the Noisy ICA algorithm is
used. The Noisy ICA (\cite{hyvarinen2004independent}) algorithm considers
a model of the form, 
\begin{equation}
y_{n}=Ws_{n}+\epsilon_{n}\label{eq:ICA_model}
\end{equation}
where, 
\begin{enumerate}
\item The factors, $s_{n}$, are independent of each other.
\item Not more than one factor is Gaussian.
\item The factors $s_{n}$ are independent of $\epsilon_{n}$.
\item The noise $\epsilon_{n}$ is Gaussian with co-variance $\sigma^{2}I$,
where $I$ is an Identity matrix.
\end{enumerate}
Given $\left\{ y_{n}\right\} _{n=1}^{N}$, the Noisy ICA Algorithm
estimates $W$ and $s_{n}$ uniquely. 

Here, an algorithm is developed by combining the proposed RRQR technique
with the basic ICA algorithm. This helps in relaxing the conditions
3,4 presented above. The algorithm is summarized below:
\begin{enumerate}
\item Obtain $Q$ and $f_{n}$ from the mixtures, $y_{n}$, as in data model
(\ref{eq:1.3}) using the RRQR based technique presented in \emph{Section
5}.
\item Determine $s_{n},\,A$ from $f_{n}$ such that, 
\begin{eqnarray*}
s_{n} & \triangleq & A^{-1}f_{n}\\
 & \triangleq & \left[\begin{array}{ccc}
s_{n}^{(1)} & \ldots & s_{n}^{\left(p\right)}\end{array}\right]^{\top}
\end{eqnarray*}
where $\tilde{A}$ is a matrix obtained by using the basic ICA, such
that the components of $\tilde{s}_{n}$ are independent. 
\end{enumerate}
Thus, 
\[
y_{n}=QAA^{-1}f_{n}+\varepsilon_{n}
\]
Assuming $W=QA$, 
\begin{equation}
y_{n}=Ws_{n}+\varepsilon_{n}\label{eq:ICA_with_preprocessing}
\end{equation}

Thus, $W$ and $s_{n}$ could be extracted as in (\ref{eq:ICA_with_preprocessing})
under the following assumptions: 
\begin{enumerate}
\item The components of $s_{n}$ are independent with each other.
\item Not more than one factor has Gaussian distribution.
\item All $s_{n}$ should have correlations with the past. \emph{i.e., }$\mathcal{E}\left\{ s_{n+l}s_{n}\right\} \neq0$
for at least one $l=1,\ldots,m$.
\item The random noise $\varepsilon_{n}$ can be correlated with the factors
$s_{n}$, \emph{i.e., }$\mathcal{E}\left\{ s_{n+l}\varepsilon_{n}\right\} $
need not necessarily be zero for all $l\geq0$.
\item The auto co-variance of random noise, $\varepsilon_{n}$, can be $\Sigma$,
where $\Sigma$ is any positive semi definite matrix. 
\end{enumerate}
The graphical representation of the proposed algorithm is depicted
in \emph{Figure (\ref{fig:ICA-implementation-Model}).}

\begin{center}
\begin{figure*}[t]
\begin{centering}
\includegraphics[clip,width=14cm,height=8cm]{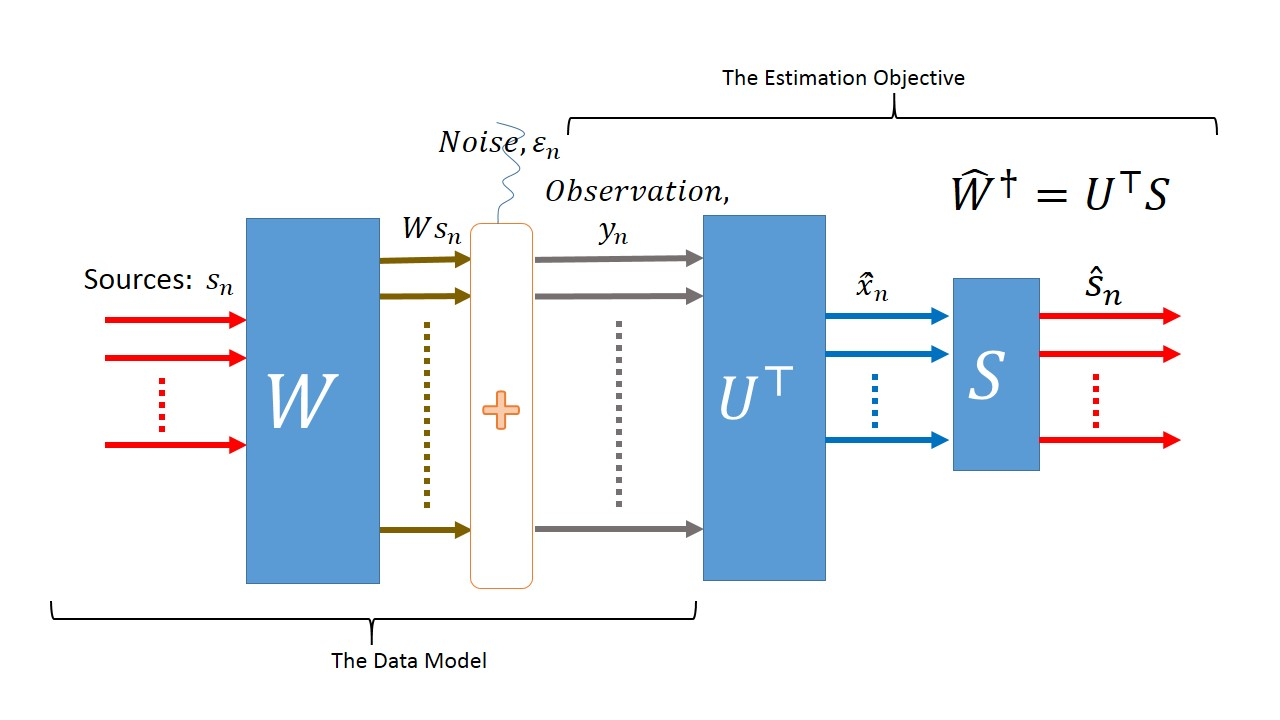}
\par\end{centering}
\caption{\label{fig:ICA-implementation-Model}ICA implementation Model}
\end{figure*}
\par\end{center}

In this thesis, an approximation of the popular \emph{cocktail-party}
problem is simulated. The cocktail party problem is to separate multiple
sound sources using recordings from multiple microphones. Let the
number of sound sources be two, given by $s=\left[\begin{array}{cc}
s^{(1)} & s^{(2)}\end{array}\right]^{\top}$ with $N$ samples (different pairs of speech and music were considered
as sources). The sources are mixed by a random matrix $W\in\mathbb{R}^{K\times2}$
where each element of $W$ is $\left|N(0,1)\right|$. Finally measurement
noise $\varepsilon_{n}$ is added. Two scenarios are considered where
in the first case, each element of $\varepsilon_{n}$, is taken to
be $N(0,\,0.1)$ and in the second case, $\varepsilon_{n}$, is taken
to be $N(0,\,\Sigma)$. Here $\Sigma\in\mathbb{R}^{K\times K}$ are
as described in (\ref{errcorr}). Note, this is an approximation of
the cocktail party problem as it does not include reverberations from
the room or time delay between different recordings or multiple noise
sources which might exceed the number of microphones used. For a more
detailed analysis of the cocktail party problem refer, (\cite{hoffmann2011recognition}).
The performance of the algorithm is evaluated using PESQ score (\cite{rix2001perceptual}).
The PESQ is a tool that compares the original signal $s_{n}$ with
the recovered signal $\hat{s}_{n}$ and it gives a best score of $4.5$
when both match exactly. The ICA algorithm used is from (\cite{hyvarinen2005fastica}).
Since the speech and music sources have a super Gaussian distribution,
The kurtosis is assumed positive and $tanh$ non-linearity is used
as a measure of non gaussianity in the implementation of \emph{fast
ICA. }

It could be noted from the observations in \emph{Table \ref{tab:ICA}}
that, the RRQR based method performs better in almost all cases. It
could also be noted that in the first scenario of noise, where the
same magnitude of noise was included in each observation, as the number
of observed mixtures ($K$) increased, the separation was more accurate
implying that information from the increased number of mixtures was
used to separate out the components. The second scenario is quite
different, as the number of observed mixtures increased, the noise
variance also increased \emph{i.e., }the noise variance was proportional
to the number of observations as can be seen from (\ref{errcorr}).
Thus, with the increase in the number of observed mixtures, the increase
in the noise magnitude offset the increase in the performance thus,
with increase in $K$, decrease in the performance of the algorithm
is seen. It was also noted that the PCA based method of (\cite{bai2002})
was unable to determine the number of sources under the second scenario
and its estimate for the number of sources were between 20-30. Thus,
in the simulation, the number of sources were assumed to be $2$ in
this case while using PCA. The EVD and the proposed method did not
suffer from this problem and were able to estimate the number of sources
as $2$ quite consistently. It is also interesting to note that while
speech-music separation was the best, speech-speech separation was
the worst. Thus, the separation characteristic also depends on the
type of source to be separated. 

\begin{center}
\begin{table*}[t]
\begin{centering}
\begin{tabular}{|c|c|c|c||c|c|c||c|c|c|}
\hline 
Source Type & \multicolumn{3}{c||}{speech-speech} & \multicolumn{3}{c||}{speech-music} & \multicolumn{3}{c}{music-music}\tabularnewline
\hline 
$K$ - Mixtures & 10 & 40 & 100 & 10 & 40 & 100 & 10 & 40 & 100\tabularnewline
\hline 
\hline 
 & \multicolumn{9}{c|}{(Average PESQ Score ) with $\varepsilon_{i,\,j}\sim N(0,0.1)$}\tabularnewline
\hline 
RRQR & 2.282 & 2.813 & 3.1300 & 2.931 & 3.667 & 4.092 & 2.868 & 3.534 & 3.823\tabularnewline
\hline 
EVD & 2.275 & 2.780 & 3.1277 & 2.932 & 3.666 & 3.990 & 2.868 & 3.534 & 3.823\tabularnewline
\hline 
PCA  & 2.285 & 2.812 & 3.1312 & 2.931 & 3.665 & 4.092 & 2.868 & 3.534 & 3.823\tabularnewline
\hline 
 & \multicolumn{9}{c|}{(Average PESQ Score ) with $\varepsilon_{n}\sim N(0,\Sigma)$}\tabularnewline
\hline 
RRQR & 3.515 & 3.088 & 2.670 & 4.261 & 4.017 & 3.695 & 4.221 & 4.006 & 2.835\tabularnewline
\hline 
EVD & 3.507 & 3.079 & 2.665 & 4.264 & 4.019 & 3.695 & 4.221 & 4.007 & 3.667\tabularnewline
\hline 
PCA  & 4.064 & 4.086 & 3.956 & 3.886 & 3.581 & 3.386 & 3.887 & 3.355 & 3.667\tabularnewline
\hline 
\end{tabular}
\par\end{centering}
\caption{\label{tab:ICA}Errors Associated when used as pre-processing with
ICA}
\end{table*}
\par\end{center}

\section{Application on a finance portfolio\label{subsec:finance-app}}

The proposed RRQR method, EVD and PCA were applied to model the returns
of the S\&P 500 index constituents. The S\&P 500 is a diversified
index that represents large companies in the US. It is often taken
as a representative index of the US market as a whole. We consider
the 500 stocks that are members of the S\&P Index as of July 31, 2013.
We collect the data for these 500 stocks from January 1, 2010 - July
31, 2013. The prices are adjusted for splits and dividends. The daily
logarithmic returns are calculated from the adjusted prices. If $P(n)$
and $P(n-1)$ are the prices at time at $n$ and $n-1$, then the
logarithmic return is: 
\begin{eqnarray*}
y_{n} & = & \log\left(P(n)\right)-\log\left(P(n-1)\right).
\end{eqnarray*}
Sixteen stocks did not have data for all days and hence were omitted
from the study. The following steps are applied to develop the factor
models that are used to forecast the returns. 
\begin{enumerate}
\item Let the total number of days be $N$ and the number of stocks be $K$.
\item The mean across time is set to zero for all stocks, \emph{i.e.,} 
\[
\tilde{y}_{n}=y_{n}-\bar{y},\:n=1,2,\ldots,N.
\]
where $\bar{y}=\frac{1}{N}\sum_{n=1}^{N}y_{n}$ is the mean.
\item Now, the demeaned data is modeled as $\tilde{y}_{n}=Qf_{n}+\varepsilon_{n}$
where the factor loading matrix, $Q$ and factors $f_{n}$ are estimated
using data from $n=1,\ldots,499$. The following three algorithms
were used for estimating $Q$ and $f_{n}$ :

\begin{enumerate}
\item The EVD algorithm
\item The RRQR algorithm
\item The PCA based algorithm
\end{enumerate}
\item An AR model of order 10 was fit for the factors $f_{n}$ using Yule
- Walker equations. 
\item One step ahead out sample prediction $\tilde{y}_{n+1}=Qf_{n+1}$,
are computed for days, $n=500,\ldots509$. Here $f_{n+1}$ is the
one step ahead prediction got using the AR model. Wherein, for estimation
of $n=501,\,502,\,\ldots$, the true value of $n=500,\,501,\,\ldots$
is used.
\item $Q$ and $f_{n}$ are re-estimated using data from $n=10,\ldots,\,509$
and an AR model is estimated. One step predictions are made for $n=510-519.$
This process is repeated till the one step ahead predictions are obtained
for the last $400$ days.
\end{enumerate}
Using the above methodology, the following observations were taken
\begin{enumerate}
\item Model order $(q)$ is the mean of the model order estimates over every
update of $Q$ and $f_{n}$. 
\item RMSE error $\left(e_{rmse}\right)$ is the mean of the $e_{rmse}$
estimates over every update of $Q$ and $f_{n}$. Note, the $e_{rmse}$
error for each 500 day block estimates of $Q$ and $f_{n}$ is computed
using (\ref{eq:error_rmse}).
\item Forecast error $(e_{fe})$ as in (\ref{eq:error_fe}) for last 400
days.
\end{enumerate}
The results observed are tabulated in \emph{Table \ref{tab:Finance_data}}

\begin{center}
\begin{table}
\caption{\label{tab:Finance_data}Results from finance data application}

\centering{}%
\begin{tabular}{|c|c|c|c|}
\hline 
 & EVD-Method & Proposed Method & PCA\tabularnewline
\hline 
\hline 
$q$ & 1 & 1 & 6\tabularnewline
\hline 
$e_{rmse}$ & 0.0144 & 0.0146 & 0.0129\tabularnewline
\hline 
$e_{fe}$ & 0.0163 & 0.0163 & 0.0164\tabularnewline
\hline 
\end{tabular}
\end{table}
\par\end{center}

It is a known fact that $S\&$P index is in such a way that the stocks
depend only on a single factor, as opposed to the six factors estimated
by PCA. Here both EVD and RRQR have same prediction errors while PCA
has a higher prediction error.

%% file: includes/Chapter-8.tex
\chapter{\label{chap:Conclusion-and-Discussions:}Conclusion and Discussions}

In this thesis, factor modelling of an observed $K$-dimensional multivariate
series, $y_{n}$, of length $N$ $(n=1,2,\ldots,N)$ was considered.
The goal was to estimate a low $p$-dimensional series called factors,
$f_{n}$ ($p<K$) and a constant matrix $Q$ such that $y_{n}=Qf_{n}+\varepsilon_{n}$,
where $\varepsilon_{n}$ was random noise. The random noise was assumed
to have covariance, $\Sigma_{\varepsilon\varepsilon}(l)=\Sigma\delta(l)$,
where $\Sigma$ is a positive definite matrix. It was also assumed
that the factors can be correlated with past noise and not with future
noise \emph{i.e., $\Sigma_{f\varepsilon}(l)\neq0$ for $l\geq0$ }and
\emph{$\Sigma_{f\varepsilon}(l)=0$ for $l<0$. }In principle a ``fat
matrix'', $M$, can be formed by augmenting the autocovariance matrices
$\left\{ \Sigma_{yy}(l)\right\} _{l=1}^{m}$. This matrix $M$ would
be of rank $p$ and span the same column space as $Q$ \emph{i.e.,
$\mathcal{R}\left(M\right)=\mathcal{R}\left(Q\right)$}. Hence, determining
$p$ and $Q$ from SVD/QR decompositions of $M$ is straight forward.
In practice, only the finite time estimates of the autocovariances
are observed and hence only a perturbed version $\tilde{M}$ of $M$
is observable. This matrix, $\tilde{M}$, would be full row rank with
probability one. In this paper, the model order $p$ was determined
by equating it to the ``numerical rank'' of the matrix $\tilde{M}$.
The numerical rank of $\tilde{M}$ was determined by maximizing a
ratio, (\ref{eq:model_order}), formed using the diagonal values of
the $R$ matrix of the \noun{Hybrid-III} RRQR decomposition, \cite{chandrasekaran1994rank}.
With the model order $p$ determined, an estimate of $Q$ was obtained
using \noun{Hybrid-I} RRQR decomposition of $\tilde{M}$. The factors,
$f_{n}$, were obtained by projecting the observations $y_{n}$ on
the space spanned by the estimate of $Q$.

The \noun{Hybrid-III} algorithm ensures that 
\begin{equation}
\sigma_{p}\left(\tilde{M}\right)\sqrt{\left(p+1\right)\left(K-p\right)}\geq\sigma_{p}\left(\tilde{M}_{p}\right)\geq\frac{\sigma_{p}\left(\tilde{M}\right)}{\sqrt{p\left(K-p+1\right)}},\label{eq:Hybrid-III bound}
\end{equation}
were $\tilde{M}_{p}$ denotes the first $p$ columns selected from
$\tilde{M}\Pi$ with $\Pi$ being the permutation matrix generated
by the \noun{Hybrid-III} algorithm. \noun{Hybrid-I} and \noun{Hybrid-I}I
algorithms satisfy the upper and lower bounds, respectively, of (\ref{eq:Hybrid-III bound}).
The above bounds were crucial for proving the asymptotic convergence
of the model order. In the simulations, QR-CP as well the other Hybrid
algorithms were found to to give similar results. Nevertheless, convergence
results could be proved only for Hybrid algorithms. In particular,
convergence of the model order scheme could be proved only for the
Hybrid-III algorithm. 

Conventional PCA based factor modelling cannot allow noise to be correlated
with factors and that noise should have covariance of $\Sigma_{\varepsilon\varepsilon}(0)=\sigma^{2}I$.
The proposed algorithm relaxes the above and allows noise to have
any covariance $\Sigma$. Thus it was seen from \emph{Simulations
\ref{subsec:Simulation-2} and \ref{subsec:Simulation-3}} that when
the noise had covariance anything other than $\sigma^{2}I$, the PCA
based model order detection schemes \cite{bai2002} failed to detect
the correct number of factors $p$, the proposed RRQR based scheme
remained robust.

The method of using $\tilde{M}$ to determine the parameters in the
factor model was first approached by \cite{lam2011estimation} wherein
EVD of $\tilde{M}\tilde{M}^{\top}$ was used to obtain estimates for
$p,\,Q$ and $f_{n}$. It was proved that in certain cases, the proposed
RRQR based methodology is better than the EVD based method w.r.t the
asymptotic properties (as both the dimension and duration of observations,
$K$ and $N$, tends to infinity). The numerical \emph{Simulations
\ref{subsec:Simulation-1}, \ref{subsec:Simulation-2} and \ref{subsec:Simulation-3}}
also demonstrate that RRQR betters EVD in certain cases. 

Using \emph{Example 1, }it was shown that the factors might not have
correlation at all lags and hence, in order to acquire all factors,
lags, $l=1,\ldots,m$ needs to be used with $m$ high enough. Conversely,
it could be stated that using $l=a,\ldots,\,b$ where $\min\left(K,N\right)\geq b\geq a\geq1$
to construct $M$ would result in the estimation of factors having
correlation within those particular lags alone. This is useful when,
say, an $a$-step ahead prediction is desired (persistent factors)
or if the undesirable factors have correlation only in lags less than
$a$.

The use of the proposed method of factor modelling as pre-processing
to ICA was demonstrated in \emph{Simulation 3}. Finally, w.r.t the
finance data application, the EVD method and the proposed method perform
similarly and identify only one dominant factor. This is expected
since the S\&P 500 is a well diversified index and the only common
factor among all the constituent stocks is the US market factor. 

This work can be extended to devise a recursive algorithm to estimate
the parameters of the factor model based on RRQR decomposition.

%% file: includes/Appendix.tex
\chapter*{\label{chap:Appendix-1}Appendix 1}

\addcontentsline{toc}{chapter}{Appendix 1}

\section*{Interlacing Property\foreignlanguage{british}{\label{subsec:Interlacing-Property}}}

\addcontentsline{toc}{section}{Concepts from Linear Algebra}

\selectlanguage{british}%
The following lemma relates to the interlacing property of eigen values
for symmetric matrices:
\selectlanguage{english}%
\begin{lem}
\cite{li2012generalized,fan1957imbedding}Let $A$ be real symmetric
matrix of dimension $n$ and $B\in\mathbb{R}^{k\times k}$ be a principal
sub matrix of $A$ with $1\leq k<n$. Let $a_{1}\geq a_{2}\geq...\geq a_{n}$
and $b_{1}\geq b_{2}\geq...\geq b_{k}$ be the respective eigenvalues.
Then\textbf{
\begin{equation}
a_{j}\geq b_{j}\geq a_{n-k+j},j=1,2...k\label{eq:thm9.1}
\end{equation}
}
\end{lem}
\begin{proof}
Refer \cite{li2012generalized} for a proof.
\end{proof}
\begin{thm}
Let $\Sigma\in\mathbb{R}^{m\times n}$, $m<n$, and $\Sigma_{11}^{k\times l}$,
$k<l$, be any sub matrix of $\Sigma$ with \foreignlanguage{british}{$\sigma_{1}\geq\sigma_{2}\geq\ldots\geq\sigma_{m}$
and $\lambda_{1}\geq\lambda_{2}\geq\ldots\geq\lambda_{k}$ being their
corresponding singular values. Then,
\[
\sigma_{j}\geq\lambda_{j}\geq\sigma_{n-k+j},j=1,2...k.
\]
}
\end{thm}
\selectlanguage{british}%
\begin{proof}
First we start by establishing the connection between the eigen values
and singular values so as to connect the result in the previous lemma
to pertain to singular values. The following properties are used: 

1. Singular values of a matrix $\Sigma\in\mathbb{R}^{m\times n}$
, $m<n$, are equal to the square roots of the eigen values of the
matrix $A=\Sigma^{\top}\Sigma$ (or) $\Sigma\Sigma^{\top}$. 

2. Also, let QR decomposition of $\Sigma$ be, $\Sigma=QR$. Then
singular values of $R$ equals that of $\Sigma$, $\sigma_{1}\geq\sigma_{2}\geq\ldots\geq\sigma_{m}$,
since $Q$ is an orthogonal matrix.

Let $\Sigma_{11}\in\mathbb{R}^{k\times l}$, $k<l$, be a principal
sub matrix of $\Sigma$. Using the second property, without loss of
generality, $\Sigma$ can be written as upper triangular, for the
purpose of dealing with singular values, $i.e.,$ 
\[
\Sigma=\begin{array}{cc}
 & l\qquad n-l\\
\begin{array}{c}
k\\
m-k
\end{array} & \left[\begin{array}{cc}
\Sigma_{11} & \Sigma_{12}\\
0 & \Sigma_{22}
\end{array}\right]
\end{array}.
\]
 Also, let $\lambda_{1}\geq\lambda_{2}\geq\ldots\geq\lambda_{k}$
be the singular values of $\Sigma_{11}$. Now, Consider the symmetric
matrix, $A\triangleq\Sigma^{\top}\Sigma$ given by, 
\[
A=\Sigma^{\top}\Sigma=\begin{array}{cc}
 & k\qquad\qquad n-k\\
\begin{array}{c}
k\\
n-k
\end{array} & \left[\begin{array}{cc}
\Sigma_{11}^{\top}\Sigma_{11} & \Sigma_{12}^{\top}\Sigma_{12}\\
\Sigma_{12}^{\top}\Sigma_{12} & \Sigma_{22}^{\top}\Sigma_{22}
\end{array}\right]
\end{array}.
\]
Note, the symmetric matrix $\Sigma_{11}^{\top}\Sigma_{11}$ is a principal
sub matrix of $A$. Now, Lemma 9.1 can be applied assuming $B=\Sigma_{11}^{\top}\Sigma_{11}$.
Thus, the eigen values of $\Sigma^{\top}\Sigma$ and $\Sigma_{11}^{\top}\Sigma_{11}$
are related as in (\ref{eq:thm9.1}). Thus taking the positive square
root of which would correspond to the singular values of $\Sigma$
and $\Sigma_{11}$. Thus, 
\[
\sigma_{j}\geq\lambda_{j}\geq\sigma_{n-k+j},j=1,2...k.
\]
\end{proof}
This above theorem is used in Chapters 4, 5 and 6.
\selectlanguage{english}%

\chapter*{\label{chap:Appendix-2}\foreignlanguage{british}{Appendix 2}}

\addcontentsline{toc}{chapter}{Appendix 2}
\selectlanguage{british}%

\section*{Landau notations and Mixing Properties}

\selectlanguage{english}%
\addcontentsline{toc}{section}{landau notations and Mixing Properties}
\selectlanguage{british}%

\subsection*{Properties of Order Notations (landau notations)}

Here the properties of little $o$, big $O$, little $o_{P}$ and
big $O_{P}$ are summarized. Refer also \cite{lehmann1999elements}.

Properties of little $o$ and big $O$: 
\begin{lyxlist}{00.00.0000}
\item [{1.}] $a_{n}=O(a_{n})$
\item [{2.}] If $a_{n}=o(b_{n})$, then $a_{n}=O(b_{n})$
\item [{3.}] If $a_{n}=O(b_{n})$, then $O(a_{n})+O(b_{n})=O\left(a_{n}\right)$.
\item [{4.}] If $a_{n}=O(b_{n})$, then $o(a_{n})+o(b_{n})=o(a_{n})$
\item [{5.}] Let $c$ be a constant, then $cO(a_{n})=O(a_{n})$ and $co(b_{n})=o(b_{n})$.
\item [{6.}] $O(a_{n})O(b_{n})=O(a_{n}b_{n})$; $o(a_{n})o(b_{n})=o(a_{n}b_{n})$
\end{lyxlist}
The properties of $O_{P}$ and $o_{P}$ are summarized as follows: 
\begin{lyxlist}{00.00.0000}
\item [{1.}] Let $a_{n}$ be any stochastic sequence and $b_{n}$ be a
deterministic sequence then, $o\left(a_{n}\right)=b_{n}$ implies
the stochastic sequence $\frac{b_{n}}{a_{n}}\overset{P}{\rightarrow}0$.
\end{lyxlist}
while almost all the properties pertaining to $O$ and $o$ could
be applicable to $O_{P}$ and $o_{P}$, the following two properties
are not so obvious:
\begin{lyxlist}{00.00.0000}
\item [{2.}] Let $a_{n}$ and $b_{n}$ deterministic sequences, then, $O\left(a_{n}\right)+O_{P}\left(b_{n}\right)=O_{P}\left(a_{n}+b_{n}\right).$
\end{lyxlist}
\selectlanguage{english}%
\begin{proof}
Let,
\[
C_{n}=A_{n}+X_{n}
\]
where $A_{n}$ is a deterministic sequence and $X_{n}$ is a random
sequence. Let the rate of convergence of $A_{n}$and $X_{n}$be as
follows: 
\begin{eqnarray*}
A_{n} & = & O\left(a_{n}\right)\\
X_{n} & = & O_{P}\left(b_{n}\right)
\end{eqnarray*}
where $a_{n}$ and $b_{n}$ are deterministic sequences. By definition
of $O(.)$ and $O_{P}(.)$, 
\begin{eqnarray*}
\left|\frac{A_{n}}{a_{n}}\right| & \leq & M_{1}\,\forall\,n>n_{0}\\
P\left\{ \left|\frac{X_{n}}{b_{n}}\right|\leq M_{2}\right\}  & \geq & 1-\epsilon_{n_{0}}\,\forall\,n>n_{0}
\end{eqnarray*}
Thus, 
\[
C_{n}=O\left(a_{n}\right)+O_{P}\left(b_{n}\right)
\]
\[
\implies P\left(\left|\frac{A_{n}}{a_{n}}\right|+\left|\frac{X_{n}}{b_{n}}\right|\leq M_{1}+M_{2}\right)\geq1-\epsilon_{n_{0}}\:\forall n>n_{0}
\]
\[
\implies C_{n}=O_{P}\left(a_{n}+b_{n}\right)
\]
if $a_{n}>b_{n}$ then, 
\[
C_{n}=O_{P}\left(a_{n}\right).
\]
\end{proof}
\begin{lyxlist}{00.00.0000}
\item [{3.}] \foreignlanguage{british}{Let $a_{n}$ and $b_{n}$ deterministic
sequences, then, $O\left(a_{n}\right)O_{P}\left(b_{n}\right)=O_{P}\left(a_{n}b_{n}\right).$}
\end{lyxlist}
\begin{proof}
Let,
\[
C_{n}=A_{n}X_{n}
\]
where $A_{n}$ is a deterministic sequence and $X_{n}$ is a random
sequence. Let the rate of convergence of $A_{n}$and $X_{n}$be as
follows: 
\begin{eqnarray*}
A_{n} & = & O\left(a_{n}\right)\\
X_{n} & = & O_{P}\left(b_{n}\right)
\end{eqnarray*}
where $a_{n}$ and $b_{n}$ are deterministic sequences. By definition
of $O(.)$ and $O_{P}(.)$, 
\begin{eqnarray*}
\left|\frac{A_{n}}{a_{n}}\right| & \leq & M_{1}\,\forall\,n>n_{0}\\
P\left\{ \left|\frac{X_{n}}{b_{n}}\right|\leq M_{2}\right\}  & \geq & 1-\epsilon_{n_{0}}\,\forall\,n>n_{0}
\end{eqnarray*}
Thus, 
\[
C_{n}=O\left(a_{n}\right)O_{P}\left(b_{n}\right)
\]
\[
\implies P\left(\left|\frac{A_{n}}{a_{n}}\right|\left|\frac{X_{n}}{b_{n}}\right|\leq M_{1}M_{2}\right)\geq1-\epsilon_{n_{0}}\:\forall n>n_{0}
\]
\[
\implies C_{n}=O_{P}\left(a_{n}b_{n}\right)
\]
\end{proof}
\selectlanguage{british}%

\subsection*{Mixing Conditions}

In Chapter \ref{chap:Data-Model}, the stationarity of the time series
was due to the assumptions on the mixing conditions namely, $\alpha\left(.\right)$
and $\psi(.)$. Here an explanation pertaining to the same are summarized.
For more details refer \cite[pp 363 - 367]{billingsley2008probability}
and \cite[sec 2.6]{fan2003nonlinear}.

To express dependency across time for an observed stochastic sequence,
one could use the Moving Average (MA) model. But it is too simplistic
and cannot comply with complicated dependence structures. Hence, mixing
conditions are used. Mixing conditions state that that, across time,
two well separated instances are independent of each other. This enables
the use of law of Large Numbers and the Central Limit Theorem for
a non-independent stochastic series.

Let $X_{1},\,X_{2},\ldots$ be a sequence of random variables. Now,
$\alpha$ mixing is defined as, 
\[
\alpha_{n}=\sup_{A\in\sigma\left(X_{1},\,\ldots,\,X_{k}\right);\,B\in\sigma\left(X_{k+n},\,X_{k+n+1},\,\ldots\right)}\left|P(A\cap B)-P(A)P(B)\right|
\]
where $\sigma\left(X_{1},\,\ldots,\,X_{k}\right)$ be the sigma algebra
generated by $X_{1},\ldots,\,X_{k}$. Also, the $\psi$ mixing condition
is given by 
\[
\psi_{n}=\sup_{A\in\sigma\left(X_{1},\,\ldots,\,X_{k}\right);\,B\in\sigma\left(X_{k+n},\,X_{k+n+1},\,\ldots\right)}\left|1-\frac{P(B|\,A)}{P(B)}\right|.
\]
Now, if $\alpha_{n}\rightarrow0$, then the random variables $X_{i}$
and $X_{i+n}$ are taken to be independent for large $n$ and the
sequence is addressed as to have $\alpha-$mixing. Note, if the mean
and co-variance of $X_{n}$ is independent of $n$, then $X_{n}$
is assumed to be stationary. 

The following theorem is taken from Billingsley, \cite[Thm 27.4]{billingsley2008probability}:
\begin{thm}
Suppose $X_{1},\,X_{2},\,\ldots$ is stationary with $\alpha_{n}=O(n^{-5})$,
$\mathcal{E}\left(X_{n}\right)=0$ and $\mathcal{E}\left(X_{n}^{12}\right)<\infty$.
Let 
\[
S\triangleq\sum_{n=1}^{N}X_{n},
\]
then, 
\[
N^{-1}\mbox{Var}\left(S_{n}\right)\rightarrow\sigma^{2}=\mathcal{E}\left(X_{1}^{2}\right)+2\sum_{k=1}^{\infty}\mathcal{E}\left(X_{1}X_{1+k}\right),
\]
where the series converges absolutely. If $\sigma>0$, then $\frac{S_{n}}{\sigma\sqrt{N}}\sim N\left(0,1\right)$.
\end{thm}
\begin{proof}
Refer \cite[Thm 27.4]{billingsley2008probability} for proof
\end{proof}
This above theorem is important for numerous aspects as it implies
that our $\tilde{M}$, (\ref{eq:M_perturbed-1}), tends to $M$, (\ref{eq:M_def}),
as $N\rightarrow\infty$ for fixed $K$. It also the gives the rate
of growth for a sum of a series satisfying the above conditions to
be $O_{P}\left(\sqrt{N}\right).$ This fact is used in deriving lemma
(\ref{lem:lemma 5}). 

In lemma (\ref{lem:lemma 5}), we obtained the matrix $\tilde{\Sigma}_{x\varepsilon}(l)-\Sigma_{x\varepsilon}(l)$
where $\tilde{\Sigma}_{x\varepsilon}(l)$ is a sample covariance matrix
between $x_{n}$ and $\varepsilon_{n}$ and $\Sigma_{x\varepsilon}(l)$
is their ideal co-variance matrix. Therein, $x_{n}$ and $\varepsilon_{n}$
are given by 
\[
y_{n}=Hx_{n}+\varepsilon_{n}.
\]
where $y_{n}$ is the stochastic sequence, which satisfies the properties
in Theorem 9.2 (Refer Condition \emph{(A5)} of Chapter \ref{chap:Data-Model}).
Since $x_{n}$ and $\varepsilon_{n}$ are also linearly related to
$y_{n}$, they too satisfy those properties ensuring that the every
element in the matrix $x_{n+l}\varepsilon_{n}^{\top}-\mathcal{E}\left(x_{n+l}\varepsilon_{n}^{\top}\right)$
also satisfy the conditions for Theorem 9.2. Since every element in
the matrix $\tilde{\Sigma}_{x\varepsilon}(l)-\Sigma_{x\varepsilon}(l)=\frac{1}{N}\sum_{n=1}^{N}\left(x_{n+l}\varepsilon_{n}^{\top}-\mathcal{E}\left[x_{n+l}\varepsilon_{n}^{\top}\right]\right)$
correspond to the sum of the above sequence from $n=1\ldots N$ divided
by $N$, they all converge at the rate $O_{p}\left(\frac{1}{\sqrt{N}}\right)$. 

Note that $\alpha-mixing$, also known as \emph{strong mixing,} is
weaker than $\psi-mixing$. Hence $\psi-mixing$ implies $\alpha-mixing$.
For more details regarding the property of $\psi$ mixing, refer \cite{fan2003nonlinear}.\selectlanguage{english}%